\documentclass[aps,10pt,prb,twocolumn,superscriptaddress,floatfix,letterpaper,longbibliography]{revtex4-2}

\usepackage[only,mapsfrom]{stmaryrd}
\usepackage[bottom]{footmisc}
\usepackage{amsmath,amstext,amssymb,amsthm,tikz,tikz-cd}

\usetikzlibrary{decorations.pathreplacing}
\usepackage{amsmath}
\usepackage{amsthm}
\usepackage{subcaption}
\usepackage{chngcntr}   

\counterwithout{subfigure}{figure} 
\usepackage{hyperref}
 \hypersetup{
           breaklinks=true,   
           colorlinks=true,   
           pdfusetitle=true,  
        }

\hypersetup{colorlinks,allcolors=blue}
\usepackage{amssymb}
\usepackage{commands}
\usetikzlibrary{calc}
\usepackage{tikz,tikz-cd,tikzit}

\tikzstyle{red dot}=[fill=red, draw=black, shape=circle]
\tikzstyle{blue dot}=[fill=blue, draw=black, shape=circle]
\tikzstyle{dash dot}=[fill=white, draw=black, shape=circle, dashed]
\tikzstyle{greeen dot}=[fill=green, draw=black, shape=circle]
\tikzstyle{yellow dot}=[fill=yellow, draw=black, shape=circle]
\tikzstyle{circle dot}=[fill=white, draw=black, shape=circle]
\tikzstyle{black dot}=[fill=black, draw=black, shape=circle]

\tikzstyle{dashline}=[-, dashed]
\tikzstyle{redline}=[-, draw=red]
\tikzstyle{blueline}=[-, draw=blue]
\tikzstyle{green line}=[-, draw=green]
\tikzstyle{yellow line}=[-, draw=yellow]
\tikzstyle{purple line}=[-, draw=magenta]

\usetikzlibrary{braids}
\usepackage{graphicx}
\usepackage{bbm}
\usepackage{changes}
\usepackage{physics}
\usepackage{booktabs}
\usepackage{diagbox} 
\usepackage{multirow}
\usepackage{makecell} 

\begin{document}
\title{Condensation Completion and Defects in 2+1D Topological Orders}
\author{Gen Yue}
\email{gyue@link.cuhk.edu.hk}
\affiliation{Department of Physics, The Chinese University of Hong Kong, Central Avenue, Hong Kong, China}
\author{Longye Wang}
\email{wanglongye@link.cuhk.edu.hk}
\affiliation{Department of Physics, The Chinese University of Hong Kong, Central Avenue, Hong Kong, China}
\author{Tian Lan}
\email{tlan@cuhk.edu.hk}
\affiliation{Department of Physics, The Chinese University of Hong Kong, Central Avenue, Hong Kong, China}

\begin{abstract}
    We review the condensation completion of a modular tensor category $\cC$, which yields a fusion 2-category $\Sigma\cC$  of separable algebras, bimodules over algebras and bimodule maps in $\cC$. Physically, $\Sigma\cC$ is the fusion 2-category of codimension-1 defects, codimension-2 defects and instantons in the  $2+1$D topological order $\cC$. We realize the rough-rough wall and $e$-$m$ exchange wall in Toric Code model on the lattice by deforming the Hamiltonian based on the corresponding algebraic data. We apply condensation completion to Toric Code, $3\mathbf{F}$, two-laryer semion and $\mathbb{Z}_4$ topological orders, and explicitly enumerate their $1$d and $0$d defects along with fusion rules. We also mention other applications of condensation completion: alternative interpretations of condensation completion of a braided fusion category; condensation completion of the category of symmetry charges and its correspondence to gapped phases with symmetry; for a topological order $\cC$, one can find all gapped boundaries of the stacking of $\cC$ with its time-reversal conjugate through computing the condensation completion of $\cC$.
\end{abstract}
\maketitle
{\small \setcounter{tocdepth}{2} \tableofcontents }
\section{Introduction}
Understanding different types of phases and phase transitions is a fundamental inquiry in the field of condensed matter physics. For a considerable time, it was widely believed that the Landau-Ginzburg theory of spontaneous symmetry breaking provided a universal framework to describe all kinds of phases and phase transitions. However, in recent decades, many unconventional topologically ordered phases beyond Landau's paradigm have been discovered and have gained substantial attention. There are various approaches to studying these new phases. From a microscopic point of view, people write down many exactly solvable lattice models~\cite{KITAEV20032,Levin_2005,KITAEV20062,Levin_2012,Hu_2013,Heinrich_2016,green2023enriched},  construct the ground state wave functions or partition functions~\cite{Chen_2013,Wang_2018,Wang_2020} and use many numerical methods~\cite{Chen_2011, Schuch_2011, Pollmann_2012, Jiang_2012,Bauer_2013} to study the exotic properties and classifications of these phases. From a macroscopic point of view, people study the observables of the phases in the thermodynamic and long wavelength limit, it turns out that those observables have rather rich algebraic structure and may be described by (higher) category theory~\cite{Kitaev_2012,Lan_2016,Barkeshli_2019, gaiotto2019condensations,Kong_2020,Kong2020DefectsIT,Kong_2022,PhysRevB.105.235143,aasen2022characterizationclassificationfermionicsymmetry, Johnson_Freyd_2022, Lan2023CategoryOS}.

One direct application of category theory in topological phases is that the topological defects in a topological order form a fusion (higher) category. For example, particles in $2+1$D\footnote{We will use D for spacetime dimension and d for spatial dimension.} string-net model with input fusion category $\cC$ forms a modular tensor category $Z(\cC)$~\cite{Levin_2005}, particles in $2+1$D twisted quantum double model with group $G$ and a three cocycle $\omega\in \text{H}^3(G,U(1))$ form a modular tensor category (MTC) $D^{\omega}(G)$\cite{KITAEV20032,Hu_2013}, strings in $3+1$D toric code model form a non-degenerate braided fusion 2-category $Z(2\Rep \mathbb Z_2)$\cite{Kong2020DefectsIT}. In these models, people consider the codimension-2 and higher defects in $n+1$D ($n\geq 2$) spacetime, forming a braided fusion $(n-1)$-category. It's also natural to ask about the codimension-1 defects, the higher codimension defects on codimension-1 defects, and their fusion rules. To answer this question, one can construct these lower codimensional defects directly if there is already a lattice model or a field theory realization of the phase. As an example, the codimension-1 and higher defects in the string-net model have been studied in such manner in the seminal work~\cite{Kitaev_2012}. However, this is quite tedious and not systematical. The category of defects is a universal description of a phase, which doesn't rely on the specific model that realizes the phase. It turns out that for an $n+1$D anomaly-free topological orders, the fusion $n$-category of codimension-$1$ and higher defects in it can be completely determined by the braided fusion $(n-1)$-category of codimension-2 and higher defects through a categorical algorithm called condensation completion~\cite{Johnson_Freyd_2022}. 

If the category of codimension-1 and higher defects can be determined by the category of codimension-2 and higher defects, then why should we care about it? Here we give two reasons
\begin{enumerate}
    \item Although it has been encoded in, seeing the codimension-1 defects and the related fusion rules from the category of codimension-2 and higher defects is not straightforward. It's worthwhile to give the data explicitly.
    \item The fusion $n$-category of codimension-1 defects is a mathematically  complete theory. The process of doing condensation completion of a category is conceptually similar with the Cauchy completion that generates the real numbers $\R$ from the rationals $\Q$. Just as Cauchy completion ``fills the gaps" in the rational numbers by formally adding limit for every Cauchy sequence, thereby yielding the complete space $\R$, the condensation completion of a category $\cC$, denoted as $\mathrm{Kar}(\cC)$ systematically adds objects representing absolute limits (universal properties preserved by functors) that might be missing in $\cC$ itself.
    This completeness of a mathematical theory is important for physics: Newtonian mechanics fundamentally requires the field of real numbers $\R$ to rigorously formulate calculus and differential equations;
    quantum mechanics relies on the completeness of Hilbert spaces to justify the operations of taking limits, differentials and integrals, such as those appearing in the Schrödinger equation and path integrals. In direct analogy, $\mathrm{Kar}(\cC)$ enables well-defined semi-simplicity and functorial constructions that depend on the presence of absolute limits. The fusion (higher) category of codimesnion-1 defects in a topological order is such an mathematically complete category. The codimension-1 defects are mathematically a kind of absolute limits called condensates~\cite{gaiotto2019condensations}. As we'll discuss in \ref{PhysicsCondensation} and Remark.\ref{WhyCondensation}, codimension-1 defects emerge physically by condensing certain particles along a codimension-1 surface.  

\end{enumerate}
In this paper, we will first provide a concise overview of the mathematical framework for condensation completion, and delve into concrete examples: the condensation completion of modular tensor categories with $\mathbb Z_2\times \mathbb Z_2$ and $\mathbb Z_4$ fusion rules, respectively.  In addition, we talk about some other applications of condensation completion. One such application is that the condensation completion of a general braided fusion $(n-1)$-category will give the potentially anomalous categorical data of an $(n-1)$d phase. Another application is that one can classify certain topological phases with symmetry by condensation completion of the category of symmetry charges. Also, for a $2+1$D topological order $\cC$, one can find all gapped boundaries of the stacking of $\cC$ with its time-reversal conjugate through computing the condensation completion of $\cC$.  
\begin{notation} 
Symbols and abbreviations used in this paper are summarized in Table \ref{symbols}.
\begin{table}[htb]
\renewcommand{\arraystretch}{1.0}
\scalebox{0.8}{
\begin{tabular}{ll}
  \hline
  \textbf{Symbol} & \textbf{Definition} \\
  \hline
  $\Box$ & Tensor product of fusion 2-categories\\
  \hline
  $\sim_{M}$ & Morita equivalence\\
  \hline
  $\mathrm{B}\cC$ & Delooping of the category $\cC$\\
  \hline
  $\mathrm{Kar}\cC$ & Karoubi completion of the category $\cC$\\
  \hline
  $\Sigma\cC$ & $\mathrm{Kar}\mathrm{B}\cC$, Karoubi completion of the category $\mathrm{B}\cC$,\\
  &i.e., condensation completion of $\cC$\\
  \hline
  $[G,\psi]$ & Group algebra of the group $G$ over field $\C$,\\
  & with multiplication twisted by $\psi\in \mathrm{H}^2(G, U(1))$\\
  \hline
  \textbf{Abbreviation} &  \textbf{Corresponding topological order} \textbf{of MTCs}\\
  \hline
  \textbf{TC} &  Toric code\\
  \hline
  \textbf{3F} & Three-fermion\\
  \hline
  \textbf{S} & Semion\\
  \hline
  \textbf{SS} & Two layers of semion\\
  \hline
  $\DS$ & Double-semion\\
  \hline
\end{tabular}}
\caption{symbols and abbreviation.}
\label{symbols}
\end{table}
\end{notation}

\section{Condensation Completion}
\label{CondensationCompletion}
In this section, we will review the algebraic construction of condensation completion. The idea of completing a theory by all the ``condensed defects" or ``condensation descendants" has a long history and was studied in a series of works~\cite{Carqueville_2016,Kitaev_2012,KW1405.5858,douglas2018fusion,gaiotto2019condensations} (see also Remark 3.18 in~\cite{KLW+2003.08898}). The modern term \emph{condensation completion}, as introduced by Gaiotto and Johonson-Freyd~\cite{gaiotto2019condensations}, is a generalization of the Karoubi envelope (also known as idempotent completion~\cite{BALMER2001819} or Cauchy completion) from ordinary category to higher category theory. Mathematically, 
condensation completion of a higher category $\cC$ is the process of incorporating all condensation descendants into it, resulting in a complete higher category $\mathrm{Kar}(\cC)$ where all absolute limits exist. This allows for properties such as the semisimplicity of a higher category to be well-defined.

\begin{figure}[htbp]
    \centering
\begin{tikzpicture}[scale=0.9]
\tikzstyle{every node}=[font=\small,scale=0.9]
	\begin{pgfonlayer}{nodelayer}
		\node [style=none] (0) at (-1, 6) {};
		\node [style=none] (1) at (-2, 2) {};
		\node [style=none] (2) at (8, 6) {};
		\node [style=none] (3) at (7, 2) {};
		\node [style=none] (17) at (-0.75, 2.5) {$\mathbbm1$};
		\node [style=none] (18) at (0.25, 5.5) {$\mathbbm1$};
		\node [style=none] (19) at (0.25, 4) {$M$};
		\node [style=none] (26) at (0, 5) {};
		\node [style=none] (28) at (-0.5, 3) {};
		\node [style=circle dot] (29) at (-0.25, 4) {};
		\node [style=none] (30) at (1.25, 2.5) {$A$};
		\node [style=none] (31) at (2.25, 5.5) {$A$};
		\node [style=none] (32) at (2.25, 4) {$N$};
		\node [style=none] (33) at (2, 5) {};
		\node [style=none] (34) at (1.5, 3) {};
		\node [style=greeen dot] (35) at (1.75, 4) {};
		\node [style=none] (36) at (3.25, 2.5) {$B$};
		\node [style=none] (37) at (4.25, 5.5) {$C$};
		\node [style=none] (38) at (4.25, 4) {$P$};
		\node [style=none] (39) at (4, 5) {};
		\node [style=none] (40) at (3.5, 3) {};
		\node [style=yellow dot] (41) at (3.75, 4) {};
		\node [style=none] (42) at (5.75, 4) {$...$};
		\node [style=none] (43) at (6.75, 5.5) {$\Sigma\cC$};
	\end{pgfonlayer}
	\begin{pgfonlayer}{edgelayer}
		\draw (0.center) to (2.center);
		\draw (0.center) to (1.center);
		\draw (1.center) to (3.center);
		\draw (2.center) to (3.center);
		\draw [style=dashline] (26.center) to (29);
		\draw [style=dashline] (29) to (28.center);
		\draw [style=redline] (33.center) to (35);
		\draw [style=redline] (35) to (34.center);
		\draw [style=blueline] (39.center) to (41);
		\draw [style=purple line] (41) to (40.center);
	\end{pgfonlayer}
\end{tikzpicture}
 \caption{$\mathbbm 1, A, B,C$ are separable algebras in $\cC$ (objects in $\Sigma \cC$); $M,N,P$ are bimodules over algebras (1-morphisms in $\Sigma \cC$).}
    \label{Alg-Defects}
\end{figure}

\subsection{Mathematical Definition}
In this paper, we mainly focus on condensation completion of the delooping of a modular tensor category $\cC$,  denoted as $\Sigma \cC\equiv \mathrm{Kar}(\mathrm{B}\cC)$, where the delooping $\mathrm{B}\cC$ is a monoidal 2-category with a single object $*$, and $\Hom_{\mathrm{B}\cC}(*,*) = \cC$. Let's consider the category $\cC$ that describes particles in a $2+1$D topological order. As for $\mathrm{B}\cC$, particles are just interpreted as $0$d domain walls between trivial $1$d defects. Physically, the goal of condensation completion is to obtain a fusion 2-category $\mathfrak{C}$, where objects in $\mathfrak{C}$ are string-like topological defects in this topological order, 1-morphisms are all possible particle-like defects, and 2-morphisms are instantons. Gaiotto and Johnson-Freyd provided an inductive construction of the condensation completion of an $n$-category in their paper. Truncate their general construction to $2$-category case, one can get a concrete workable category model of $\Sigma\cC$. We list the data of $\Sigma\cC$ as following (see Fig.\ref{Alg-Defects} for the physical picture),
\begin{enumerate}
    \item Objects in $\Sigma\cC$ are separable algebras (see~\ref{SeparableAlg}) in $\mathcal{C}$.  Two objects $A,A'$ are isomorphic in $\Sigma \mathcal{C}$ if they are Morita equivalent algebras (see~\ref{MoritaEqAlg}) in $\mathcal{C}$, denoted by $A\sim_{M} A’$.
    \item Given two separable algebras $(A,\mu_A, \eta_A)$ and $(B,\mu_B,\eta_B)$, the hom-category $\Hom(A,B)$ is the category of $B$-$A$-bimodules $_B \cC_A$, i.e., 1-morphisms in $\Sigma\mathcal{C}$ are bimodules in $\cC$, and 2-morphisms are bimodule maps. 
    \item The composition of hom-category is a functor $\circ$ definded by the relative tensor product as followings,
    \begin{equation}
        \label{FusionOfPoints}
        \circ:=\ot[B]: {}_D\cC_B\times{}_B\cC_A\to{}_D\cC_A
    \end{equation}
    where $A,B,D\in\Sigma\cC$. The relative tensor product of bimodules over algebras and bimodule maps is defined in ~\ref{RTP}.
    
    \item The tensor product $\Box$ in $\Sigma\mathcal{C}$ is induced by the braiding structure of $\cC$, more explicitly, 
    \begin{itemize}
        \item for two objects $(A,\mu_A,\eta_A), (B,\mu_B,\eta_B)\in\Sigma\cC$, the tensor product $A\Box B:=(A\otimes B, \mu_{A\otimes B},\eta_{A\ot B})$ with multiplication defined as followings,
    
    
    \begin{equation}
    \label{multi}
        \begin{tikzcd}
	{(A\otimes B) \otimes (A\otimes B)} && {A\otimes (B \otimes A)\otimes B} \\
	&& {A\otimes (A \otimes B)\otimes B} \\
	{A\otimes B} && {(A\otimes A) \otimes (B\otimes B)}
	\arrow["{\mu_{A\otimes B}}"', from=1-1, to=3-1]
	\arrow["\alpha", from=1-1, to=1-3]
	\arrow["c_{B,A}", from=1-3, to=2-3]
	\arrow["\alpha", from=2-3, to=3-3]
	\arrow["{\mu_A\otimes \mu_B}", from=3-3, to=3-1]
\end{tikzcd}
    \end{equation}
    where $\alpha$ is the associator in $\mathcal{C}$; $\mu_A$ and $\mu_B$ are multiplication of $A$ and $B$, respectively; and $\eta_{A\ot B} = \eta_A\ot \eta_B$
    \item for two 1-morphisms $(M,\rho_M,\tau_M)\in{}_{A'}\cC_{A},(N,\rho_N,\tau_N)\in{}_{B'}\cC_{B}$, the tensor product $M\Box N:=(M\otimes N, \rho_{M\Box N}, \tau_{M\Box N})$ with actions defined as followings,
    \begin{equation}
    \label{TP1morR}
        \begin{tikzcd}
	{(M\otimes N)\otimes (A\otimes B)} && {M\otimes (N\otimes A)\otimes B} \\
	&& {M\otimes (A\otimes N)\otimes B} \\
	{M\otimes N} && {(M\otimes A)\otimes (N\otimes B)}
	\arrow["{\tau_{M\Box N}}"', from=1-1, to=3-1]
	\arrow["\alpha", from=1-1, to=1-3]
	\arrow["{c_{N,A}}", from=1-3, to=2-3]
	\arrow["\alpha", from=2-3, to=3-3]
	\arrow["{\tau_M \otimes \tau_N}"', from=3-3, to=3-1]
        \end{tikzcd}
    \end{equation}
    \begin{equation}
    \label{TP1morL}
        \begin{tikzcd}
	{(A'\otimes B')\otimes (M\otimes N)} && {A'\otimes (B'\otimes M)\otimes N} \\
	&& {A'\otimes (M\otimes B')\otimes N} \\
	{M\otimes N} && {(A'\otimes M)\otimes (B'\otimes N)}
	\arrow["{\rho_{M\Box N}}"', from=1-1, to=3-1]
	\arrow["\alpha", from=1-1, to=1-3]
	\arrow["{c_{B',M}}", from=1-3, to=2-3]
	\arrow["\alpha", from=2-3, to=3-3]
	\arrow["{\rho_M \otimes \rho_N}"', from=3-3, to=3-1]
        \end{tikzcd}
    \end{equation}
    \item for two 2-morphisms $f,g$, the tensor product $f\Box g:=f\otimes g$ as bimodule map.
    \item for four 1-morphisms $M\in {}_{A'}\cC_{A}$, $N\in{}_{B'}\cC_{B}$, $P\in {}_{A''}\cC_{A'}$, $Q\in{}_{B''}\cC_{B'}$, the functoriality of tensor product gives a natural 2-isomorphism 
    \begin{equation}
    \label{tensorator}
        \phi_{P,Q,M,N}:(P\Box Q)\circ(M\Box N)\cong (P\circ M)\Box (Q\circ N) 
    \end{equation}
    which is induced by the braiding in $\cC$ as following,
    \begin{equation}
        \begin{tikzcd}
	P\ot Q\ot M\ot N && (P\ot Q)\ot[A'\ot B'](M\ot N) \\
	\\
	P\ot M\ot Q\ot N && (P\ot[A'] M)\ot (Q\ot[B'] N)
        \arrow["u_{A'\ot B'}", from=1-1, to=1-3]
	\arrow["{\phi_{P,Q,M,N}}", from=1-3, to=3-3]
	\arrow["{c_{Q,M}}"', from=1-1, to=3-1]
	\arrow["u_{A'}\ot u_{B'}"', from=3-1, to=3-3]
        \end{tikzcd}
    \end{equation}
    where $u_{A'\ot B'},u_{A'}, u_{B'}$ are quotient maps for corresponding relative tensor products. The interchange 2-isomorphism is given by the composition of two 2-isomorphisms \eqref{tensorator} as follows,
    \begin{widetext}
        \begin{equation}
        \phi_{P,N}:(P\Box B')\circ(A'\Box N)\xrightarrow{}(P\circ A')\Box (B'\circ N)\cong (A''\circ P)\Box(N\circ B)\xrightarrow{} (A''\Box N)\circ(P\Box B)
    \end{equation}
    \end{widetext}
    \end{itemize}
\end{enumerate}
$\Sigma \mathcal{C}$ has actually much more data and conditions they satisfy than what we have mentioned above, we refer readers to see ~\cite{xi2023class} for the detail of the condensation completion of a general braided fusion category. With this model $\Sigma\mathcal{C}$, physical intuitions are matched by concrete mathematical objects, as shown in table ~\ref{corres}. We will identify them with each other from now on.

\begin{remark}
\label{WhyInvertible}
 Two $1$d defects $A, A'$ are isomorphic as objects in $\Sigma\cC$ when they are Morita equivalent; $A,A'$ are not necessarily isomorphic as algebras. There can be nontrivial but invertible $0$d domain wall between them. A pair of invertible $0$d domain walls between $A$ and $A'$, $A'$ and $A$ can shrink to the trivial $0$d defect on $A$, and that's exactly the definition of invertible bimodules and Morita equivalence \ref{MoritaEqAlg}. $A$ and $A'$ are physically different 1d defects, the invertible 0d defects between them are nontrivial and can not be omitted. If we choose a defect in every Morita class as the representative, and identify every defect as direct sum of these representatives and $0$d defects between them like what we do in $1$-category, the identification is always up to some invertible $0$d defects.
\end{remark}

\begin{remark}
    \label{rmk.RelationToKar(BC)}
    According to the construction in \cite{gaiotto2019condensations},
for a general monoidal category $\cC$, objects in $\kar(\mathrm{B}\cC)$ are so called condensation algebras, which are non-unital special Frobenius algebras in $\cC$. A non-unital special Frobenius algebra is a triple $(A,\mu:A\ot A\rightarrow A,\Delta: A\rightarrow A\ot A)$. $(A,\mu)$ is a (non-unital) associative algebra, $(A,\Delta)$ is a (non-counital) coassociative coalgebra.  We use the red line to denote the algebra $A$, and  
\begin{equation*}
     \begin{tikzpicture}
        \draw[thick,red](0,0) -- (0.5,0.5);
        \draw[thick,red](0.5,0.5)--(1,0);
        \draw[thick, red](0.5,0.5)--(0.5,1);
        \fill [black] (0.5,0.5) circle [radius=1pt];
        \node[black] at (0.8,0.8) {$\mu$};    
        \draw[thick,red](3.5,0)--(3.5,0.5);
        \draw[thick,red](3,0.8)--(3.5,0.5);
        \draw[thick,red](4,0.8)--(3.5,0.5);
        \node[black] at (3.8,0.3) {$\Delta$};
        \fill [black] (3.5,0.5) circle [radius=1pt];
    \end{tikzpicture}
\end{equation*}
are the multiplication and comultiplication.  
They satisfy the special Frobenius conditions.
\begin{equation}
    \begin{tikzpicture}[baseline={(0.5,0.75)}]
        \draw[thick,red](0.5,0)--(0.5,0.3);
        \draw[thick,red](0.5,0.3)--(0,0.8);
        \draw[thick,red](0.5,0.3)--(1,0.8);
        \draw[thick,red](0,0.8)--(0.5,1.3);
        \draw[thick,red](0.5,1.3)--(1,0.8);
        \draw[thick,red](0.5,1.3)--(0.5,1.6);
        \fill[black] (0.5,0.3) circle [radius = 1pt];
        \fill[black] (0.5,1.3) circle [radius = 1pt];
    \end{tikzpicture} = \ 
    \begin{tikzpicture}[baseline={(0.5,0.75)}]
        \draw[thick,red] (0.5,0)--(0.5,1.6);
    \end{tikzpicture}
\end{equation}
\begin{equation}
\label{Ainter}
\begin{tikzpicture}[baseline={(0.5,0.75)}]
    \draw[thick,red](0,0) -- (0.5,0.5);
    \draw[thick,red](0.5,0.5)--(1,0);
        \draw[thick, red](0.5,0.5)--(0.5,1);
        \draw[thick,red](0.5,1)--(0,1.5);
        \draw[thick,red](0.5,1)--(1,1.5);
        \fill [black] (0.5,0.5) circle [radius=1pt];
        \fill [black] (0.5,1) circle [radius=1pt];
\end{tikzpicture}    = 
\begin{tikzpicture}[baseline={(0.5,0.75)}]
\draw[thick,red] (0,0)--(0.2,0.5);
\draw[thick,red] (0.2,0.5)--(1,1);
\draw[thick,red] (1,1)--(1.2,0);
\draw[thick,red] (0.2,0.5)--(0,1.5);
\draw[thick,red](1,1)--(1.2,1.5);
\fill[black] (0.2,0.5) circle [radius=1pt];
\fill[black] (1,1) circle [radius=1pt];
\end{tikzpicture}=
\begin{tikzpicture}[baseline={(0.5,0.75)}]
\draw[thick,red] (0,0)--(0.2,1);
\draw[thick,red] (0.2,1)--(1,0.5);
\draw[thick,red] (1,0.5)--(1.2,1.5);
\draw[thick,red] (0.2,1)--(0,1.5);
\draw[thick,red](1,0.5)--(1.2,0);
\fill[black] (0.2,1) circle [radius=1pt];
\fill[black] (1,0.5) circle [radius=1pt];
\end{tikzpicture}
\end{equation}
A 1-morphism between $(A,\mu_A,\Delta_A)$ and $(B,\mu_B,\Delta_B)$ is a so-called condensation bimodule over $A$ and $B$.
The theorem 3.3.3 in \cite{gaiotto2019condensations} shows that if $\cC$ is rigid (any object has left and right duals), then the 2-category $\kar(\mathrm{B}\cC)$ is equivalent to the the 2-category of separable algebras and bimodules in $\cC$ in the associative unital sense. Note that for a separable algebra $(A,\mu,\eta)$ in a rigid monoidal category $\cC$, although there is no canonical choice of bimodule splitting $\Delta: A\rightarrow A\ot A$, any choice of $\Delta$ extends $(A,\mu)$ to a condensation algebra $(A,\mu,\Delta)$. The lemma.3.3.2 in \cite{gaiotto2019condensations} ensures different choices of $\Delta$ give equivalent condensation algebras (as objects in $\kar(\mathrm{B}\cC)$). Moreover, if a separable algebra $(A,\mu,\eta)$  in a MTC is connected, i.e. $\hom_\cC(\mathbbm 1,A)\cong\C$, then it has a unique (symmetric normalized-special) Frobenius algebra structure \cite{Fuchs_2002,KONG2014436}.
\end{remark}

\subsection{Physical Realization}
\label{PhysicsCondensation}
Let $A$ be a separable algebra with multiplication $\mu$ \ref{SeparableAlg} in a fusion category $\cC$ and $\Delta$ be a bimodule splitting (or a chosen comultiplication). 
If $\cC$ is further the MTC of particles of a $2+1$D topological order $\mathfrak{C}$, we can construct a 1d gapped defect using the data of a separable algebra $A\in \cC$ and $\Delta$ by
\begin{enumerate}
    \item first creating $A$ particles freely along a 1d line;
    \item then turning on a large enough local interaction between $A$ particles based on eqn.\eqref{Ainter},
\end{enumerate}
as shown in Fig.~\ref{AdefectPhy}.

  \begin{figure}
    \centering
    \begin{subfigure}[b]{0.23\textwidth}
   \begin{tikzpicture}[baseline={(0.5,1)}]
    \fill[yellow!30!white] (0,0)--(3,0)--(4,2)--(1,2)--cycle;
    \draw[black] (0.2,0.3) node[right]{$\cC$};
\end{tikzpicture}    
\subcaption{}
\label{nothing}
    \end{subfigure}
    \begin{subfigure}[b]{0.23\textwidth}
        \begin{tikzpicture}[baseline={(0.5,1)}]
    \fill[yellow!30!white] (0,0)--(3,0)--(4,2)--(1,2)--cycle;
    \draw[black] (0.2,0.3) node[right]{$\cC$};
    \foreach \i in {0,0.1,0.2,...,0.9} {
        \fill[red] ($(1.1,1)!\i!(3,1)$) circle[radius=1.3pt];
    }
    \draw[black] (1.1,1)node[left]{$\cdots$};
    \draw[black] (3.5,1)node[left]{$\cdots$};
    \draw[red]  (2,1) node[above]{$\cdots A\ \cdots\ A\cdots$};
\end{tikzpicture}
\subcaption{}
\label{freeA}
    \end{subfigure}
\begin{subfigure}[b]{0.23\textwidth}
\begin{tikzpicture}[baseline={(0.5,1)}]
    \fill[yellow!30!white] (0,0)--(3,0)--(4,2)--(1,2)--cycle;
    \draw[black] (0.2,0.3) node[right]{$\cC$};
    
    \foreach \i in {0,0.1,...,0.9} {
        \coordinate (mid) at ($(1.1,1)!{\i +0.05}!(3,1)$);
        \fill[gray!30] (mid) ellipse [x radius=0.17, y radius=0.08]; 
    }
    \foreach \i in {0,0.1,...,0.9} {
        \coordinate (mid) at ($(1.1,1)!{\i +0.05}!(3,1)$);
        \draw[black, line width=0.4pt] (mid) ellipse [x radius=0.17, y radius=0.08]; 
    }
    
    \foreach \i in {0,0.1,0.2,...,0.9} {
        \fill[red] ($(1.1,1)!\i!(3,1)$) circle[radius=1pt];
    }
    
    \draw[black] (1.1,1)node[left]{$\cdots$};
    \draw[black] (3.5,1)node[left]{$\cdots$};
    \draw[red]  (2,1) node[above]{$\cdots A\ \cdots\ A\cdots$};
\end{tikzpicture}
    \subcaption{}
    \label{Ainteraction}
\end{subfigure}
\begin{subfigure}[b]{0.23\textwidth}
    \begin{tikzpicture}[baseline={(0.5,1)}]
    \fill[yellow!30!white] (0,0)--(3,0)--(4,2)--(1,2)--cycle;
    \draw[black] (0.2,0.3) node[right]{$\cC$};
    \draw[thick,red] (0.5,1)--(2.9,1)node[above]{$A$}--(3.5,1);
\end{tikzpicture}
  \subcaption{}
  \label{condenseA}
\end{subfigure}
\caption{~\ref{nothing} is a $2+1$D topological order $\cC$ in the ground state. In~\ref{freeA},  many free $A$ particles are created along a 1d line. In~\ref{Ainteraction}, each grey circle means a local interaction~\eqref{Ainter} between two $A$ particles. Then these interacting $A$ particles effectively form a gapped $1$d defect labeled by $A$ as in~\ref{condenseA}.}
\label{AdefectPhy}
\end{figure}

Note that this idea is discussed in section 2.4 in \cite{gaiotto2019condensations}, where it serves as a key physical intuition of the mathematical definition of condensation completion. Also, the lattice model of anyon condensation for string net model in~\cite{christian2023latticemodelcondensationlevinwen} is constructed aligning the similar idea.

Let's take the toric code model as an example. Consider the toric code model on the square lattice. There is a 2d Hilbert space on each link, and the Hamiltonian is 
\begin{equation}
    H_0=-\sum_v A_v-\sum_p B_p,
\end{equation}
where \begin{equation}
    A_v=\begin{tikzpicture} [baseline={(0.5,0.35)}]
    \draw[thick,red](0,0.5)--(1,0.5);
    \draw[thick,red](0.5,0)--(0.5,1);
    \fill[black] (0.5,0.5) circle [radius=1pt];
    \draw[black](0.5,0.3)node[right]{$v$};
    \end{tikzpicture}, \ 
    B_p = \begin{tikzpicture}[baseline={(0.5,0.15)}]
        \draw[thick,blue](0,0)--(0.5,0)--(0.5,0.5)--(0,0.5)--cycle;
        \draw[black](0.04,0.25) node[right]{$p$};
    \end{tikzpicture}
\end{equation}
is the product of $\sigma^x$ on the four edges connected with the vertex $v$, and product of $\sigma^z$ on the four edges around the plaquette $p$, respectively.

It's well known that there are four simple anyons $\mathbbm 1, e, m, f$ in toric code, whose fusion and braiding are described by the MTC $\ve_{\Z_2\times Z_2}$ (we will give the categorical data in section \ref{ToricCode}. In the lattice model, an $e$ particle on the vertex $v_0$ can be represented by a state with $A_{v_0}=-1$ and $A_v=B_p=1$ for $v\neq v_0$ and any plaquette $p$; a $m$ particle on the plaquette $p_0$ can be represented by a state with $B_{p_0}=-1$, and $A_v=B_p=1$ for $p\neq p_0$ and any vertex $v$; and $f$ is the fusion of $e$ and $m$.

In section~\ref{ToricCode}, we will find that there are 6 gapped domain walls in toric code, corresponding to 6 separable algebras in $\ve_{\Z_2\times \Z_2}$. In the following we will construct the lattice model with gapped domain wall labeled by separable algebras $\mathbbm 1\oplus e$ and $\mathbbm 1\oplus f$ in $\ve_{\Z_2\times \Z_2}$. The other three nontrivial gapped domain walls can be constructed similarly.

For $A=\mathbbm 1\oplus e$, it is a 2-dimensional algebra. We can choose the basis $\{\ket{\mathbbm 1},\ket{e}\}$, where $\ket{\mathbbm 1}$ is the unit, and the multiplication and comultiplication are given by 
\begin{equation}
\label{multi1+e}
    \begin{split}
        \mu: (\mathbbm 1\oplus e)\otimes (\mathbbm 1\oplus e)& \rightarrow \mathbbm 1\oplus e\\
           \ket{e} \ket{e}&\mapsto \ket{\mathbbm 1}
    \end{split}
\end{equation}
\begin{equation}
\label{comulti1+e}
    \begin{split}
        \Delta: \mathbbm 1\oplus e& \rightarrow (\mathbbm 1\oplus e)\otimes (\mathbbm 1\oplus e)\\
         \ket{\mathbbm 1}&\mapsto \frac{1}{2} (\ket{\mathbbm 1}\ket{\mathbbm 1
         }+\ket{e}\ket{e})\\
         \ket{e}&\mapsto \frac{1}{2}(\ket{\mathbbm 1}\ket{e}+\ket{e}\ket{\mathbbm 1})
    \end{split}
\end{equation}
To construct a $\mathbbm 1\oplus e$ wall, we first remove all the $A_v$ operators on the line as shown in fig.\ref{1+eWallLattice}. 
\begin{equation}
 H_{\mathbbm 1\oplus e-\mathrm{free}} = -\sum_{v\notin L}A_v-\sum_p B_p   
\end{equation}
For $v\in L$, $A_v=\pm 1$ both correspond to the ground states, it's equivalent to create free $\mathbbm 1\oplus e$ particles on the vertices along $L$.  Note that the local interaction $\Delta\circ\mu$ \eqref{Ainter} maps $\ket{\mathbbm1}\ket{\mathbbm 1}$ and $\ket{e}\ket{e}$ to $\frac{1}{2} (\ket{\mathbbm 1}\ket{\mathbbm 1
         }+\ket{e}\ket{e})$, maps $\ket{\mathbbm 1}\ket{e}$ and $\ket{e}\ket{\mathbbm 1}$ to $\frac{1}{2}(\ket{\mathbbm 1}\ket{e}+\ket{e}\ket{\mathbbm 1})$. It can be effectively realized by
\begin{equation}
\label{1+eLocalInter}
    \frac{1+\sigma_l^z}{2},\  \forall\ \mathrm{link}\ l\in L,
\end{equation}
since $\sigma^z$ creates a pair of $e$ particles. Dropping the constant terms, the $\mathbbm 1\oplus e$ domain wall model is 
\begin{equation}
    H_{\mathbbm 1\oplus e-\mathrm{condense}} = -\sum_{v\notin L}A_v-\sum_p B_p-\lambda\sum_{l\in L}\sigma_l^z
\end{equation}

 When $\lambda\gg 1$, the local interaction equivalently fix the degree of freedom on the link along the wall to be $\sigma^z_l=1$, and the $\mathbbm 1\oplus e$ wall is the fusion of two rough boundaries of toric code~\cite{kong2022invitation}. 
\begin{figure}
    \centering
    \begin{tikzpicture}[thick, scale=0.8]
  \foreach \x in {1,2,3} {
    \ifnum\x=2
      \draw[blue] (\x,0) -- (\x,4); %
    \else
      \draw[black] (\x,0) -- (\x,4); 
    \fi
  }
  \foreach \y in {1,2,3} {
    \draw[black] (0,\y) -- (4,\y);
  }
  \foreach \y in {1,2,3} {
    \fill[white] (2,\y) circle[radius=0.8pt]; 
    \draw[black, line width=0.4pt] (2,\y) circle[radius=1pt]; 
  }
  \draw[blue] (2,0) node[right]{$L$};
\end{tikzpicture}
    \caption{$\mathbbm 1\oplus e$ wall, the circles on the vertices mean the $A_v$ operators are removed. We turn on a large enough local interaction $-\lambda\frac{1+\sigma_l^z}{2}$ for links aling the blue line.}
    \label{1+eWallLattice}
\end{figure}

For $A=\mathbbm 1\oplus f$, the algebra multiplication and comultiplication is similar to $\mathbbm 1\oplus e$ but replacing $e$ in \eqref{multi1+e} and \eqref{comulti1+e}  with $f$. To construct the $\mathbbm 1\oplus f$ wall, we first create free $\mathbbm 1\oplus f$ particles along a (thicken) line $L$. As shown in fig.\ref{1+fWallLattice}, the thicken line $L$ cover a line of vertical links and a line of horizontal links on the adjacent right side. An $f$ anyon is the composition of an $e$ anyon and an $m$ anyon. We fix a framing convention here that the $m$ anyon is always to the lower right of the $e$ anyon. So an $f$ anyon on the line $L$ is located at a pair $(v,p)$ of a vertex $v$ and a plaquette $p$ to the lower right of $v$. Similarly, we'll denote a pair of links on the horizontal line and vertical line as $l_\parallel$ and $l_\perp$, respectively, where $l_\parallel$ is to the lower right of $l_\perp$ (see fig.~\ref{1+fTrap}). To trap free $\mathbbm 1\oplus f$ particles, we first remove the $A_v$ and $B_p$ operators for all $(v,p)$ covered by $L$ and then add the interaction $-A_vB_p$. The ground state of the deformed Hamiltonian will correspond to vacuum $A_v=B_p=1$ or an  $f$ particle $A_v=B_p=-1$

\begin{equation}
    H_{\mathbbm 1\oplus f-\mathrm{free}} = -\sum_{v\notin L}A_v-\sum_{p\notin L}B_p-\sum_{(v,p)\in L}A_vB_p.
\end{equation}

The interaction $\Delta\circ \mu$ for $\mathbbm 1\oplus f$ can be effectively realized by 

\begin{equation}
    -\frac{1+\sigma^z_{l_\perp}\sigma^x_{l_\parallel}}{2},\ \forall (l_\perp, l_\parallel) \in L 
\end{equation}
since $\sigma^z_{l_\perp}\sigma^x_{l_\parallel}$ create a pair of $f$ particles. Then,
\begin{equation}
\begin{split}
\label{1+fcondense}
    &H_{\mathbbm 1\oplus f-\mathrm{condense}}\\ =& -\sum_{v\notin L}A_v-\sum_{p\notin L}B_p-\sum_{(v,p)\in L}A_vB_p-\lambda\sum_{(l_\perp,l_\parallel)\in L}\sigma^z_{l_\perp}\sigma^x_{l_\parallel}
\end{split}
\end{equation}

\begin{figure}
    \centering
    \begin{tikzpicture}[thick, scale=0.8]
    \foreach \y in {0,1,2,3}{
    \fill[purple, opacity=0.3] (0.9,\y) rectangle (1.8,\y+1);
    }
  \foreach \x in {1,2,3} {
      \draw[black] (\x,0) -- (\x,4); 
  }
  \foreach \y in {1,2,3} {
    \draw[black] (0,\y) -- (4,\y);
  }
  \foreach \y in {1,2,3} {
    \fill[white] (1,\y) circle[radius=0.8pt]; 
    \draw[black, line width=0.4pt] (1,\y) circle[radius=1pt]; 
  }
    \foreach \y in {0,1,2,3} {
    \draw[black,  line width=0.8pt, dashed] 
      (1.5, \y+0.5) circle[radius=0.4]; 
  }
  \draw[purple] (1,0) node[left]{$L$};
\end{tikzpicture}
    \caption{$\mathbbm 1\oplus f$ wall, the circles on the vertices and in the plaquettes mean the $A_v$ and $B_p$ operators are removed. For every pair $(v,p)$ covered by the  purple line, we use $-A_vB_p$ to trap an $\mathbbm 1\oplus f$ particle; and we turn on a large enough interaction $-\lambda\frac{1+\sigma^z_{l_\perp}\sigma^x_{l_\parallel}}{2},\ \forall (l_\perp, l_\parallel) \in L$. }
    \label{1+fWallLattice}
\end{figure}

\begin{figure}
    \centering
      \begin{tikzpicture}[thick, scale=0.8]
\fill[purple,opacity=0.3](-0.1,2) rectangle (0.9,0);
\draw[black, thick](-1,1)--(0.2,1)node[below]{$v$}--(1,1);
\fill[black] (0,1) circle[radius=1.2pt];
\draw[black, thick](0,0)--(0,2);
\fill[gray,opacity=0.5] (0,1) rectangle (1,0);
\draw[black] (0.5,0.7)node[below]{$p$};
\draw[purple] (0,1.5)node[left]{$L$};
    \end{tikzpicture}
    \begin{tikzpicture}[thick,scale=0.8]
        \fill[purple,opacity=0.3](-0.1,2) rectangle (0.9,0);
\draw[black, thick](-1,1)--(1,1);
\fill[black] (0,1) circle[radius=1.2pt];
\draw[black, thick](0,0)--(0,2);
\fill[gray,opacity=0.5] (0,1) rectangle (1,0);
\draw[black,thick] (0,1.5)node[right]{$l_{\perp}$};
\draw[black,thick] (0.2,0.7)node[right]{$l_{\parallel}$};
\draw[purple] (0,1.5)node[left]{$L$};
    \end{tikzpicture}
    \caption{a pair ($v,p$) $\in L$ and a pair ($l_{\perp},l_{\parallel}$)$\in L$}
    \label{1+fTrap}
\end{figure}

 We denote the $\sigma^z$ eigenstates with eigenvalues $\pm 1$  to be $\ket{\uparrow},\ket{\downarrow}$, respectively; denote the $\sigma^x$ eigenstate with eigenvalues $\pm 1$ to be $\ket{+},\ket{-}$, respectively. 
When $\lambda\gg 1$, the local interaction forces  $\sigma^z_{l_\perp}\sigma^z_{l_{\parallel}}$ to be $+1$, the local space \begin{tikzpicture}
    \draw[black,thick] (0.4,0)--(0,0)--(0,0.4);
\end{tikzpicture} on the wall is effectively a spin $\frac{1}{2}$, and will be denoted as a red link \begin{tikzpicture}
    \draw[red,thick] (0, 0.2)--(0.4,0.2);
\end{tikzpicture}. 
We identify $\ket{\uparrow,+}$ with $\ket{+}$ and $\ket{\downarrow,-}$ with $\ket{-}$, then the $B_p$ operator adjacent to the left hand side of the wall acting on the effective spin is equivalent to 
\begin{equation}
    C_{s} = \begin{tikzpicture}[baseline={(0.8,0.2)}]
        \draw[black,thick] (0.6,0.6)--(0.3,0.6)node[above]{$\sigma^z$}--(0,0.6)--(0,0.3)node[left]{$\sigma^z$}--(0,0)--(0.3,0)node[below]{$\sigma^z$}--(0.6,0);
        \draw[black,dotted] (0.6,0)--(0.6,0.6);
        \draw[red,thick] (0.6,0.3)--(1.2,0.3);
        \draw[black, thick] (1,0.3)node[above]{$\sigma^x$};
        \draw[black, thick] (0.1,0.3)node[right]{$s$};
    \end{tikzpicture},
\end{equation}

 and the $A_vB_p$ operator on the wall acting on the effective spin is equivalent to
 \begin{equation}
     D_t = \begin{tikzpicture} [baseline = {(0.8,0.2)}]
         \draw[black, thick] (0,0.3)--(0.3,0.3)node[above]{$\sigma^x$}--(0.6,0.3); 
         \draw[red,thick] (0.6,0)--(1.2,0);
         \draw[black,thick] (0.9,0)node[below]{$\sigma^y$};
         \draw[red,thick] (0.6,0.6)--(1.2,0.6);
         \draw[black,thick] (0.9,0.6)node[above]{$\sigma^y$};
         \draw[black, thick] (1.2,0)--(1.2,0.3)node[right]{$\sigma^z$}--(1.2,0.6);
          \draw[black,dotted] (0.6,0)--(0.6,0.6);
          \draw[black,thick] (0.7,0.3)node[right]{$t$};
     \end{tikzpicture}.
 \end{equation}
 As shown in fig.\ref{ZigZagWall}, the deformed Hamiltonian~\eqref{1+fcondense} now is equivalent to 
 \begin{equation}
 \label{zigzagH}
     H_{\mathbbm 1\oplus f-\mathrm{condense}}\\ = -\sum_{v\notin L}A_v-\sum_{p\ \text{in bulk}}B_p-\sum_{s}C_s-\sum_{t} D_t.
 \end{equation}
Note that \eqref{zigzagH} constructed from separable algebra $\mathbbm 1\oplus f$ recovers the result in \cite{Kitaev_2012}. It's a commuting projector Hamiltonian, the four kinds of point-like defects on the wall are given by vacuum, $C_s=-1$, $D_t=-1$, and their fusion. Remember the string operators for $e$ and $m$ particles are $\prod \sigma^z$ and $\prod \sigma^x$, respectively. Moving an $e$ particle from the left bulk to the wall and moving an $m$ particle from the right bulk to the wall end up being $D_t=-1$. Moving an $m$ particle from the left bulk to the wall and moving an $e$ particle from the right bulk to the wall end up being $C_s=-1$. So passing through the zigzag wall exchanges $e$ and $m$ particles.

 \begin{figure}
     \centering
     \begin{tikzpicture}[scale=0.8]
    \draw[black, thick] (1, -0.5) -- (1, 3.5); 
     \draw[black, dotted] (2, -0.5) -- (2, 3.5);
\foreach \y in {0.5,1.5,2.5} {
    \draw [black,thick](0, \y) -- (2, \y); 
}

\foreach \y in {0,1,2,3} {
    \draw [black,thick](2, \y) -- (4, \y); 
}

\foreach \y in {0,1,2,3} {
    \draw [red,thick](2, \y) -- (3, \y); 
}
  \draw[black,thick] (3, -1) -- (3, 4); 
\end{tikzpicture}
     \caption{$\mathbbm 1\oplus f$ condensation effectively realizes the zigzag wall which exchanges $e$ and $m$ particles. The red line is the effective spin fixed by $\sigma^z_{l_\perp}\sigma^z_{l_{\parallel}}=1, (l_{\perp},l_{\parallel})\in L$.}
     \label{ZigZagWall}
 \end{figure}

\begin{remark}
    \label{WhyCondensation}
 For higher dimensional topological orders,  codimension-1 defects can be realized in a similar way. Let $\cC$ be a braided fusion $(n-1)$-category of codimension-2 and higher defects in a $n+1$D topological order, a condensation algebra is an object $A\in \cC=\Omega B\cC$, viewed as a domain wall between trivial codimension-1 defect, satisfying a system of commuting condensation diagrams, including associativity,  co-associativity, separability and so on (see Proposition 2.2.4 in \cite{gaiotto2019condensations} for details). Follow the similar process, we freely generate many $A$ defects on a codimension-1 ``surface", and turn on the interaction between $A$ defects according to their algebraic structure. It will result into an effective codimension-1 defect labeled by $A$. 
\end{remark}
\begin{remark}
    In~\cite{Kitaev_2012}, Kiteav and Kong constructed all gapped 1d defects of string-net model. For a string-net model with input fusion category $\cC$,  a gapped 1d defects is given by a $\cC$-$\cC$ bimodule category $\cM$. $\cC$-$\cC$ bimodule categories form a fusion 2-category $\Fun(\Sigma\cC\boxtimes(\Sigma \cC)^{\mathrm{op}},2\ve)\simeq \Fun(\Sigma\cC,\Sigma\cC)\simeq \Sigma Z_1(\cC)$, which is exactly the condensation completion of the category of anyons in the model. 
\end{remark}

\onecolumngrid
    \begin{center}
        \begin{table}
    \centering
    \begin{tabular}{|c|c| c| c|}
    \hline
       &$\mathfrak{C}$  &$\Sigma\mathcal{C}$  \\
       \hline
     object  & 1d defect & separable algebra in $\mathcal{C}$\\
     \hline
     1-morphism & 0d gapped domain wall between 1d defect  & bimodule over algebras in $\mathcal{C}$\\
     \hline
     2-morphism & instanton &  bimodule map\\
     \hline
     tensor product & fusion of 1d defects  & tensor product of separable algebras\\
     \hline
    \end{tabular}
    \caption{fusion 2-category model for 2+1D topological orders}
    \label{corres}
\end{table}
    \end{center}
\twocolumngrid

\section{ Algorithms of Calculation}
 \label{Algorithms}
We encourage readers who are not familiar with the language of category to see \ref{AlgebraAndModule} for basic definitions and properties of algebras and modules in the context of category theory.

\subsection{Finding $1$d defects} 
$1$d defects are  separable algebras.
This paper will focus on the case when the modular tensor category $\mathcal{C}$ is pointed braided fusion category. The most general example of a pointed fusion category
 is the category $\text{Vec}_G^{\omega,c}$
  of finite dimensional vector spaces graded by a finite abelian group
 $G$ with the associator twisted by the 3-cocycle $\omega \in Z^3(G, U(1))$ and braiding $c$. For a general finite group $G$, according to \cite{10.1155/S1073792803205079,etingof2016tensor}, the separable algebras in  $\text{Vec}_G^\omega$ are classified by a pair $(H,\psi)$ up to Morita equivalence. Where $H$ is a subgroup of $G$ and $\psi\in \text{C}^2(H, U(1))$ satisfying $d\psi = \omega|^{-1}_{H\times H\times H}$. We denote the group algebra of $H$ with the multiplication twisted by $\psi$ as $[H,\psi]$. $\psi$ is omitted when it is trivial. We use $[H,\psi]$ as a representative in each Morita equivalence class. We'll denote the category of right $[H,\psi]$-modules  $(\text{Vec}_G^\omega)_{[H,\psi]}$ as $\cM(H,\psi)$.
  However, the Morita class of separable algebras in $\text{Vec}_G^\omega$ is not in one-to-one correspondence with $(H,\psi)$; there is still redundancy. If two such twisted group algebras are conjugated to each other under the adjoint action of $G$, then their module categories are equivalent as $\text{Vec}_{G}^\omega$-module\cite{Natale_2017}. One should quotient out this equivalence relation to get the correct classification. We refer readers to \ref{Gomega} for the details.
 \subsection{Fusion rule of $1$d defects}
  The fusion of two $1$d defects is the tensor product of two algebras defined in \ref{CondensationCompletion}, and can be expressed as a Morita equivalent algebra which is the direct sum of some simple $1$d defects. 
 
Let $\cA$ be a fusion n-category, $A\in \cA$ be an object, and $\{A_i\}$ be a set of representative of isomorphic class of simple objects in $\cA$.
We can then do a direct sum decomposition 
\begin{equation}
    A\simeq\oplus_{i} (V_{A}^i \odot A_i)
\end{equation}
where $V_A^i\in n\ve$ is a linear $(n-1)$-category, and $\odot: n\ve \times \cA \rightarrow \cA$ is the action of $n\ve$ on $\cA$, which is  determined by the linear structure of $\cA$ 
 \begin{equation}
     \Hom_{\cA}(V\odot A,B) \simeq \Hom_{n\ve}(V, \Hom_{\cA}(A,B)), 
 \end{equation}
 $\forall V \in n\ve; A,B\in \cA$. When $n=1$, the right hand side is nothing but a vector space with dimension $\dim(V)\dim(\Hom_{\cA}(A,B))$, so $V\odot A = A^{\oplus \dim(V)}$. When $n=2$, $\Hom_{\cA}(A,B)$ is a finite semisimple linear 1-category, which is just a direct sum of $|\Hom_{\cA}(A,B)|$ $\ve$'s, where $|\Hom_{\cA}(A,B)|$ is the number of simple objects in $\Hom_{\cA}(A,B)$. So $\Hom_{2\ve}(V,\Hom_{\cA}(A,B))$ is $|V||\Hom_{\cA}(A,B)|$ copies of $\ve$, and then $V\odot A \simeq A^{\oplus |V|}$. For $n\geq 3$, $\Hom_{\cA}(A,B)$ is no longer simply some copies of $(n-1)\ve$, $V\odot X$ is then not some copies of $X$.

When $n=2$,
accroding to the unique decomposition theorem in semisimple 2-category\cite{douglas2018fusion},
$V_A^i$ doesn't depend on the choice of representatives but only the equivalent classes of simple objects. 
It's easy to see in general 
$V_{A}^i\neq \Hom_\cA(A,A_i)$. Because in a higher category, the hom-category between two non-isomorphic simple objects is in general not 0. Similarly, the fusion rule in $\cA$ is  
 \begin{equation}
\label{FusionCoeffients}
     A\Box B \simeq \oplus_{C\in \mathrm{Irr}(\cA)} V_{A,B}^C\odot C
 \end{equation}
 $V_{A,B}^C \in n\ve$, and $V_{A,B}^C\neq \Hom_\cA(A\Box B,C)$ in general, and we don't have a general procedure to determine the fusion coefficient $V_{A,B}^C$ yet.
 
 However, for a pointed braided fusion category $\cC=\ve_{G}^{\omega,c}$ mentioned above, the relative Deligne tensor product of finite semisimple modules over $\cC$ was studied in~\cite{Decoppet_2023RDT}. Let $\cM(E,\phi)$, $\cM(F,\psi)$ be two left $\cC$-module, then the relative tensor product of them is multiples of a certain simple module $\cM(H,\rho)$
 \begin{equation}
 \label{RDTEQ}
     \cM(E,\phi)\boxtimes_\cC \cM(F,\psi)\simeq \cM(H,\rho)^{\oplus\alpha}
 \end{equation}
 where the subgroup $H$, $\rho\in \mathrm{H}^2(H,U(1))$ and the multiplicity $\alpha$ are all determined by the subgroups $E,F$ and $\phi\in\mathrm{H}^2(E,U(1)), \psi\in \mathrm{H}^2(F,U(1))$. We review the details in \ref{RDT}. The tensor product of separable algebra in $\cC$ then satisfies the same fusion rule
 \begin{equation}
     [E,\phi]\Box [F,\psi] \sim_{M} [H,\rho]^{\oplus \alpha}.
 \end{equation}

 \begin{remark}
The fusion of 1-dimensional gapped phases or gapped defects has studied through various methods in recent years. In~\cite{Roumpedakis_2023}, the authors calculate the fusion of 1+1D defects coming from gauging 1-form symmetry in 2+1D quantum field theory, and the fusion coefficients are 1+1D TQFT. While we are preparing this paper, the authors of~\cite{stephen2024fusiononedimensionalgappedphases} use finite depth quantum circuits and local unitaries to study the fusion of one-dimensional gapped phases and defects on them. The fusion coefficients are  also interpreted as 1+1D system. In general, when we fuse codimension-1 defects in a n+1D topological order or fuse two $n$D topological phases, the fusion coefficient $V_{A,B}^C \in n\ve$. We can interpret every indecomposable object in $n\ve$ as an anomaly-free $(n-1)+1$D topological order admitting gapped boundary \cite{Lan2023CategoryOS,Kong_2022ql1}, and $V_{A,B}^C\odot C$ is the stacking of this anomaly free topological order $V_{A,B}^C$ with the $(n-1)+1$D defect or phase $C$. Since 1+1D anomaly-free topological order is either trivial or some copies of the trivial one, so the stacking result is just some copies of $C$. In higher dimensions, for example, if we consider the fusion of membrane defects $A$ and $B$ in a $3+1$D topological order, $V_{A, B}^C$ labels some anomaly-free $2+1$D non-chiral topological order, and the stacking $V_{A,B}^C\odot C$ can be rather nontrivial.
 \end{remark}
\subsection{Finding 0d defects}
  The 0d defects between 1d defects are bimodules.
  One can use the free module decomposition to find the simple modules over an separable algebra $A$ in a fusion category $\cC$. As said in Lemma \ref{A.1}, separability of the algebra $A$ makes sure that any simple right $A$ module is a direct summand of a free module $i\otimes A$ for some simple object $i\in\cC$. To find all the simple right $A$-modules, we just need to do direct sum decomposition in the category of right $A$-module $\cC_A$,  for $i\otimes A, \forall i\in \mathcal{C}$. In the similar way we can find all simple $A$-$B$-bimodules for $A,B$ both separable algebras.
  \subsection{Fusion of 0d defects}
  After knowing all the $0$d defects, we can compute their fusion rule, both along with and perpendicular to $1$d defects. The former is to compute the composition, and the latter is to compute the tensor product.  The tensorator structure \eqref{tensorator} just means that the order of composition and tensor product of defects on the walls doesn't matter.

  The only one tricky thing is that with invertible bimodules we choose, one can represent the tensor product of bimodules with a bimodule between the Morita equivalent $1$d defects which are in the form of direct sum of some simple $1$d defects.  Remember the fusion of point like defects is the tensor product of 1-morphisms defined as \eqref{TP1morR} and \eqref{TP1morL}. If we have $M\in {}_A\cC_B$, $N\in {}_{A'}\cC_{B'}$, then $M\Box N\in {}_{A\Box A'}\cC_{B\Box B'}$. As mentioned in Remark.\ref{WhyInvertible}, if $D$ and $E$ are $1$d defects of direct sum of some representatives, and $A\Box A'\sim_M D$, $B\Box B'\sim_M E$, with the isomorphism we choose, we can identify an object in ${}_D\cC_E$ as the fusion result. More explicitly, if $S\in {}_D\cC_{A\Box A'}$ and $P\in {}_{B\Box B'}\cC_{E}$ are two invertible bimodules between $D$ and $A\Box A'$, $B\Box B'$ and $E$ respectively, then at the cost of these two invertible $0$d defects, we can use\begin{equation}
  \label{PointsWithWall}
    (S\ot[A\Box A'] M\Box N )\ot[B\Box B'] P \in {}_{D}\cC_{E}
\end{equation} 
as the fusion result of $M$ and $N$. 

\begin{remark}
Such identification is actually an equivalence between hom-categories. If $A\sim_M A'$ and $B\sim_M B'$ in $\Sigma \cC$, then ${}_A\cC_{A'}$ and ${}_B\cC_{B'}$ are equivalent, and the equivalence is induced by the 1-isomorphism i.e. invertible bimodules between $A$ and $A'$, $B$ and $B'$. Since invertible bimodules may not be unique, there is no canonical equivalence between ${}_A\cC_{A'}$ and ${}_B\cC_{B'}$. 
\end{remark}

In the cases we care about, algebras and modules are both some group element graded vector spaces, and algebra actions are linear maps. To find direct summands of a free module, one just need to find subspaces of the free module that are invariant under algebra actions. As for the fusion rule of bimodules, the composition is relative tensor product, all the maps in the relative tensor product \ref{RTP} are now linear maps. So all the calculations are straightforward linear algebra.

\section{Example: Toric Code}
\subsection{Separable Algebras in Toric Code}
\label{ToricCode}
Toric code given by Kitaev \cite{KITAEV20032} is the simplest non-chiral nontrivial 2+1D topological order. We can use the fusion category $\text{Vec}_{\mathbb{Z}_2\times \mathbb{Z}_2}$ to describe the anyons and their fusion data. We denote the group elements of $\mathbb{Z}_2\times \mathbb{Z}_2$ by $\{1,a,b,c\}$, satisfying $a^2=b^2=c^2=1, ab=c$. There are four simple objects $\mathbbm{1}, e,m,f$ in  $\text{Vec}_{\mathbb{Z}_2\times \mathbb{Z}_2}$, which only contains a 1-dimensional space graded by $\{1, a, b, c\}$, respectively. Due to the group multiplications, they satisfy the fusion rule $e\otimes e = m\otimes m = f\otimes f =\mathbbm {1}, e\otimes m = m\otimes e = f$, and then correspond to the four simple anyons in toric code. As a modular tensor category, we take the trivial associator, and the following braidings between simple objects\begin{equation}
\label{TCBraiding}
    c_{e,m}=c_{f,f}=c_{e,f}=c_{f,m}=-1,
\end{equation} 
all other braidings between simple objects are trivial.


$G=\mathbb{Z}_2\times \mathbb{Z}_2$ has $5$ subgroup, $  \mathbb Z_1=\{1\}, (\mathbb Z_2)_e=\{1,a\}, (\mathbb Z_2)_m=\{1,b\}, (\mathbb Z_2)_f = \{1,c\}, \mathbb Z_2\times \mathbb Z_2=\{1,a,b,c\} $. Since
\begin{equation}
    \text{H}^2(\mathbbm{Z}_2, U(1))=\{1\},  \text{H}^2(\mathbbm{Z}_2\times \mathbbm{Z}_2, U(1))=\mathbbm{Z}_2,
\end{equation} 
there are 6 Morita equivalent classes of separable algebras in $\text{Vec}_{\mathbb{Z}_2\times \mathbb{Z}_2}$. As objects in $\text{Vec}_{\mathbb{Z}_2\times \mathbb{Z}_2}$, 
\begin{equation}
    \begin{split}
        &[\mathbb Z_1]=\mathbbm1,\\
        &[(\Z_2)_e]=\mathbbm1\oplus e, [(\Z_2)_m)]=\mathbbm1\oplus m, [(\Z_2)_f]=\mathbbm 1\oplus f, \\&
        [\Z_2\times \Z_2]=\mathbbm1\oplus e\oplus m\oplus f, \\ &[(\Z_2\times \Z_2,\psi)]=(\mathbbm1\oplus e\oplus m\oplus f)_{\psi}.
    \end{split}
\end{equation} We can choose the basis vector in $\mathbbm 1, e, m ,f$ as $|\mathbbm1\rangle, |e\rangle, |m\rangle, |f\rangle$ such that the multiplication rules of $[\Z_2\times \Z_2]$ are similar to the fusion rules, as shown in table \ref{AlgM}. $[\mathbb Z_1], [(\Z_2)_e], [(\Z_2)_m], [(\Z_2)_f] $ are just the sub-algebras of $[\Z_2\times \Z_2]$.

\begin{table}
    \centering
    \begin{tabular}{c c c c c}  
        & $|\mathbbm1\rangle$ & $|e\rangle$ & $|m\rangle$ & $|f\rangle$ \\ 
  \hline
$|\mathbbm{1}\rangle$ & $|\mathbbm{1}\rangle$ & $|e\rangle$ & $|m\rangle$ & $|f\rangle$ \\  
 \hline
 $|e\rangle$ & $|e\rangle$ & $|\mathbbm1\rangle$&$|f\rangle$& $|m\rangle$ \\
 \hline
 $|m\rangle$ &  $|m\rangle$&  $|f\rangle$&  $|\mathbbm 1\rangle$ &  $|e\rangle$\\
 \hline
 $|f\rangle$ &  $|f\rangle$& $|m\rangle$ & $|e\rangle$ & $|\mathbbm1 \rangle$ \\
    \end{tabular}
    \caption{Multiplications of $[\Z_2\times \Z_2]$, row times column}
    \label{AlgM}
\end{table}

The multiplication of $[\Z_2\times \Z_2,\psi]$ is twisted by a nontrivial 2-cocycle $\psi\in \text{H}^2(\mathbb{Z}_2\times\mathbb{Z}_2, U(1))$. For simplicity, we use \begin{equation}
    \psi(b,a)=\psi(b,c)=\psi(c,a)=\psi(c,c)=-1.
\end{equation}
Then, with the basis vector $|\mathbbm1\rangle, |e\rangle, |m\rangle, |f\rangle$, the multiplication rules of $[\Z_2\times \Z_2,\psi]$ are shown in table \ref{AlgM2}.

\begin{table}
    \centering
    \begin{tabular}{c c c c c}  
        & $|\mathbbm1\rangle$ & $|e\rangle$ & $|m\rangle$ & $|f\rangle$ \\ 
  \hline
$|\mathbbm{1}\rangle$ & $|\mathbbm{1}\rangle$ & $|e\rangle$ & $|m\rangle$ & $|f\rangle$ \\  
 \hline
 $|e\rangle$ & $|e\rangle$ & $|\mathbbm1\rangle$&$|f\rangle$& $|m\rangle$ \\
 \hline
 $|m\rangle$ &  $|m\rangle$&  $-|f\rangle$&  $|\mathbbm 1\rangle$ &  $-|e\rangle$\\
 \hline
 $|f\rangle$ &  $|f\rangle$& $-|m\rangle$ & $|e\rangle$ & $-|\mathbbm1 \rangle$ \\
    \end{tabular}
    \caption{Multiplications of $[\Z_2\times \Z_2,\psi]$, row times column}
    \label{AlgM2}
\end{table}

These are 6 algebras are 6 1-dimensional gapped domain wall phases in toric code. 

\subsection{Bimodules Over Algebras}
Now we compute the bimodules of these separable algebras. We denote $\text{Vec}_{\mathbb{Z}_2\times\mathbb{Z}_2}$ as $\TC$ for simplicity.

\begin{subsubsection}{$[\Z_1]$-bimodule.}
Since $[\Z_1]=\mathbbm 1$, then $_{[\Z_1]}\TC_{[\Z_1]}=\TC$ with monoidal structure just the tensor product of $\TC$.
\end{subsubsection}

\begin{subsubsection}{$[(\Z_2)_e], [(\Z_2)_m], [(\Z_2)_f]$-bimodules.}
Since the algebra structures of these three algebras are quite similar, we just take $[(\Z_2)_m]$ as an example and just show the results of the other two.

$[(\Z_2)_m]=\mathbbm1\oplus m$, $[(\Z_2)_m]\otimes \mathbbm1 \otimes [(\Z_2)_m]=\langle\ |\mathbbm1\rangle|\mathbbm 1\rangle|\mathbbm1\rangle,|\mathbbm1 \rangle|\mathbbm1 \rangle| m \rangle, |m \rangle|\mathbbm1 \rangle|\mathbbm1 \rangle,|m \rangle|\mathbbm1 \rangle|m \rangle \ \rangle$. The left and right $A(H_3)$ actions are just $A(H_3)$ multiplications, and then we can write down the results of the $A(H_3)$ actions, see table \ref{1+m action}. 
\begin{table}
    \centering
    \begin{tabular}{c c c c c c}
        Left & $\mathbbm 1\mathbbm 1\mathbbm 1$ & $m\mathbbm 1 m$& $\mathbbm 1\mathbbm 1 m$&$m\mathbbm 1\mathbbm 1$ & Right \\ 
        \hline
         $\mathbbm 1$& $\mathbbm 1\mathbbm 1\mathbbm 1$ & $m\mathbbm 1 m$& $\mathbbm 1\mathbbm 1 m$&$m\mathbbm 1\mathbbm 1$ & \\
         \hline
         $m$ & $m\mathbbm 1\mathbbm 1$ & $\mathbbm1\mathbbm 1 m$& $m\mathbbm 1 m$&$\mathbbm1\mathbbm 1\mathbbm 1$  &\\
         \hline
         &$\mathbbm 1\mathbbm 1\mathbbm 1$ & $m\mathbbm 1 m$& $\mathbbm 1\mathbbm 1 m$&$m\mathbbm 1\mathbbm 1$& $\mathbbm1$\\
         \hline
         &$\mathbbm 1\mathbbm 1 m$ & $m\mathbbm 1 \mathbbm1$& $\mathbbm 1\mathbbm 1 \mathbbm1$&$m\mathbbm 1 m$& $m$
    \end{tabular}
    \caption{$[(\Z_2)_m]$ action on $[(\Z_2)_m]\otimes \mathbbm1\otimes [(\Z_2)_m]$, we neglect the bracket for simplicity.}
   \label{1+m action}
\end{table} 
Then there are two non-isomorphic indecomposable sub-bimodules, \begin{equation}
\label{(1+m)_0}
\begin{split}
   &M_0^m = \mathbbm 1\oplus m\\& =\langle |\mathbbm1\rangle|\mathbbm1\rangle|\mathbbm1\rangle+|m\rangle|\mathbbm1\rangle |m\rangle, |\mathbbm1\rangle|\mathbbm 1\rangle|m\rangle+|m\rangle|\mathbbm1\rangle|\mathbbm1\rangle \ \rangle
\end{split}
\end{equation}
\begin{equation}
\begin{split}
    &M_1^m = \mathbbm 1\oplus m \\&=\langle |\mathbbm1\rangle|\mathbbm1\rangle|\mathbbm1\rangle-|m\rangle|\mathbbm1\rangle |m\rangle, |\mathbbm1\rangle|\mathbbm 1\rangle|m\rangle-|m\rangle|\mathbbm1\rangle|\mathbbm1\rangle \ \rangle
\end{split}
\label{(1+m)_1}
\end{equation}

To see why they are non-isomorphic bimodules, let's first suppose there is a bimodule isomorphism. As a morphism in the category of graded vector space, it keeps grading
\begin{equation}
    \begin{split}
        \phi: M_0^m&\rightarrow M_1^m\\
        \mathbbm 1 \mathbbm 1 \mathbbm 1+ m\mathbbm 1 m&\mapsto c_1(\mathbbm 1\mathbbm 1\mathbbm 1-m\mathbbm 1 m)\\
        \mathbbm 1\mathbbm 1 m+m\mathbbm 1\mathbbm 1 &\mapsto c_2(\mathbbm 1\mathbbm 1 m- m\mathbbm 1\mathbbm 1).
    \end{split}
\end{equation}
A bimodule map commutes with the algebra action (see \ref{BimoduleMap}). For the left $\ket{m}$ action
\begin{equation}
        c_2(\mathbbm 1\mathbbm 1 m- m\mathbbm 1\mathbbm 1) = -c_1(\mathbbm 1\mathbbm 1 m- m\mathbbm 1\mathbbm 1),
\end{equation}
and for the right $\ket{m}$ action
\begin{equation}
     c_2(\mathbbm 1\mathbbm 1 m- m\mathbbm 1\mathbbm 1) = c_1(\mathbbm 1\mathbbm 1 m- m\mathbbm 1\mathbbm 1),
\end{equation}
then $c_1=c_2=0$, so $M_0^m$ and $M_1^m$ are not isomorphic to each other.

More explicitly, $M_0^m\simeq [(\Z_2)_m]$ as $[(\Z_2)_m]$-bimodule. For $M_1^m$, we choose the basis
\begin{equation*}
    \begin{split}
\ket{\underline{\mathbbm 1}}&=|\mathbbm1\rangle|\mathbbm1\rangle\mathbbm1\rangle-|m\rangle|\mathbbm1\rangle |m\rangle \\
\ket{\underline{m}}&=|\mathbbm1\rangle|\mathbbm 1\rangle|m\rangle-|m\rangle|\mathbbm1\rangle|\mathbbm1\rangle .
    \end{split}
\end{equation*}
The right $[(\Z_2)_m]$ action is
\begin{equation}
\begin{split}
\tau_{M_1^m} :\   M_1^m&\otimes[(\Z_2)_m] \rightarrow  M_1^m \\
 \ket{\underline{\mathbbm1}}&\otimes\ket{m} \mapsto \ket{\underline{m}}\\
 \ket{\underline{m}}&\otimes\ket{m} \mapsto \ket{\underline{\mathbbm 1}}\\
\end{split}
\end{equation}
Later we will refer to such form of action with no extra phase factors as of the fusion rule type. The left $[(\Z_2)_m]$ action is 
\begin{equation}
\begin{split}
\rho_{M_1^m} :\  [(\Z_2)_m]&\otimes M_1^m \rightarrow  M_1^m \\
\ket{m}&\otimes \ket{\underline{\mathbbm1}} \mapsto -\ket{\underline{m}}\\
\ket{m}&\otimes \ket{\underline{m}} \mapsto -\ket{\underline{\mathbbm 1}}\\
\end{split}
\end{equation}
Such form of action with extra phase factors will be called of twisted type. Note that different choices of bases may lead to different phase factors (more generally, matrix elements or structure coefficients) representing the same action; the terms ``fusion rule type" and "twisted type" are always relative to a given basis choice.

Similarly,
\begin{equation*}
    \begin{split}
        [(\Z_2)_m]\otimes e \otimes [(\Z_2)_m]=\langle&\ |\mathbbm1\rangle|e\rangle|\mathbbm1\rangle,|m\rangle|e\rangle|\mathbbm1\rangle,\\
        &|\mathbbm1\rangle|e\rangle|m\rangle,|m\rangle|e\rangle|m\rangle\ \rangle.\\
    \end{split}
\end{equation*}

For this case, we list the actions in table \ref{1+m action2}.

\begin{table}
    \centering
    \begin{tabular}{c c c c c c}
        Left & $\mathbbm 1em$ & $me\mathbbm1$& $\mathbbm 1e\mathbbm 1$&$mem$ & Right \\ 
        \hline
         $\mathbbm 1$&$\mathbbm 1em$ & $me\mathbbm1$& $\mathbbm 1e\mathbbm 1$&$mem$ & \\
         \hline
         $m$ & $mem$&$\mathbbm1 e \mathbbm1$&$me\mathbbm1$&$\mathbbm1 em$  &\\
         \hline
         &$\mathbbm 1em$ & $me\mathbbm1$& $\mathbbm 1e\mathbbm 1$&$mem$& $\mathbbm1$\\
         \hline
         &$\mathbbm1 e\mathbbm1$&$ee\mathbbm1$&$\mathbbm1em$&$me\mathbbm1$& $m$
    \end{tabular}
    \caption{$[(\Z_2)_m]$ action on $[(\Z_2)_m]\otimes e\otimes [(\Z_2)_m]$, we omit the bracket for simplicity.}
   \label{1+m action2}
\end{table} 

Then there are two non-isomorphic indecomposable sub-bimodules, \begin{equation}
\label{e+f_0}
\begin{split}
   &N_0^m = e\oplus f \\ &=\langle\ |\mathbbm1\rangle|e\rangle|\mathbbm1\rangle+|m\rangle|e\rangle|m\rangle, |\mathbbm1\rangle|e\rangle|m\rangle+|m\rangle|e\rangle|\mathbbm1\rangle\ \rangle\\
\end{split}
\end{equation}
\begin{equation}
\label{e+f_1}
\begin{split}
    &N_1^m = e\oplus f \\&= \langle\ |\mathbbm1\rangle|e\rangle|\mathbbm1\rangle-|m\rangle|e\rangle|m\rangle, |\mathbbm1\rangle|e\rangle|m\rangle-|m\rangle|e\rangle|\mathbbm1\rangle\ \rangle
\end{split}
\end{equation}

The left and right $[(\Z_2)_m]$ actions on $N_0^m$ are of the fusion rule type. The right action of $[(\Z_2)_m]$ on $N_1^m$ is of the fusion rule type, while the left action is twisted: 
\begin{equation}
\begin{split}
    \rho_{N_1^m}: [(\Z_2)_m] &\otimes N_1^m \rightarrow N_1^m\\
     \ket{m}&\otimes\ket{\underline{e}}\mapsto -\ket{\underline{f}}\\
     \ket{m}&\otimes \ket{\underline{f}}\mapsto -\ket{\underline{e}}
\end{split}
\end{equation}
Here the basis we use for $N_1^m$ is again \eqref{e+f_1},
\begin{equation*}
    \begin{split}
    \ket{\underline{e}} &=|\mathbbm1\rangle|e\rangle|\mathbbm 1\rangle-|m\rangle|e\rangle|m\rangle \\
\ket{\underline{f}}
&=|\mathbbm1\rangle|e\rangle|m\rangle-|m\rangle|e\rangle\ket{\mathbbm1}
    \end{split}
\end{equation*}

One can show that $[(\Z_2)_m]\otimes m \otimes [(\Z_2)_m]$, $[(\Z_2)_m]\otimes f \otimes [(\Z_2)_m]$ will not give bimodule that is not isomorphic to the 4 bimodules mentioned above. So there are just 4 simple objects in $_{[(\Z_2)_m]}\TC_{[(\Z_2)_m]}$.  

Now we compute the relative tensor product (or composition) of bimodules, which physically means the $0$d defects on a $1$d defect fuse along it. We show $ N_0^m
\ot[ {[(\Z_2)_m]} ]  M_1^m$ as an example, all other relative tensor products are computed similarly. We'll short the relative tensor product as $\ot[A]$ if there is no confusion.

The relative tensor product of $ N_0^m$ and $M_1^m$  is defined by (see \ref{RTP})
\begin{equation}
   \begin{tikzcd}
	{N_0^m\otimes [(\Z_2)_m] \otimes M^m_1} & {N^m_0\otimes M_1^m} & {N_0^m\ot[A]M_1^m}
	\arrow["\tau"', shift right, from=1-1, to=1-2]
	\arrow["\rho", shift left, from=1-1, to=1-2]
	\arrow["u", from=1-2, to=1-3]
\end{tikzcd}
\end{equation}    
where $\rho$ is the left action for $M_1$, $\tau$ is the right action for $N_0$, and $u$ is the quotient map. More concretely, let $\ket{v}\in N_0^m$, $\ket{w}\in M_1^m$, $u$ maps $\ket{v}\otimes \ket{w}$ to $\ket{v}\ot[A]\ket{w}\in N_0^m\ot[A]M_1^m$.  Since $u\circ \rho = u\circ \tau$, we have
\begin{equation}
    \begin{split}
        \ket{f}\ot[A]\ket{\mathbbm 1}&= -\ket{e}\ot[A]\ket{m}\equiv \ket{\Tilde{f}}\\
        \ket{f}\ot[A]\ket{m}& = -\ket{e}\ot[A]\ket{\mathbbm 1} \equiv \ket{\Tilde{e}}
    \end{split}
\end{equation}
The induced left action (see remark.\ref{InducedModuleAction}) is defined by
\begin{equation}
   \begin{tikzcd}
	{[(\Z_2)_m]\otimes N_0^m\otimes [(\Z_2)_m] \otimes M_1^m} & {N_0^m\otimes [(\Z_2)_m]\otimes M_1^m} \\
	{[(\Z_2)_m]\otimes N_0^m\otimes M_1^m} & {N_0^m\otimes M_1^m} \\
	{[(\Z_2)_m]\otimes N_0^m \ot[A] M_1^m} & {N_0^m\ot[A] M_1^m}
	\arrow["{\text{id}\otimes\tau}"', shift right, from=1-1, to=2-1]
	\arrow["{\text{id}\otimes\rho}", shift left, from=1-1, to=2-1]
	\arrow["{\text{id}\otimes u}", from=2-1, to=3-1]
	\arrow["\rho", from=1-1, to=1-2]
	\arrow["{\bar{\rho}}", dashed, from=3-1, to=3-2]
	\arrow["\tau"', shift right, from=1-2, to=2-2]
	\arrow["u", from=2-2, to=3-2]
	\arrow["\rho", from=2-1, to=2-2]
	\arrow["\rho", shift left, from=1-2, to=2-2]
\end{tikzcd}
\end{equation}    

By tracking the commuting diagram, we can get the induced left action $\bar{\rho}$,
    \begin{equation}
    \begin{split}
           \bar{\rho}: [(\Z_2)_m]\otimes N_0^m\ot[A] M_1^m &\rightarrow N_0^m\ot[A] M_1^m\\
           \ket{\mathbbm 1}\otimes \ket{\Tilde{e}} &\mapsto \ket{\Tilde{e}}\\
           \ket{\mathbbm 1}\otimes \ket{\Tilde{f}} &\mapsto \ket{\Tilde{f}}\\
           \ket{m}\otimes \ket{\Tilde{e}} &\mapsto -\ket{\Tilde{f}}\\
           \ket{m}\otimes \ket{\Tilde{f}} &\mapsto -\ket{\Tilde{e}}
    \end{split}
    \end{equation}
In the same way, the induced right action $\bar{\tau}$ is
    \begin{equation}
        \begin{split}
            \bar{\tau}: N_0^m\ot[A] M_1^m \otimes  [(\Z_2)_m] &\rightarrow N_0^m\ot[A]M_1^m\\
            \ket{\Tilde{e}}\otimes \ket{\mathbbm 1} &\mapsto \ket{\Tilde{e}}\\
            \ket{\Tilde{f}}\otimes \ket{\mathbbm 1} &\mapsto \ket{\Tilde{f}}\\
            \ket{\Tilde{e}}\otimes \ket{m} &\mapsto \ket{\Tilde{f }}\\
            \ket{\Tilde{f}}\otimes \ket{m} &\mapsto \ket{\Tilde{e}}
        \end{split}
    \end{equation}

so $N_0^m\ot[A]M_1^m\cong N_1^m$ is a $[(\Z_2)_m]$-bimodule with the bimodule isomorphism
\begin{equation}
\begin{split}
    \eta:     N_0^m\ot[A] M_1^m&\rightarrow N_1^m\\
                 \ket{\Tilde{e}}&\mapsto \ket{\underline e}\\
                 \ket{\Tilde{f}}&\mapsto \ket{\underline f}\\
\end{split}
\end{equation}
All the nontrivial relative tensor products $\ot[A]$ are 
    \begin{equation}
\begin{split}
   & M_1^m\ot[A] M_1^m\cong N_0^m\ot[A] N_0^m \cong N_1^m\ot[A] N_1^m \cong M_0^m, \\ 
   & M_1^m\ot[A] N_0^m \cong N_0^m\ot[A] M_1^m\cong N_1^m,
\end{split}
\end{equation}
so $_{[(\Z_2)_m]}\TC_{[(\Z_2)_m]}$ is equivalent to $\ve_{\Z_2\times \Z_2}^\omega$ for some $\omega\in \mathrm{H}^3(\Z_2\times \Z_2, U(1))$.
Since $_{[(\Z_2)_m]}\TC_{[(\Z_2)_m]}$ is Morita equivalent to $\ve_{\Z_2\times \Z_2}$ as fusion categories, we must have 
$_{[(\Z_2)_m]}\TC_{[(\Z_2)_m]}\simeq \text{Vec}_{\mathbb{Z}_2\times \mathbb{Z}_2}$ as fusion categories. Physically, $M_0^m, M_1^m, N_0^m, N_1^m$ are $4$ point-like defects on the domain wall $[(\Z_2)_m]$, their fusion on the wall satisfies the $\mathbb Z_2 \times \mathbb Z_2$ fusion rule.
\\\\
We summarize the results of the other two algebras: 

$_{[(\Z_2)_e]}\TC_{[(\Z_2)_e]}\simeq \text{Vec}_{\mathbb Z_2 \times \mathbb Z_2}$ as fusion categories. There are 4 simple objects \begin{equation}
\label{rrWall0dDefect}
   M_0^e = \mathbbm1\oplus e, M_1^e=\mathbbm 1\oplus e,  N_0^e = m\oplus f, N_1^e = m\oplus f,
\end{equation}
where $M_0^e$ is the tensor unit, $M_0^e$ and $N_0^e$ have fusion rule type of $[(\Z_2)_e]$ two-sided actions, and $M_1^e$, $N_1^e$ have fusion rule type of right 
$[(\Z_2)_e]$ actions and twisted left actions
\begin{equation}
\begin{split}
     \rho_{M_1^e}: [(\Z_2)_e] &\otimes M_1^e \rightarrow M_1^e\\
     \ket{e}&\otimes\ket{e}\mapsto -\ket{\mathbbm1}\\
     \ket{e}&\otimes \ket{\mathbbm1}\mapsto -\ket{e}
\end{split}
\end{equation}
\begin{equation}
\begin{split}
     \rho_{N_1^e}:[(\Z_2)_e] &\otimes N_1^e \rightarrow N_1^e\\
     \ket{e}&\otimes\ket{m}\mapsto -\ket{f}\\
     \ket{e}&\otimes \ket{f}\mapsto -\ket{m}.
\end{split}
\end{equation}
Physically, $ M_0^e, M_1^e, N_0^e, N_1^e$ are $4$ point-like defects on the domain wall $[(\Z_2)_e]$, they have the $\mathbb Z_2 \times \mathbb Z_2$ fusion rule along the wall.

$_{[(\Z_2)_f]}\TC_{[(\Z_2)_f]}\simeq \text{Vec}_{\mathbb Z_2 \times \mathbb Z_2}$ as fusion categories. There are 4 simple objects \begin{equation}
\label{0dDWOnExchange}
    M_0^f=\mathbbm 1\oplus f, M^f_1 = \mathbbm 1\oplus f, N^f_0 = e\oplus m, N^f_1 = e\oplus m,
\end{equation}
where $M_0^f$ is the tensor unit, $M^f_0$ and $N^f_0$ have fusion rule type of $[(\Z_2)_f]$ two-sided actions, and $M^f_1$, $N^f_1$ have fusion rule type of right 
$[(\Z_2)_f]$ actions and twisted left actions
\begin{equation}
\begin{split}
     \rho_{M_1^f}: [(\Z_2)_f]&\otimes M^f_1 \rightarrow M^f_1\\
     \ket{f}&\otimes\ket{f}\mapsto -\ket{\mathbbm1}\\
     \ket{f}&\otimes \ket{\mathbbm1}\mapsto -\ket{f}
\end{split}
\end{equation}

\begin{equation}
\begin{split}
     \rho_{N^f_1}: [(\Z_2)_f] &\otimes N_1^f \rightarrow N_1^f\\
     \ket{f}&\otimes\ket{e}\mapsto -\ket{m}\\
     \ket{f}&\otimes \ket{m}\mapsto -\ket{e}.\\
\end{split}
\end{equation}

Physically, $ M^f_0, M^f_1, N^f_0, N^f_1$ are $4$ point-like defects on the domain wall $\mathbbm 1\oplus f$, they have the $\mathbb Z_2 \times \mathbb Z_2$ fusion rule along the wall.
\end{subsubsection}

\begin{subsubsection}{$[\Z_2\times \Z_2]$-bimodule.}
$[\Z_2\times \Z_2]=\mathbbm1 \oplus e\oplus m\oplus f$, and $_{[\Z_2\times \Z_2]}\mathcal{C}_{[\Z_2\times \Z_2]}\simeq \text{Vec}_{\mathbb Z_2\times \mathbb Z_2}$ as fusion categories. There are 4 simple objects
\begin{equation}
\label{Z2Z2Bimod}
\begin{split}
        M_0=M_1=M_2=M_3 = \mathbbm1 \oplus e\oplus m\oplus f.
\end{split}
\end{equation} 
$M_0$ is just the algebra itself. $M_1, M_2, M_3$ have fusion rule type of right $[\Z_2\times \Z_2]$ actions, and twisted left $[\Z_2\times \Z_2]$ actions. We write down the twisted type of action on $\ket{\mathbbm 1}$, and the other left actions can be derived from the right actions using the compatible conditions \eqref{CompatibleAction}.
\begin{equation}
    \begin{split}
        \rho_{M_1}: [\Z_2\times \Z_2] &\otimes M_1 \rightarrow M_1 \\
                       \ket{e}&\otimes\ket{\mathbbm 1}\mapsto -\ket{e}\\
                       \ket{f}&\otimes\ket{\mathbbm1}\mapsto -\ket{f}
    \end{split}
\end{equation}
\begin{equation}
    \begin{split}
        \rho_{M_2}: [\Z_2\times \Z_2] &\otimes M_2 \rightarrow M_2 \\
                       \ket{m}&\otimes\ket{\mathbbm 1}\mapsto -\ket{m}\\
                       \ket{f}&\otimes\ket{\mathbbm1}\mapsto -\ket{f}
    \end{split}
\end{equation}

\begin{equation}
    \begin{split}
        \rho_{M_3}: [\Z_2\times \Z_2] &\otimes M_3 \rightarrow M_3\\
                       \ket{e}&\otimes\ket{\mathbbm 1}\mapsto -\ket{e}\\
                       \ket{m}&\otimes\ket{\mathbbm1}\mapsto -\ket{m}
    \end{split}
\end{equation}
Physically, $ M_{0,1,2,3}$ are $4$ point-like defects on the $[\Z_2\times \Z_2]$ wall, and they have the $\mathbb Z_2 \times \mathbb Z_2$ fusion rule along the wall.
\end{subsubsection}

\begin{subsubsection}{$[\Z_2\times \Z_2,\psi]$-bimodule.}
$[\Z_2\times \Z_2,\psi]=(\mathbbm1 \oplus e\oplus m\oplus f)_\psi$, and $_{[\Z_2\times \Z_2,\psi]}\TC_{[\Z_2\times \Z_2,\psi]}\simeq \ve_{\mathbb Z_2\times \mathbb Z_2}$ as fusion categories. There are 4 simple objects 
\begin{equation}
\begin{split}
   M_0^\psi= M^\psi_1= M^\psi_2= M^\psi_3 = \mathbbm1 \oplus e\oplus m\oplus f.
\end{split}
\end{equation}
$M^\psi_0$ is just the algebra itself. $M^\psi_1, M^\psi_2, M^\psi_3$ have fusion rule type of right $[\Z_2\times \Z_2,\psi]$ actions, see table \ref{AlgM2}, and twisted left $[\Z_2\times \Z_2,\psi]$ actions. We spell out the twisted action on $\ket{\mathbbm 1}$ 
\begin{equation}
\begin{split}
    \rho_{M_1^\psi}: [\Z_2\times \Z_2,\psi]&\otimes M_1^\psi \rightarrow M_1^\psi\\
                           \ket{e}&\otimes\ket{\mathbbm 1}\mapsto -\ket{e}\\
                           \ket{f}&\otimes\ket{\mathbbm 1}\mapsto -\ket{f}
\end{split}
\end{equation}
\begin{equation}
\begin{split}
    \rho_{M_2^\psi}: [\Z_2\times \Z_2,\psi]&\otimes M_2^\psi \rightarrow M_2^\psi\\
                           \ket{m}&\otimes\ket{\mathbbm 1}\mapsto -\ket{m}\\
                           \ket{f}&\otimes\ket{\mathbbm 1}\mapsto -\ket{f}
\end{split}
\end{equation}
\begin{equation}
\begin{split}
    \rho_{M_3^\psi}: [\Z_2\times \Z_2,\psi]&\otimes M_3^\psi \rightarrow M_3^\psi\\
                           \ket{e}&\otimes\ket{\mathbbm 1}\mapsto -\ket{e}\\
                           \ket{m}&\otimes\ket{\mathbbm 1}\mapsto -\ket{m}
\end{split}
\end{equation}
Physically, $ M_{0,1,2,3}^\psi$ are $4$ point-like defects on the $[\Z_2\times \Z_2,\psi]$ wall, and they have the $\mathbb Z_2 \times \mathbb Z_2$ fusion rule along the wall.
\end{subsubsection}

We show some examples of other bimodules over two different algebras and list the results.
\begin{subsubsection}{$[\Z_1]$-$[(\Z_2)_f]$-bimodule.}
    $[\Z_1]$--$[(\Z_2)_f]$--bimodule is just a right $[(\Z_2)_f]$-module, there are two simple non-isomorphic right $[(\Z_2)_f]$-module, $\chi_+=\mathbbm 1\oplus f$ and $\chi_-=e\oplus m$ with the fusion rule type of right $[(\Z_2)_f]$-action, so the  $[\Z_1]$-$[(\Z_2)_f]$-bimodule category in $\TC$ is $\ve\oplus \ve$.
\begin{remark}
    As we'll see later in the section \ref{FDW}, $[(\Z_2)_f]=\mathbbm 1\oplus f$ corresponds to the invertible domain wall in toric code called $e$-$m$ exchange wall, and the two simple modules $\chi_+,\chi_-$ are just the dislocation defects at the end of the $e$-$m$ exchange wall\cite{Bombin_2010,Kitaev_2012}. The results are the same as left $[(\Z_2)_f]$-modules. We can then calculate the fusion rule of the dislocations (shrinking a $e$-$m$ exchange wall to a point-like defect) as defined in \eqref{FusionOfPoints}
    \begin{equation}
    \begin{split}
        \chi_{\pm}\ot[\mathbbm 1\oplus f] \chi_{\pm} &= \mathbbm 1\oplus f\\
        \chi_{\pm}\ot[\mathbbm 1\oplus f]\chi_{\mp} & = e\oplus m.
    \end{split}
    \end{equation}   
\end{remark}
\end{subsubsection}
lubricant
\begin{subsubsection}{$[(\Z_2)_e]$-$[(\Z_2)_m]$-bimodule.} 
     To calculate $_{[(\Z_2)_e]}\TC_{[(\Z_2)_m]}$, we decompose free $[(\Z_2)_e]$-$[(\Z_2)_m]$-bimodules $[(\Z_2)_e]\otimes i \otimes [(\Z_2)_m]$ where $i\in \{\mathbbm1,e,m,f\}$. It turns out that these $4$ free bimodules are isomorphic indecomposable bimodules, so $_{[(\Z_2)_e]}\TC_{[(\Z_2)_m]}$ is just $\ve$.
\end{subsubsection}

\begin{subsubsection}{$[(\Z_2)_f]$-$[\Z_2\times \Z_2,\psi]$-bimodule.}
    As the same, we decompose free $[(\Z_2)_f]$-$[\Z_2\times \Z_2,\psi]$-bimodules $[(\Z_2)_f]\otimes i \otimes [\Z_2\times \Z_2,\psi]$, $i\in \{\mathbbm1,e,m,f\}$. For each $i$, there are two non-isomorphic simple sub-bimodules $M_{1,i} = \mathbbm 1\oplus e\oplus m\oplus f$ and $M_{2,i} = \mathbbm 1\oplus e\oplus m\oplus f$. $M_{1,i}$ are isomorphic for different $i$, so are $M_{2,i}$, so $_{[(\Z_2)_f]}\TC_{[\Z_2\times\Z_2,\psi]}\simeq \ve\oplus\ve$.
\end{subsubsection}
\\

All the algebra bimodule categories are listed at table.\ref{Bimod}.
\onecolumngrid
    \begin{center}
        \begin{table}[!h]
    \centering
    \begin{tabular}{|c|c|c|c|c|c|c|} \hline 
 $\text{Bimod}$& $\mathbbm 1$ & $\mathbbm 1\oplus e$ & $\mathbbm 1 \oplus m$ & $\mathbbm 1\oplus f$ & $[\Z_2\times \Z_2]$& $[\Z_2\times \Z_2,\psi]$  \\ \hline
$\mathbbm 1$ & $\ve_{\mathbb Z_2\times \mathbb Z_2}$ &$\ve\oplus \ve$ & $\ve\oplus\ve$ & $\ve\oplus \ve$ & $\ve$ & $\ve$  \\ \hline
 $\mathbbm 1\oplus e$&	$\ve\oplus\ve$ &	$\ve_{\mathbb Z_2\times \mathbb Z_2}$&	$\ve$ & $\ve$ & $\ve\oplus\ve$& $\ve\oplus\ve$\\ \hline
$\mathbbm 1\oplus m$ & $\ve\oplus\ve$& $\ve$ &  $\ve_{\mathbb Z_2\times \mathbb Z_2}$&	$\ve$ & $\ve\oplus\ve$ & $\ve\oplus\ve$\\ \hline
$\mathbbm 1\oplus f$ & $\ve\oplus\ve$ & $\ve$ & $\ve$ &  $\ve_{\mathbb Z_2\times \mathbb Z_2}$ & $\ve\oplus\ve$ & $\ve\oplus\ve$\\ \hline
$[\Z_2\times\Z_2]$ & $\ve$& $\ve\oplus\ve$ & $\ve\oplus\ve$ &	$\ve\oplus\ve$ &  $\ve_{\mathbb Z_2\times \mathbb Z_2}$& $\ve$\\ \hline
$[\Z_2\times \Z_2, \psi]$ & $\ve$ & $\ve\oplus\ve$ & $\ve\oplus\ve$ &	$\ve\oplus\ve$ & $\ve$&  $\ve_{\mathbb Z_2\times \mathbb Z_2}$\\ \hline
    \end{tabular}
    \caption{Algebra bimodule categories in toric code and $\TF$, algebras in rows on the left and columns on the right.}
   \label{Bimod}
\end{table} 
    \end{center}
\twocolumngrid

\subsection{Fusion of Domain Walls}
\label{FDW}
Remember the tensor product of domain walls is given by (\ref{multi}) and (\ref{comulti}), and the isomorphism between two objects in $\Sigma \TC$ means they are Morita equvalent algebra in $\TC$. One can immediately derive the fusion rules like  
\begin{equation}
\label{FusionRule1}
    \begin{split}
&[(\Z_2)_e]\Box [(\Z_2)_m] \cong [\Z_2\times \Z_2];\\
    &[(\Z_2)_m]\Box [(\Z_2)_e] \cong [\Z_2\times \Z_2,\psi]\\
    &[(\Z_2)_e]\Box [(\Z_2)_f] \cong [\Z_2\times \Z_2]\\
     &[(\Z_2)_m]\Box [(\Z_2)_f] \cong [\Z_2\times \Z_2,\psi]\\
\end{split}
\end{equation}
since the
  left and right hand sides are just isomorphic as algebras in $\mathcal{C}$, and the naive map
  \begin{equation}
  \begin{split}
      \phi:[(\Z_2)_i]\Box[(\Z_2)_j]&\rightarrow [\Z_2\times \Z_2]\  \mathrm{or} \ [\Z_2\times \Z_2,\psi], \\
\ket{i}\ket{j}&\mapsto \ket{ij}
  \end{split}
  \end{equation}
 where $(i,j)=\{(e,m),(m,e),(e,f),(f,e)\}$,
  is an algebra isomorphism.

There are also fusion rules like
\begin{equation}
\label{FusionWallTC}
    \begin{split}
        &[(\Z_2)_e]\Box[(\Z_2)_e]\cong [(\Z_2)_e] \oplus [(\Z_2)_e]; \\
        &[(\Z_2)_m]\Box[(\Z_2)_m]\cong [(\Z_2)_m] \oplus [(\Z_2)_m];\\
    \end{split}
\end{equation}
we can solve the condition for an isomorphism between objects in $\TC$ to be an algebra isomorphism, and there are non-zero solutions (for example \eqref{AlgIso}).
  When two algebras $A_1$ $A_2$ are isomorphic to each other with an algebra isomorphism $\phi$, then they are also Morita equivalent and the invertible bimodules between them can just be chosen as themselves. That's why we use $\cong$ rather than $\sim_{M}$ in~\eqref{FusionRule1} and  \eqref{FusionWallTC}.  The action of one on the other is induced by the algebra isomorphism $\phi$, i.e. \begin{equation}
     A_1\otimes A_2 \xrightarrow{\phi^{-1}}A_1\otimes A_1 \xrightarrow{\mu}A_1, 
 \end{equation} 
and 
\begin{equation}
\label{IsoInv1}
    A_1\ot[A_2]A_2\cong A_1
\end{equation}
as $A_1$-bimodule,
\begin{equation}
\label{IsoInv2}
    A_2\ot[A_1]A_1 \cong A_2
\end{equation}
as $A_2$-bimodule.

 While there are other fusion rules like \begin{equation}
 \label{ff=1}
     [(\Z_2)_f]\Box [(\Z_2)_f] \sim_M [\Z_1]=\mathbbm 1;
    \end{equation}    \begin{equation}
        [\Z_2\times \Z_2]\Box [(\Z_2)_f] \sim_M [(\Z_2)_e]
 \end{equation} 
the algebras on the left and right hand sides now are just Morita equivalent. 

Let's take $[(\Z_2)_f]\Box[(\Z_2)_f]$ as an example. The multiplication rule of $[(\Z_2)_f]\Box[(\Z_2)_f]$ are \eqref{multi} \begin{equation}
    \ket{ij}\otimes \ket{kl} \mapsto (-2\delta_{jk}\delta_{jf}+1)\ket{i\ot k,j\ot l}, i,j,k,l = \mathbbm 1, f,
\end{equation} 
where $\ket{\mathbbm 1},\ket{f}$ are the basis of $[(\Z_2)_f]$ as we defined in \ref{ToricCode}.
To prove this fusion rule, we can find a pair of invertible bimodules (see \ref{MoritaEqAlg}) between $ [(\Z_2)_f] \Box [(\Z_2)_f] $ and $\mathbbm 1$.

Let's consider a left $[(\Z_2)_f]\Box[(\Z_2)_f]$-module $M=\mathbbm 1\oplus f$, with action
\begin{equation}
\label{InvBim1f}
    \begin{split}
        \rho_M: [(\Z_2)_f]\Box[(\Z_2)_f]\otimes M&\rightarrow M\\
        \ket{\mathbbm 1\mathbbm 1}\otimes \ket{\mathbbm 1}&\mapsto \ket{\mathbbm 1}\\
        \ket{ff}\otimes \ket{\mathbbm 1}&\mapsto -i\ket{\mathbbm 1}\\
        \ket{\mathbbm 1 f}\otimes \ket{\mathbbm 1}&\mapsto \ket{f}\\
        \ket{f\mathbbm 1}\otimes \ket{\mathbbm 1}&\mapsto i\ket{f}\\
        \ket{\mathbbm 1\mathbbm 1}\otimes \ket{f}&\mapsto \ket{f}\\
          \ket{ff}\otimes \ket{f}&\mapsto i\ket{f}\\
            \ket{\mathbbm 1 f}\otimes \ket{f}&\mapsto \ket{\mathbbm 1}\\
              \ket{f\mathbbm 1}\otimes \ket{f}&\mapsto -i\ket{\mathbbm 1}\\
    \end{split}
\end{equation} and 
a right $[(\Z_2)_f]\Box[(\Z_2)_f]$-module $M^\vee=\mathbbm 1\oplus f$ with action
 \begin{equation}
 \begin{split}
      \tau_{M^\vee}:M^\vee \ot ( [(\Z_2)_f]\Box[(\Z_2)_f]) &\rightarrow M^\vee \\
      \ket{\mathbbm 1}\otimes \ket{\mathbbm 1\mathbbm 1}&\mapsto \ket{\mathbbm 1}\\
      \ket{\mathbbm 1}\otimes \ket{ff}&\mapsto -i\ket{\mathbbm 1}\\
      \ket{\mathbbm 1}\otimes \ket{\mathbbm 1f}&\mapsto \ket{f}\\
      \ket{\mathbbm 1}\otimes \ket{f\mathbbm 1}&\mapsto -i\ket{f}\\
      \ket{f}\otimes\ket{\mathbbm 1\mathbbm 1}&\mapsto \ket{f}\\
      \ket{f}\otimes\ket{ff}&\mapsto i\ket{f}\\
      \ket{f}\otimes\ket{\mathbbm 1f}&\mapsto \ket{\mathbbm 1}\\
      \ket{f}\otimes\ket{f\mathbbm 1}&\mapsto i\ket{\mathbbm 1}\\
 \end{split}
 \end{equation}
 
On the one hand, one can show that \begin{equation}
    M\ot M^\vee \cong [(\Z_2)_f]\Box[(\Z_2)_f]
 \end{equation}
 as $[(\Z_2)_f]\Box[(\Z_2)_f]$-bimodule, and a bimodule isomorphism can be chosen explicitly as \begin{equation}
     \begin{split}
         \ket{\mathbbm 1\mathbbm 1}&\mapsto  \ket{\mathbbm 1\mathbbm 1}+i\ket{ff}\\
         \ket{ff}&\mapsto \ket{\mathbbm 1\mathbbm 1}-\ket{ff}\\
          \ket{\mathbbm 1f}&\mapsto  \ket{\mathbbm 1f}+i\ket{f\mathbbm 1}\\
          \ket{f\mathbbm 1}&\mapsto \ket{\mathbbm 1f}-\ket{f\mathbbm 1}\\
     \end{split}
 \end{equation} 
 which is obviously invertible. 
 
 On the other hand, one can compute the relative tensor product $ M^\vee \ot[{ [(\Z_2)_f]\Box [(\Z_2)_f]} ]  M$ is a 1-dimensional vector space graded by the identity group element, so
  \begin{equation}
    M^\vee \ot[{ [(\Z_2)_f]\Box [(\Z_2)_f]} ]  M\cong \mathbbm 1
 \end{equation}
 as object in $\ve_{\mathbb Z_2 \times \mathbb Z_2}$. So $ [(\Z_2)_f]\Box [(\Z_2)_f]$ and $\mathbbm 1$ are Morita equivalent and then $[(\Z_2)_f]\Box [(\Z_2)_f] \sim_M \mathbbm 1$.

We list all the fusion rule of domain walls of toric code in table.\ref{FusionRuleTC}, we can see $ [(\Z_2)_f]=\mathbbm 1\oplus f$ wall is the only invertible wall from the table. 
\begin{remark}
 According to Theorem 3.6, Proposition 3.7 and Corollary 3.8 of \cite{davydov2013structure} and Theorem 5.20 of \cite{FROHLICH2006192}, given a separable algebra $A$ in a modular tensor category $\cC$ such that the left and right center of $A$, denoted by $C_l(A)$ and $C_r(A)$, are also separable, then the categories of local right modules over $C_l(A)$ and $C_r(A)$ are equivalent as modular tensor categories \begin{equation}
   \phi: \cC_{C_l(A)}^0 \simeq \cC_{C_r(A)}^0.
\end{equation}
We can construct the gapped domain wall corresponding to $A$ in this way, i,e. two original phases and an anyon condensation phase sandwich within a braided autoequivalence of the condensed phase (one can also see more detailed and rigorous discussion in~\cite{xu20242morita}). There are three condensable algebras for the toric code case, $\mathbbm 1, \mathbbm1\oplus e, \mathbbm1\oplus m$. $\mathbbm1\oplus e, \mathbbm1\oplus m$ can be condensed and give rough boundary and smooth boundary of toric code respectively. Since $\ve$ has only the identity autoequivalence, these two Lagrangian algebras can produce $4$ gapped domain walls, which are rough-rough ($[(\Z_2)_e]= \mathbbm1\oplus e$), rough-smooth ($[\Z_2\times \Z_2]= \mathbbm 1\oplus e\oplus m\oplus f$), smooth-smooth ($[(\Z_2)_m]=\mathbbm 1\oplus m$), and smooth-rough ($[\Z_2\times \Z_2,\psi]= (\mathbbm 1\oplus e\oplus m\oplus f)_{\psi}$). If we condensed nothing at both sides, we can have identity and $e$-$m$ exchange as two braided autoequivalence of toric code, such two domain walls are just trivial domain wall and so called $e$-$m$ exchange domain wall corresponding to the separable algebras $\mathbbm 1$ and $\mathbbm 1\oplus f$, respectively.  One can compute the fusion of these domain walls easily with this picture in mind, see \cite{kong2022invitation} for details.\\
\end{remark}
\onecolumngrid
\begin{center}
    \begin{table}[!h]
    \begin{tabular}{|c|c|c|c|c|c|c|} \hline 
 $\Box$& $\mathbbm 1$ & $\mathbbm 1\oplus e$ & $\mathbbm 1 \oplus m$ & $\mathbbm 1\oplus f$ & $[\Z_2\times \Z_2]$& $[\Z_2\times \Z_2,\psi]$  \\ \hline
$\mathbbm 1$ & $\mathbbm 1$ &$\mathbbm 1\oplus e$ & $\mathbbm 1 \oplus m$ & $\mathbbm 1\oplus f$ & $[\Z_2\times \Z_2]$& $[\Z_2\times \Z_2,\psi]$  \\ \hline
 $\mathbbm 1\oplus e$&	$\mathbbm 1\oplus e$ &	$2(\mathbbm 1\oplus e)$ &	$[\Z_2\times \Z_2]$ & $[\Z_2\times \Z_2]$ & $2[\Z_2\times \Z_2]$& $\mathbbm 1\oplus e$\\ \hline
$\mathbbm 1\oplus m$ & $\mathbbm 1\oplus m$& $[\Z_2\times \Z_2,\psi]$ & $2(\mathbbm 1\oplus m)$ &	$[\Z_2\times \Z_2,\psi]$ & $\mathbbm 1\oplus m$ & $2[\Z_2\times \Z_2,\psi]$\\ \hline
$\mathbbm 1\oplus f$ & $\mathbbm 1\oplus f$ & $[\Z_2\times \Z_2,\psi]$ & $[\Z_2\times \Z_2]$ & $\mathbbm 1$ & $\mathbbm 1\oplus m$ & $\mathbbm 1\oplus e$\\ \hline
$[\Z_2\times \Z_2]$ & $[\Z_2\times \Z_2]$ & $\mathbbm 1\oplus e$ & $[\Z_2\times \Z_2]$ &	$\mathbbm1\oplus e$ & $[\Z_2\times \Z_2]$& $2(\mathbbm 1\oplus e)$\\ \hline
$[\Z_2\times \Z_2,\psi]$ & $[\Z_2\times \Z_2,\psi]$ & $2[\Z_2\times \Z_2,\psi]$ & $\mathbbm 1\oplus m$ &	$\mathbbm 1\oplus m$ & $2(\mathbbm 1\oplus m)$& $[\Z_2\times \Z_2,\psi]$\\ \hline
    \end{tabular}
    \caption{Fusion rule of domain walls of toric code, row fuses column.}
   \label{FusionRuleTC}
\end{table} 
\end{center}
\twocolumngrid

\subsection{Fusion of Point like Defects together with Domain Walls }


Let's compute the tensor product of 0d defects on the $[(\Z_2)_e]=\mathbbm 1\oplus e$  wall. There are four simple 0d defects~\eqref{rrWall0dDefect} satisfying $\mathbbm Z_2\times \mathbbm Z_2$ fusion rule along the wall.
The basis of $[(\Z_2)_e]\oplus [(\Z_2)_e]=(\mathbbm 1\oplus e) \oplus (\mathbbm 1\oplus e)$ is chosen to be $\ket{\mathbbm 1}_i,\ket{e}_i$, $i=a,b$. $\ket{\mathbbm 1}_i,\ket{e}_i$ satisfy $\mathbbm Z_2$ multiplication rule, and the multiplications of vectors from different sectors vanish.
The algebric isomorphism $[(\Z_2)_e]\Box [(\Z_2)_e] \cong [(\Z_2)_e]\oplus [(\Z_2)_e]$ can then be chosen as
\begin{equation}
\label{AlgIso}
\begin{split}
\phi:   [(\Z_2)_e]\Box [(\Z_2)_e] &\cong [(\Z_2)_e] \oplus [(\Z_2)_e]\\
    \ket{\mathbbm 1\mathbbm 1}&\mapsto \ket{\mathbbm 1}_a+\ket{\mathbbm 1}_b\\
    \ket{ee}&\mapsto \ket{\mathbbm 1}_a-\ket{\mathbbm 1}_b\\
    \ket{\mathbbm 1 e}&\mapsto \ket{e}_a+\ket{e}_b\\
    \ket{e\mathbbm 1}&\mapsto \ket{e}_a-\ket{e}_b
\end{split}    
\end{equation}

According to the discussion in \ref{FDW}, the invertible bimodules between these two algebras can be just chosen as themselves. 
Now let's compute the fusion of $M_0^e$ and $M^e_1$ based on~\eqref{PointsWithWall}, 
    \begin{equation}
    \begin{split}
        &([(\Z_2)_e]\oplus [(\Z_2)_e]\ot[{[(\Z_2)_e]\Box [(\Z_2)_e]}]M_0^e\Box M_1^e\\
        &\ot[(\mathbbm 1\oplus e)\Box (\mathbbm 1\oplus e)] [(\Z_2)_e]\Box [(\Z_2)_e]  \\\cong&  M_\star\oplus M_\bullet.
    \end{split}
\end{equation}

Where $M_\star = M_\bullet = \mathbbm 1\oplus e$.
We use $\ket{\mathbbm 1}_\star,\ket{e}_\star,\ket{\mathbbm 1}_\bullet,\ket{e}_\bullet$ to denote the basis of $M_\star$ and $M_\bullet$, respectively. The left and right $[(\Z_2)_e]\oplus [(\Z_2)_e]$ actions on $M_\star\oplus M_\bullet$ is very interesting: the left action of $[(\Z_2)_e]_a$ on $M_\bullet$ and $[(\Z_2)_e]_b$ on $M_\star$ vanish; the right action of $[(\Z_2)_e]_a$ on $M_\star$ and $[(\Z_2)_e]_b$ on $M_\bullet$ vanish. The rest are 
\begin{widetext}

\begin{equation}
\begin{split}
    \rho_{M_\star\oplus M_\bullet}: ([(\Z_2)_e]\oplus [(\Z_2)_e])\ot (M_\star\oplus M_\bullet) &\rightarrow M_\star\oplus M_\bullet\\
    \ket{\mathbbm 1}_a\ot \ket{\mathbbm 1}_\star &\mapsto \ket{\mathbbm 1}_\star\\
    \ket{\mathbbm 1}_a\ot \ket{e}_\star &\mapsto \ket{e}_\star\\
    \ket{e}_a\ot \ket{\mathbbm 1}_\star &\mapsto -\ket{e}_\star\\
    \ket{e}_a\ot \ket{e}_\star &\mapsto -\ket{\mathbbm 1}_\star\\
     \ket{\mathbbm 1}_b\ot \ket{\mathbbm 1}_\bullet &\mapsto \ket{\mathbbm 1}_\bullet\\
    \ket{\mathbbm 1}_b\ot \ket{e}_\bullet&\mapsto \ket{e}_\bullet\\
    \ket{e}_b\ot \ket{\mathbbm 1}_\bullet &\mapsto -\ket{e}_\bullet\\
    \ket{e}_b\ot \ket{e}_\bullet &\mapsto -\ket{\mathbbm 1}_\bullet,\\
\end{split}
\end{equation}  
\begin{equation}
\begin{split}
     \tau_{M_\star\oplus M_\bullet}: (M_\star\oplus M_\bullet) \ot ([(\Z_2)_e]\oplus [(\Z_2)_e])&\rightarrow M_\star\oplus M_\bullet\\
\ket{\mathbbm 1}_\bullet \ot \ket{\mathbbm 1}_a&\mapsto \ket{\mathbbm 1}_\bullet\\
\ket{e}_\bullet\ot \ket{\mathbbm 1}_a&\mapsto \ket{e}_\bullet\\
\ket{\mathbbm 1}_\bullet \otimes \ket{e}_a&\mapsto \ket{e}_\bullet\\
\ket{e}_\bullet\ot \ket{e}_a&\mapsto \ket{\mathbbm 1}_\bullet\\
     \ket{\mathbbm 1}_\star \otimes \ket{\mathbbm 1}_b &\mapsto \ket{\mathbbm 1}_\star\\
     \ket{e}_\star \otimes \ket{\mathbbm 1_b}&\mapsto \ket{e}_\star \\
     \ket{\mathbbm 1}_\star \otimes \ket{e}_b &\mapsto \ket{e}_\star\\
     \ket{e}_\star \otimes \ket{e}_b&\mapsto \ket{\mathbbm 1}_\star.
\end{split}
\end{equation}    
Diagrammatically, 
\begin{equation}
        \begin{tikzpicture}[scale=0.5]
		\node [style=none] (0) at (-4, 3) {};
		\node [style=none] (2) at (-4, -3) {};
		\node [style=none] (3) at (-7, 3) {};
		\node [style=none] (4) at (-7, -3) {};
		\node [style=none] (5) at (0, 0) {$\cong$};
		\node [style=none] (7) at (7, 3) {};
		\node [style=none] (8) at (7, -3) {};
		\node [style=none] (9) at (4, 3) {};
		\node [style=none] (10) at (4, -3) {};
		\node [style=none] (11) at (5.5, 0) {$\oplus$};
		\node [style=none] (12) at (-7, -4) {$[(\Z_2)_e]$};
		\node [style=none] (13) at (-4, -4) {$[(\Z_2)_e]$};
		\node [style=none] (14) at (4, -4) {$[(\Z_2)_e]$};
		\node [style=none] (15) at (7, -4) {$[(\Z_2)_e]$};
		\draw (3.center) to (4.center);
		\draw (0.center) to (2.center);
		\draw [style=redline] (9.center) to (10.center);
		\draw [style=blueline] (7.center) to (8.center);
    \end{tikzpicture}
\end{equation}
\begin{equation}
\label{1-e}
    \begin{tikzpicture}[scale=0.5]
		\node [style=none] (0) at (-4, 5) {};
		\node [style=red dot] (1) at (-4, 2) {};
		\node [style=none] (2) at (-4, -1) {};
		\node [style=none] (3) at (-7, 5) {};
		\node [style=none] (5) at (-7, -1) {};
		\node [style=none] (6) at (0, 2) {$\cong$};
		\node [style=dash dot] (7) at (-7, 2) {};
		\node [style=none] (8) at (7, 5) {};
		\node [style=none] (10) at (7, -1) {};
		\node [style=none] (11) at (4, 5) {};
		\node [style=none] (12) at (4, -1) {};
		\node [style=none] (14) at (5.5, 2) {$\oplus$};
		\node [style=none] (23) at (-9, 1.5) {$M^e_0$};
		\node [style=none] (24) at (-2, 1.5) {$M^e_1$};
		\node [style=red dot] (28) at (4, 2) {};
		\node [style=red dot] (29) at (7, 2) {};
		\node [style=none] (30) at (2, 1.5) {$M^e_1$};
		\node [style=none] (31) at (9, 1.5) {$M_1^e$};
		\draw (0.center) to (1);
		\draw (1) to (2.center);
		\draw (3.center) to (7);
		\draw (7) to (5.center);
		\draw [style=redline] (11.center) to (28);
		\draw [style=redline] (29) to (10.center);
		\draw [style=blueline] (28) to (12.center);
		\draw [style=blueline] (8.center) to (29);
\end{tikzpicture}
\end{equation}
\end{widetext}

Since the fusion sub-2-category of $\Sigma\TC$ generated by $[(\Z_2)_e]$ is equivalent to $\Sigma\Rep(\Z_2)$, to understand the mathematical result of the fusion of $M_1^e$ and $M_0^e$, one can take a different physical perspective (see Section~\ref{CatSET}):
$[(\Z_2)_e]=\mathbbm 1\oplus e\cong \Fun(\Z_2)\in \Sigma\Rep(\Z_2)$ can represent the the $\mathbb Z_2$-symmetry breaking $1+1$D phase.  $M_1^e$ is the point-like nontrivial $\Z_2$ symmetry defect. 
Consider two Ising chains in the $\Z_2$ symmetry breaking phase. We use $\ket{\Uparrow}$, $\ket{\Downarrow}$ to denote the ground state of all spin up and all spin down, respectively. Stacking of the two symmetry breaking chain will result in 4 ground states
$\ket{\substack{\Uparrow \\ \Uparrow}}$, $\ket{\substack{\Uparrow \\ \Downarrow}}$,$\ket{\substack{\Downarrow \\ \Uparrow}}$,$\ket{\substack{\Downarrow \\ \Downarrow}}$, they can be further divided into 2 sectors 
\begin{equation}
\begin{split}
   S_1 = \{\ket{\substack{\Uparrow \\ \Uparrow}} \stackrel{X^{\text{diag}}}{\longleftrightarrow} \ket{\substack{\Downarrow\\\Downarrow}} \}\\
    S_2= \{\ket{\substack{\Uparrow \\ \Downarrow}} \stackrel{X^{\text{diag}}}{\longleftrightarrow} \ket{\substack{\Downarrow\\\Uparrow}} \},
\end{split}
\end{equation}
which are both closed under the diagonal $\Z_2$ symmetry operator $X^{\text{diag}}$, and totally break the diagonal $\Z_2$ symmetry. That is $[(\Z_2)_e]\otimes [(\Z_2)_e] \simeq [(\Z_2)_e]\oplus [(\Z_2)_e]$.
 Now if there is a nontrivial symmetry defect $\ket{\Uparrow \Downarrow}$ or $\ket{\Downarrow\Uparrow}$ in one chain, its stacking with a chain $\ket{\Uparrow \Uparrow}$ or $\ket{\Downarrow\Downarrow}$ that is without symmetry defect, leads to four states,
 $
    \ket{\substack{\Uparrow\\\Uparrow}} \ket{\substack{\Downarrow\\\Uparrow}}, \ket{\substack{\Uparrow\\\Downarrow}}\ket{\substack{\Downarrow\\\Downarrow}},\ket{\substack{\Downarrow\\\Uparrow}}\ket{\substack{\Uparrow\\\Uparrow}},\ket{\substack{\Downarrow\\\Downarrow}}\ket{\substack{\Uparrow\\\Downarrow}}. $
 Such four states can then be divided into two types, corresponding to defects between the two sectors $S_1,S_2$:
 \begin{equation}
 \begin{split}
     S_1\textendash S_2:& \{\ket{\substack{\Uparrow\\\Uparrow}}\ket{\substack{\Downarrow\\\Uparrow}}, \ket{\substack{\Downarrow\\\Downarrow}}\ket{\substack{\Uparrow\\\Downarrow}} \}\\
     S_2\textendash S_1:&\{\ket{\substack{\Uparrow\\\Downarrow}}\ket{\substack{\Downarrow\\\Downarrow}},\ket{\substack{\Downarrow\\\Uparrow}}\ket{\substack{\Uparrow\\\Uparrow}}\}
 \end{split}\label{eq.S1S2}
 \end{equation}
 which explains the fusion rule \eqref{1-e}. In this physical picture one can see that in the $S_2$ sector there is an obvious ambiguity between the two symmetry-breaking states: it is an artificial choice to call, say $\ket{\substack{\Uparrow\\\Downarrow}}$, as either the spin up chain or the spin down chain. Depending on this choice, the defects between $S_1$ and $S_2$ \eqref{eq.S1S2} can be considered as either nontrivial or trivial.
 The same ambiguity is also present in the algebraic framework, which is exactly the choice of invertible bimodules (0d defects) between the representative separable algebras.

\section{Example: $\TF$ Topological Order}
$\Vect_{\mathbb Z_2 \times \mathbb Z_2}$ has other modular tensor category structures. We can get one by taking braiding as \begin{equation}
\label{3FBraiding}
    c_{i,j}=\begin{pmatrix}
        1&1&1&1\\
        1&-1&-1&1\\
        1&1&-1&-1\\
        1&-1&1&-1
    \end{pmatrix},
\end{equation}
and then $S,T$ matrices \begin{equation}
    S = \frac{1}{2}\begin{pmatrix}
      1&1&1&1\\
      1&1&-1&-1\\
      1&-1&1&-1\\
      1&-1&-1&1
    \end{pmatrix}
\end{equation}
\begin{equation}
    T = \begin{pmatrix}
        1&&&\\
        &-1&&\\
        &&-1&\\
        &&&-1\\
    \end{pmatrix}
\end{equation}
where the elements in both row and column are in the order of $\mathbbm 1, e,m,f$. 

Notice that $e,m,f$ are all fermions in this topological order, and that's why it's named three fermion (\textbf{3F}) topological order. The topological order is chiral, one can compute the chiral central charge $c$ from the $S,T$ matrices by using
\begin{equation}
    \frac{1}{\sqrt{D}}\sum_{i\in \text{Irr}\cC}d_i^2\theta_a = \exp(2\pi i c/8),
\end{equation}
and $c = 4 $ mod $8$ for this system. Here $d_i$ is the quantum dimension for particle $i$, $\theta_i = T_{i,i}$, $D = \sqrt{\sum_{i\in\text{Irr}\cC}d_i^2}$ is the total quantum dimension.   Since $\textbf{3F}$ has the same underlying fusion category as toric code, the separable algebras and bimodules in their condensation completion  are the same, see table~\ref{Bimod}. We'll use the same notation for algebras as in the toric code.

\subsection{Fusion of Domain Walls}

Due to the three fermions in $\TF$,  there are much more efficient ways to find the algebraic properties than the condensation completion algorithm,  we'll discuss it in \ref{3Fapp}.

Due to the special braiding statistics of the three fermions $e,m,f$, $\TF$ has a $S_3$ permutation symmetry that permutes $e,m,f$. 
According to \cite{etingof2009fusioncategorieshomotopytheory}, when $\cC$ is a modular tensor category, invertible objects in $\Sigma\cC$ are one-to-one corresponding to braided autoequivalences in $\cC$. The six domain walls $[\Z_1],...[\Z_2\times \Z_2], [\Z_2\times \Z_2,\psi]$ then must be invertible and correspond to the six group elements of $S_3$, and their fusion rule is just $S_3$ multiplication. The only issue remaining is how to identify these six walls with $S_3$ elements. Since this is not directly related to the condensation completion algorithm, we show the calculation in \ref{3Fapp}, and leave the result below.
\begin{enumerate}
    \item $[\Z_2]_e$ exchanges $m$ and $f$,
    \item $[\Z_2]_m$ exchanges $e$ and $f$,
    \item $[\Z_2]_f$ exchanges $e$ and $m$,
    \item  \begin{equation}
    [\Z_2\times \Z_2]:\begin{tikzcd}
	e && m \\
	& f
	\arrow[from=1-1, to=2-2]
	\arrow[from=2-2, to=1-3]
	\arrow[from=1-3, to=1-1]
\end{tikzcd}.
\end{equation}
\item 
\begin{equation}
\label{3Fpermutation}
    [\Z_2\times\Z_2, \psi]:\begin{tikzcd}
	e && m \\
	& f
	\arrow[from=1-1, to=1-3]
	\arrow[from=1-3, to=2-2]
	\arrow[from=2-2, to=1-1]
\end{tikzcd}.
\end{equation}
\end{enumerate}
After identifying these six walls with $S_3$ element, their fusion rules are just $S_3$ multiplications, see table~\ref{FusionRule3F}.
\onecolumngrid
\begin{center}
     \begin{table}
    \begin{tabular}{|c|c|c|c|c|c|c|} \hline 
 $\Box$& $\mathbbm 1$ & $\mathbbm 1\oplus e$ & $\mathbbm 1 \oplus m$ & $\mathbbm 1\oplus f$ & $[\Z_2\times \Z_2]$& $[\Z_2\times \Z_2,\psi]$  \\ \hline

$\mathbbm 1$ & $\mathbbm 1$ &$\mathbbm 1\oplus e$ & $\mathbbm 1 \oplus m$ & $\mathbbm 1\oplus f$ & $[\Z_2\times \Z_2]$& $[\Z_2\times \Z_2,\psi]$  \\ \hline

 $\mathbbm 1\oplus e$&	$\mathbbm 1\oplus e$ &	$\mathbbm 1$ &	$[\Z_2\times \Z_2]$ & $[\Z_2\times \Z_2,\psi]$ & $\mathbbm 1\oplus m$& $\mathbbm 1\oplus f$\\ \hline

$\mathbbm 1\oplus m$ & $\mathbbm 1\oplus m$& $[\Z_2\times \Z_2,\psi]$ & $\mathbbm 1$ &	$[\Z_2\times \Z_2]$ & $\mathbbm 1\oplus f$ & $\mathbbm 1\oplus e$\\ \hline

$\mathbbm 1\oplus f$ & $\mathbbm 1\oplus f$ & $[\Z_2\times \Z_2]$ & $[\Z_2\times \Z_2,\psi]$ & $\mathbbm 1$ & $\mathbbm 1\oplus e$ & $\mathbbm 1\oplus m$\\ \hline

$[\Z_2\times \Z_2]$ & $[\Z_2\times \Z_2]$ & $\mathbbm 1\oplus f$ & $\mathbbm 1\oplus e$ &	$\mathbbm1\oplus m$ & $[\Z_2\times \Z_2,\psi]$& $\mathbbm 1$\\ \hline

$[\Z_2\times \Z_2,\psi]$ & $[\Z_2\times \Z_2,\psi]$ & $\mathbbm 1\oplus m$ & $\mathbbm 1\oplus f$ &	$\mathbbm 1\oplus e$ & $\mathbbm 1$& $[\Z_2\times \Z_2]$\\ \hline
    \end{tabular}                                                                                                       
    \caption{Fusion rule of domain walls of $\TF$, row fuses column.}
   \label{FusionRule3F}
\end{table} 
\end{center}
   
\twocolumngrid

\subsection{Fusion of Point-like defects}
As mentioned in~\cite{Cui_2019}, for every $G$-crossed braided fusion category $\cC_{G}^\times$, we can construct a fusion 2-category $\cD(\cC_{G}^\times)$ that encodes the data of $\cC_{G}^\times$. In short, objects in $\cD(\cC_{G}^\times)$ are group elements in $G$. For $g,h\in G$, $\Hom(g,h) = (\cC_{G}^\times)_{g^{-1}h}$. The fusion structure of $\cD(\cC_{G}^\times)$ depends on the $G$-crossed braiding of $\cC_G^\times$ (see~\cite{Cui_2019} for details). Such fusion 2-category can be viewed as a subcategory of $\Sigma\cC$, restricted on all $G$ invertible walls. In the $\TF$ case, all domain walls are $S_3$ symmetry wall, so $\Sigma\TF$ and $\TF_{S_3}^\times$ are equivalent description of codimension-1 and higher symmetry defects. 
  The fusion rule of point-like defects in  $\TF$ with the $S_3$ symmetry enrichment has been studied in~\cite{Roberts_2024,Barkeshli_2019} as $\TF_{S_3}^\times$, we refer readers to these references for details.

\section{Example: Two-layer Semion}
Particles in $2+1$D semion topological order form the modular tensor category $\SM = \Vect_{\mathbb Z_2}^{\omega_s}$, where the only nontrivial value of $\omega_s$ is $\omega_s(1,1,1) = -1$. There are two simple anyons $\mathbbm 1$ and a semion $s$, with $\mathbb Z_2$ fusion rule $s\otimes s = \mathbbm 1$. The only nontrival braiding is $c_{s,s}=i$, and then $S$, $T$ matrices are
\begin{equation*}
  S = \frac{1}{\sqrt{2}}\begin{pmatrix}
        1&1\\
        1&-1\\
    \end{pmatrix}, T = \begin{pmatrix}
        1& \\
         &i\\
    \end{pmatrix}  
\end{equation*}
We can stack $\cC$ and $\cC$ itself together and get a new modular tensor category $\TS=\SM\boxtimes\SM$ describing the particles in the two-layer semion topological order. There are $4$ simple anyons $\mathbbm 1,s_1=\mathbbm 1\boxtimes s,s_2=s\boxtimes \mathbbm 1$ and $\psi= s\boxtimes s$ satisfying $\mathbb Z_2\times \mathbb Z_2$ fusion rule. The braiding in a stacked topological order is the tensor product of braidings in each layer,
   \begin{equation}
     c_{i,j}^\TS=  \begin{pmatrix}
           1&1&1&1\\
           1&i&1&i\\
           1&1&i&i\\
           1&i&i&-1
       \end{pmatrix}
   \end{equation}
   so are the  the $S,T$ matrices, 

\begin{equation}
    S_\TS = \frac{1}{2}\begin{pmatrix}
        1&1&1&1\\
        1&-1&1&-1\\
        1&1&-1&-1\\
        1&-1&-1&1
    \end{pmatrix}
\end{equation}

\begin{equation}
    T_{\TS} = \begin{pmatrix}
        1&\ &\ &\\
        \ &i&\ &\ \\
        \ &\ &i&\ \\
        \ &\ &\ &-1\\
    \end{pmatrix}
    \vspace{30pt}
\end{equation}
where the row and column are in the order of $\mathbbm 1, s_1, s_2, \psi$. Note that $\psi$ is a fermion with topological spin $-1$.
 $\TS$ is equivalent to the modular tensor category $\Vect_{\mathbb Z_2\times \mathbb Z_2}^{\omega_\TS,c^{\TS}}$, where 
 \begin{equation}
 \label{2semion3co}
     \omega_\TS(a_1\boxtimes a_2, b_1\boxtimes b_2, c_1\boxtimes c_2) = \omega_s(a_1,b_1,c_1)\omega_s(a_2,b_2,c_2),
 \end{equation}
then $\mathbbm 1,s_1, s_2, \psi$ can be identified as 1-dimensional vector spaces graded by $1, a, b, c$ (see section \ref{ToricCode} for the convention), respectively.  Such modular tensor category is chiral and the chiral central charge is $c =2$ mod $8$. 

\begin{subsection}{Separable Algebras}

In this case, we denote the subgroups of $\Z_2\times \Z_2$ in the following way, $\Z_1=\{1\}$, $(\Z_2)_{s_1}=\{1,a\}$, $(\Z_2)_{s_2}=\{1,b\}$,
 $(\Z_2)_{\psi}=\{1,c\}$.
According to the 3-cocycle \eqref{2semion3co},  we have  $\omega_\TS|_{\Z_1}=1, \omega_\TS|_{(\Z_2)_{\psi}}=1$, there are only two separable algebras, $[\Z_1] = \mathbbm 1$, $[(\Z_2)_\psi] = \mathbbm 1\oplus \psi$. We choose the basis vector in $\mathbbm 1, s_1, s_2, \psi$ as $\ket{\mathbbm 1},\ket{s_1},\ket{s_2},\ket{\psi}$, such that the multiplication rule of $[(\Z_2)_\psi]$ has the same form as the fusion rule, as shown in table~\ref{DSAlg2M}. $[\Z_1]$ is just the subalgebra of $[(\Z_2)_{\psi}]$.

\begin{table}[!h]
    \centering
    \scalebox{1}{
    \begin{tabular}{c c c}  
        & $|\mathbbm1\rangle$ & $|\psi\rangle$  \\ 
  \hline
$|\mathbbm{1}\rangle$ & $|\mathbbm{1}\rangle$ & $\ket{\psi}$ \\  
 \hline
 $\ket{\psi}$ & $\ket{\psi}$ & $|\mathbbm1\rangle$
    \end{tabular}}
    \caption{Multiplications of $[(\Z_2)_\psi]$ in two-layer semion, row times column}
    \label{DSAlg2M}
\end{table}
\end{subsection}

\begin{subsection}{Bimodules Over Algebras}
\begin{subsubsection}{$[\Z_1]$-Bimodule}
$[\Z_1]$-bimodule category ${}_{[\Z_1]}\TS_{[\Z_1]}\simeq \TS$, with monoidal structure just the tensor product of $\TS$.
\end{subsubsection}
\begin{subsubsection}{$[\Z_1]$-$[(\Z_2)_\psi]$-Bimodule}
An $[\Z_1]$-$[(\Z_2)_\psi]$-bimodule is just a right $[(\Z_2)_\psi]$-module. There are two non-isomorphic simple right $[(\Z_2)_\psi]$-modules in $\TS$,
\begin{equation*}
    M_0=\mathbbm 1 \oplus \psi, M_1= s_1\oplus s_2
\end{equation*}

$M_0$ is just the algebra itself. For $M_1$, it's a free module $s_1\otimes (\mathbbm 1\oplus\psi)$, with right action 

\begin{equation}
    \begin{split}
        &\tau: (s_1\otimes (\mathbbm 1\oplus\psi))\otimes (\mathbbm 1\oplus\psi) \\&\xrightarrow{\omega_\TS} s_1\otimes ((\mathbbm 1\oplus\psi)\otimes (\mathbbm 1\oplus\psi))\rightarrow s_1\otimes (\mathbbm 1\oplus\psi)\\
    (\ket{s_1}\ket{\psi})\ket{\psi}&\mapsto -\ket{s_1}(\ket{\psi}\ket{\psi})\mapsto -\ket{s_1}\ket{\mathbbm 1}\\
    \end{split}
\end{equation}

where the minus sign comes from the nontrivial associator 
\begin{equation}
    \omega_\TS(s_1,\psi,\psi) = \omega_s(\mathbbm 1,s,s)\omega_s(s,s,s) = -1,
\end{equation}
so  
\begin{equation}
\begin{split}
    \tau_{M_1}:M_1 \otimes [(\Z_2)_f]&\rightarrow M_1\\
    \ket{s_2}\otimes \ket{\psi}&\mapsto -\ket{s_1}\\
\end{split}
\end{equation}

and all other actions are of fusion rule type. We then have $\TS_{[(\Z_2)_{\psi}]}\simeq \ve\oplus\ve$, and similar for. Physically, they are the point-like defects living on the ends of the $\mathbbm 1\oplus \psi$ wall.
\end{subsubsection}

\begin{subsubsection}{$[(\Z_2)_\psi]$-Bimodule}
There are four $[(\Z_2)_\psi]$-bimodules 
\begin{equation*}
   M^\psi_0= \mathbbm 1\oplus \psi, M^\psi_1 =  \mathbbm 1\oplus \psi, N^\psi_0 =s_1\oplus s_2,N^\psi_1 =  s_1\oplus s_2
\end{equation*}
$M^\psi_0$ is the algebra itself. $M^\psi_1$ has the fusion rule type of right action, and the left action is twisted 
\begin{equation}
\begin{split}
     \rho_{M^\psi_1} : [(\Z_2)_\psi] \otimes M^\psi_1 &\rightarrow M^\psi_1\\
     \ket{\psi}\otimes \ket{\psi}&\mapsto -\ket{\mathbbm 1}\\
     \ket{\psi}\otimes \ket{\mathbbm 1} &\mapsto -\ket{\psi}.
\end{split}
\end{equation}
$N_0^\psi$ and $N_1^\psi$ are from the decomposition of the free bimodule $[(\Z_2)_\psi]\otimes s_1) \otimes [(\Z_2)_\psi]$. In this convention, the left action is 

\begin{equation}
\begin{split}
    \rho:& [(\Z_2)_\psi]\otimes ([(\Z_2)_\psi]\otimes s_1) \otimes [(\Z_2)_\psi])\\&\xrightarrow{\omega_\TS} [(\Z_2)_\psi]\otimes( [(\Z_2)_\psi]\otimes(s_1\otimes [(\Z_2)_\psi]))\\
    &\xrightarrow{\omega_\TS} ([(\Z_2)_\psi]\otimes [(\Z_2)_\psi])\otimes (s_1\otimes [(\Z_2)_\psi])\\
    &\rightarrow [(\Z_2)_\psi]\otimes (s_1\otimes [(\Z_2)_\psi])\\
    &\xrightarrow{\omega_\TS} ([(\Z_2)_\psi]\otimes s_1) \otimes[(\Z_2)_\psi]
\end{split}
\end{equation}

where the third line is the  $[(\Z_2)_\psi]$ multiplication and the other lines are generally nontrivial associators. The right action is 

\begin{equation}
    \begin{split}
        \tau:  &([(\Z_2)_\psi]\otimes s_1) \otimes [(\Z_2)_\psi])\otimes [(\Z_2)_\psi]\\&\xrightarrow{\omega_\TS} ([(\Z_2)_\psi]\otimes s_1)\otimes ([(\Z_2)_\psi]\otimes [(\Z_2)_\psi])\\
        &\rightarrow ([(\Z_2)_\psi]\otimes s_1) \otimes [(\Z_2)_\psi]
    \end{split}
\end{equation}

We then find the following two sub-bimodules
\begin{widetext}
\begin{equation}
\begin{split}
    N^\psi_0 &= \langle\ket{s_1},\ket{s_2}\rangle= \langle (\ket{\mathbbm 1}\ket{s_1})\ket{\mathbbm 1}+i(\ket{\psi }\ket{s_1})\ket{\psi},  (\ket{\mathbbm 1}\ket{s_1})\ket{\psi}-i(\ket{\psi }\ket{s_1})\ket{\mathbbm 1} \rangle,\\
     N^\psi_1 & = \langle\ket{s_1},\ket{s_2}\rangle=\langle (\ket{\mathbbm 1}\ket{s_1})\ket{\mathbbm 1}-i(\ket{\psi }\ket{s_1})\ket{\psi},  (\ket{\mathbbm 1}\ket{s_1})\ket{\psi}+i(\ket{\psi }\ket{s_1})\ket{\mathbbm 1} \rangle.
    \end{split}
\end{equation}    
\end{widetext}

More explicitly, for $(s_1\oplus s_2)_0$ we have the left action
\begin{equation}
    \begin{split}
        \rho_{N^\psi_0}: [(\Z_2)_\psi] \otimes N^\psi_0&\rightarrow N^\psi_0\\
        \ket{\psi}\otimes \ket{s_1}&\mapsto i\ket{s_2}\\
        \ket{\psi}\otimes \ket{s_2}&\mapsto i\ket{s_1} 
    \end{split}
\end{equation}
and the right action
\begin{equation}
    \begin{split}
        \tau_{N^\psi_0}: N^\psi_0\otimes [(\Z_2)_\psi]&\rightarrow N^\psi_0\\
        \ket{s_1}\otimes\ket{\psi}&\mapsto \ket{s_2}\\
        \ket{s_2}\otimes \ket{\psi} &\mapsto -\ket{s_1}.
    \end{split}
\end{equation}
For $N^\psi_1$ we have the left action 
\begin{equation}
    \begin{split}
        \rho_{N^\psi_1}: [(\Z_2)_\psi] \otimes N^\psi_1 &\rightarrow N^\psi_1\\
        \ket{\psi}\otimes \ket{s_1}&\mapsto -i\ket{s_2}\\
        \ket{\psi}\otimes \ket{s_2}&\mapsto -i\ket{s_1} 
    \end{split}
\end{equation}
and the right action
\begin{equation}
    \begin{split}
        \tau_{N^\psi_1}: N^\psi_1\otimes [(\Z_2)_\psi] &\rightarrow N^\psi_1\\
        \ket{s_1}\otimes\ket{\psi}&\mapsto \ket{s_2}\\
        \ket{s_2}\otimes \ket{\psi} &\mapsto -\ket{s_1}.\\
    \end{split}
\end{equation}

These four simple bimodules satisfy the  $\mathbb Z_2 \times \mathbb Z_2$ fusion rule, as fusion category we must have  ${}_{[(\Z_2)_\psi]}\TS_{[(\Z_2)_\psi]}\simeq \TS$. Physically, $  M^\psi_0, M^\psi_1, N^\psi_0, N^\psi_1$ are the four simple point-like defects living on the $[(\Z_2)_\psi]$ wall, and they have the  $\mathbb Z_2 \times \mathbb Z_2$ fusion rule along the wall.
We list all the bimodule categories in table \ref{BimodTwoSemion}.
\begin{table}
    \centering
    \begin{tabular}{|c|c|c|} \hline 
 $\text{Bimod}$& $\mathbbm 1$ & $\mathbbm 1\oplus \psi$  \\ \hline
$\mathbbm 1$ & $\ve_{\mathbb Z_2\times \mathbb Z_2}^{\omega_{\TS}}$ &$\ve\oplus \ve$  \\ \hline
 $\mathbbm 1\oplus \psi$&	$\ve\oplus\ve$ &	$\ve_{\mathbb Z_2\times \mathbb Z_2}^{\omega_{\TS}}$	\\ \hline
    \end{tabular}
    \caption{Bimodule categories, algebras in rows on the left and columns on the right.}
   \label{BimodTwoSemion}
\end{table}
\end{subsubsection}
\end{subsection}

\begin{subsection}{Fusion of Domain Walls}
The only nontrivial fusion rule of domain walls in doubled semion is \begin{equation}
\label{TwoSemionWallFusion}
[(\Z_2)_\psi]\Box[(\Z_2)_\psi] \sim_M \mathbbm 1.
\end{equation}  
Since the topological order is made of two layers of semion topological order, so the nontrivial wall should be the exchange of the two layers, exchange two times must be the trivial wall.
Graphically, 

\begin{tikzpicture}[scale=0.4]
		\node [style=none] (0) at (2, 2) {};
		\node [style=none] (1) at (2, 0) {};
		\node [style=none] (2) at (8, 2) {};
		\node [style=none] (3) at (8, 0) {};
		\node [style=none] (4) at (-2, 1) {$\mathbbm{1}:=$};
		\node [style=none] (5) at (9, 2) {$\SM$};
		\node [style=none] (6) at (9, 0) {$\SM$};
		\node [style=none] (7) at (-2, -3) {$\mathbbm1\oplus \psi:=$};
		\node [style=none] (8) at (2, -2) {};
		\node [style=none] (9) at (2, -4) {};
		\node [style=none] (10) at (6, -2) {};
		\node [style=none] (11) at (6, -4) {};
		\node [style=none] (12) at (4, -2) {};
		\node [style=none] (13) at (4, -4) {};
		\node [style=none] (14) at (8, -2) {};
		\node [style=none] (15) at (8, -4) {};
		\node [style=none] (18) at (9, -2) {$\SM$};
		\node [style=none] (19) at (9, -4) {$\SM$};
		\node [style=none] (20) at (2, -6) {};
		\node [style=none] (21) at (2, -8) {};
		\node [style=none] (22) at (3.5, -6) {};
		\node [style=none] (23) at (3.5, -8) {};
		\node [style=none] (24) at (3, -6) {};
		\node [style=none] (25) at (3, -8) {};
		\node [style=none] (26) at (5, -6) {};
		\node [style=none] (27) at (5, -8) {};
		\node [style=none] (28) at (5, -6) {};
		\node [style=none] (29) at (5, -8) {};
		\node [style=none] (30) at (7, -6) {};
		\node [style=none] (31) at (7, -8) {};
		\node [style=none] (32) at (6.5, -6) {};
		\node [style=none] (33) at (6.5, -8) {};
		\node [style=none] (34) at (8, -6) {};
		\node [style=none] (35) at (8, -8) {};
		\node [style=none] (37) at (10, -6) {};
		\node [style=none] (38) at (10, -8) {};
		\node [style=none] (39) at (13, -6) {};
		\node [style=none] (40) at (13, -8) {};
		\node [style=none] (41) at (-2, -7) {$(\mathbbm1\oplus \psi)\Box(\mathbbm1\oplus \psi):=$};
            \node[style=none] (42) at (0,-11) {$=$};
            \node [style=none] (43) at (2, -10) {};
		\node [style=none] (44) at (2, -12) {};
            \node [style=none] (45) at (8, -10) {};
		\node [style=none] (46) at (8, -12) {};
            \node [style=none] (47) at (9, -6) {$\SM$};
		\node [style=none] (48) at (9, -8) {$\SM$};
            \node [style=none] (49) at (9, -10) {$\SM$};
		\node [style=none] (50) at (9, -12) {$\SM$};
		\draw (0.center) to (2.center);
		\draw (1.center) to (3.center);
		\draw (8.center) to (12.center);
		\draw (12.center) to (11.center);
		\draw (9.center) to (13.center);
		\draw (13.center) to (10.center);
		\draw (10.center) to (14.center);
		\draw (11.center) to (15.center);
		\draw (20.center) to (24.center);
		\draw (24.center) to (23.center);
		\draw (21.center) to (25.center);
		\draw (25.center) to (22.center);
		\draw (22.center) to (26.center);
		\draw (23.center) to (27.center);
		\draw (28.center) to (32.center);
		\draw (32.center) to (31.center);
		\draw (29.center) to (33.center);
		\draw (33.center) to (30.center);
		\draw (30.center) to (34.center);
		\draw (31.center) to (35.center);
            \draw (43.center) to (45.center);
            \draw (44.center) to (46.center);
\end{tikzpicture}

From the algebraic point of view, the braided fusion subcategory of toric code generated by $\{\mathbbm 1, f\}$ is equivalent to the braided fusion subcategory of two-layer semion generated by $\{\mathbbm 1, \psi\}$, \eqref{TwoSemionWallFusion} is then from \eqref{ff=1}.

\begin{remark}
For a topological order $\cC$, We define the time reversal conjugate of $\cC$ as the same fusion category but with $\bar{c}_{a,b}:= c_{b,a}^{-1}$, and denote it as $\overline{\cC}$. For the semion topological order $\SM$, then 
$\overline{\SM}$ is the anti-semion topological order with $\bar{c}_{s,s}=-i$. The double layer anti-semion topological order
$\overline{\SM}\boxtimes\overline{\SM}$ is also a $\mathbb Z_2 \times \mathbb Z_2$ topological order with chiral central charge $c=6$ mod $8$. Similar to the two-layer semion, $\psi=s\boxtimes s$ is a fermion, and $\mathbbm 1\oplus \psi$ is only one nontrivial domain wall which is also invertible. 

$\TC$, $\TS$, $\TF$, and $\overline{\SM}\overline{\SM}$ are the four $2+1$D topological orders with $\mathbb Z_2\times \mathbb Z_2$ fusion rule that come from gauging fermion parity symmetry of the $0,4,8,12$ mod $16$ layers of $p+ip$ topological superconductors in the $16$-fold way~\cite{KITAEV20062} with chiral central charge $0,2,4,6$, respectively. 
\end{remark}

\begin{remark}
If we stack the semion topological $\cC$ with $\overline{\cC}$, we get the double semion topological order $\DS=\cC\boxtimes\overline{\cC}=Z_1(\cC)$, it's another kind of non-chiral topological order with $\mathbb Z_2\times \mathbb Z_2$ fusion rule. $p=s\boxtimes s$ is now a boson. $\mathbbm 1\oplus p$ is the only one nontrivial domain wall satisfying the fusion rule 
\begin{equation}
    \label{dsWallFusion}
    (\mathbbm 1\oplus p)\Box(\mathbbm 1\oplus p) \cong (\mathbbm 1\oplus p)\oplus (\mathbbm 1\oplus p).
\end{equation}
Graphically,

\begin{tikzpicture}[scale=0.4]
		\node [style=none] (0) at (4.5, 2) {};
		\node [style=none] (1) at (4.5, 0) {};
		\node [style=none] (2) at (10.5, 2) {};
		\node [style=none] (3) at (10.5, 0) {};
		\node [style=none] (4) at (0.5, 1) {$\mathbbm1:=$};
		\node [style=none] (5) at (11.5, 2) {$\mathcal{C}$};
		\node [style=none] (6) at (11.5, 0) {$\Bar{\mathcal{C}}$};
		\node [style=none] (7) at (0.5, -3) {$\mathbbm1\oplus p:=$};
		\node [style=none] (8) at (4.5, -2) {};
		\node [style=none] (9) at (4.5, -4) {};
		\node [style=none] (10) at (8.5, -2) {};
		\node [style=none] (11) at (8.5, -4) {};
		\node [style=none] (12) at (6.5, -2) {};
		\node [style=none] (13) at (6.5, -4) {};
		\node [style=none] (14) at (10.5, -2) {};
		\node [style=none] (15) at (10.5, -4) {};
		\node [style=none] (16) at (11.5, -2) {$\mathcal{C}$};
		\node [style=none] (17) at (11.5, -4) {$\Bar{\mathcal{C}}$};
		\node [style=none] (39) at (0.5, -7) {$(\mathbbm1\oplus p)\Box(\mathbbm1\oplus p):=$};
		\node [style=none] (40) at (4.5, -6) {};
		\node [style=none] (41) at (4.5, -8) {};
		\node [style=none] (42) at (8.5, -6) {};
		\node [style=none] (43) at (8.5, -8) {};
		\node [style=none] (44) at (6.5, -6) {};
		\node [style=none] (45) at (6.5, -8) {};
		\node [style=none] (46) at (10.5, -6) {};
		\node [style=none] (47) at (10.5, -8) {};
		\node [style=none] (48) at (10.5, -6) {};
		\node [style=none] (49) at (10.5, -8) {};
		\node [style=none] (50) at (14.5, -6) {};
		\node [style=none] (51) at (14.5, -8) {};
		\node [style=none] (52) at (12.5, -6) {};
		\node [style=none] (53) at (12.5, -8) {};
		\node [style=none] (54) at (16.5, -6) {};
		\node [style=none] (55) at (16.5, -8) {};
		\node [style=none] (56) at (4.5, -10) {};
		\node [style=none] (57) at (4.5, -12) {};
		\node [style=none] (58) at (8.5, -10) {};
		\node [style=none] (59) at (8.5, -12) {};
		\node [style=none] (60) at (6.5, -10) {};
		\node [style=none] (61) at (6.5, -12) {};
		\node [style=none] (62) at (10.5, -10) {};
		\node [style=none] (63) at (10.5, -12) {};
		\node [style=none] (64) at (11.5, -10) {};
		\node [style=none] (65) at (11.5, -12) {};
		\node [style=none] (66) at (15.5, -10) {};
		\node [style=none] (67) at (15.5, -12) {};
		\node [style=none] (68) at (13.5, -10) {};
		\node [style=none] (69) at (13.5, -12) {};
		\node [style=none] (70) at (17.5, -10) {};
		\node [style=none] (71) at (17.5, -12) {};
		\node [style=none] (72) at (11, -11) {$\oplus$};
            \node [style=none] (73) at (2.5, -11) {$=$};
            \node [style=none] (74) at (17.5, -6) {$\cC$};
            \node [style=none] (75) at (17.5, -8) {$\Bar{\cC}$};
            \node [style=none] (76) at (18, -10) {$\cC$};
            \node [style=none] (77) at (18, -12) {$\Bar{\cC}$};
		\draw (0.center) to (2.center);
		\draw (1.center) to (3.center);
		\draw (8.center) to (12.center);
		\draw (9.center) to (13.center);
		\draw (10.center) to (14.center);
		\draw (11.center) to (15.center);
		\draw [bend left=45, looseness=1.25] (12.center) to (13.center);
		\draw [bend right=45, looseness=1.25] (10.center) to (11.center);
		\draw (40.center) to (44.center);
		\draw (41.center) to (45.center);
		\draw (42.center) to (46.center);
		\draw (43.center) to (47.center);
		\draw [bend left=45, looseness=1.25] (44.center) to (45.center);
		\draw [bend right=45, looseness=1.25] (42.center) to (43.center);
		\draw (48.center) to (52.center);
		\draw (49.center) to (53.center);
		\draw (50.center) to (54.center);
		\draw (51.center) to (55.center);
		\draw [bend left=45, looseness=1.25] (52.center) to (53.center);
		\draw [bend right=45, looseness=1.25] (50.center) to (51.center);
		\draw (56.center) to (60.center);
		\draw (57.center) to (61.center);
		\draw (58.center) to (62.center);
		\draw (59.center) to (63.center);
		\draw [bend left=45, looseness=1.25] (60.center) to (61.center);
		\draw [bend right=45, looseness=1.25] (58.center) to (59.center);
		\draw (64.center) to (68.center);
		\draw (65.center) to (69.center);
		\draw (66.center) to (70.center);
		\draw (67.center) to (71.center);
		\draw [bend left=45, looseness=1.25] (68.center) to (69.center);
		\draw [bend right=45, looseness=1.25] (66.center) to (67.center);
\end{tikzpicture}

\end{remark}
\end{subsection}

\section{Example: $\mathbb{Z}_4$ Topological Orders}
Let's now consider a topological order with $\mathbb{Z}_4$ fusion rule. 
$\mathbb Z_4=\{0,1,2,3\}$. 
   It can be proved that only the category $\text{Vec}_{\mathbb Z_4}^\omega$ has modular structure\footnote{As a fusion category, $\text{Vec}_{\mathbb Z_4}$ associator can take anyone in $\text{H}^3(\mathbb Z_4, U(1)) = \mathbb Z_4$, with explict formula $\omega_{p}(a,b,c) = \exp(\frac{2\pi i p}{4^2}a(b+c-[b+c]))$, and $p,a,b,c\in \mathbb Z_4$. It turns out that only $\omega_0$ and $\omega_2$ can be extended to a braided fusion category, and only $\omega_2$ can be extended to a modular tensor category.}, where \begin{equation}\label{3cocycle}
    \omega(a,b,c) = \exp(\frac{i\pi}{4}a(b+c-[b+c])),
\end{equation}
$a,b,c\in \mathbb Z_4$, $[b+c] = (b+c)$ \text{mod} $4$.
We'll denote the simple objects in $\text{Vec}_{\mathbb Z_4}^\omega$ as $\mathbbm1, \sigma_1, \sigma_2, \sigma_3$, which are 1-dimension vector spaces  graded by $0,1,2,3$, respectively. Such category has $4$ different modular structures, corresponding to $2,6,10,14$ mod $16$ layers of $p+ip$ topological superconductor in the $16$-fold way. We'll focus on the 2-layer one with chiral central charge $c=1$ mod $8$, and the $S, T$ matrices are
\begin{equation}
    S = \frac{1}{2}\begin{pmatrix}
        1&1&1&1\\
        1&-i&-1&i\\
        1&-1&1&-1\\
        1&i&-1&-i
    \end{pmatrix}
\end{equation}
\begin{equation}
    T = \begin{pmatrix}
        1&&&\\
        &e^{i\pi/4}&&\\
        &&-1&\\
        &&&e^{i\pi/4}\\
    \end{pmatrix}
\end{equation}
and the convention of braiding $c_{a,b}$ we use are \begin{equation}
\label{braiding}
    \begin{pmatrix}
        1&1&1&1\\
1&e^{i\pi/4}&i&e^{i3\pi/4}\\
        1 & i & -1 & -i\\
        1 & e^{i3\pi/4} & -i & e^{i\pi/4}
    \end{pmatrix}.
\end{equation}
The row and column labels for the above three matrices runs in order of $\mathbbm1, \sigma_1,\sigma_2,\sigma_3$. Note that $\sigma_2$ is an emergent fermion.
We'll denote $\text{Vec}_{\mathbbm Z_4}^\omega$ as $\mathcal{C}$ for convenience. See $S,T$ matrices of the $6,10,14$-layer cases in \ref{ST}.

\subsection{Separable Algebras}
$\mathbb Z_4$ has three subgroups, $  \{0\}=\Z_1$, $\{0,2\}=\Z_2$ and $ \{0,1,2,3\}= \Z_4$. Using the formula \eqref{3cocycle}, one can see $\omega|_{\Z_1}=1$, $\omega|_{\Z_2}=1$. Since $\text{H}^2(\Z_1, U(1))$ and $  \text{H}^2(\Z_2, U(1))$ are trivial, we don't have any 2-cocycle to twist the $\Z_1$ and $\Z_2$ group algebra. Then we get two separable algebras $[\Z_1] = \mathbbm1$, $[\Z_2] = \mathbbm1\oplus \sigma_2$. For subgroup $\Z_4$, since $\omega|_{\Z_4} = \omega$, which is cohomology nontrivial, we don't have a separable algebra associated with $\Z_4$. 

We choose the basis vector in $\mathbbm1, \sigma_{1,2,3}$ as $\ket{\mathbbm1}, \ket{\sigma_{1,2,3}}$,  respectively, such that the multiplication rule of $[\Z_2]$ has the same form as the fusion rule, as shown in table.\ref{Alg2M}. $[\Z_1]$ is just the subalgebra of $[\Z_2]$.

\begin{table}[!h]
    \centering
    \scalebox{1}{
    \begin{tabular}{c c c}  
        & $|\mathbbm1\rangle$ & $|\sigma_2\rangle$  \\ 
  \hline
$|\mathbbm{1}\rangle$ & $|\mathbbm{1}\rangle$ & $\ket{\sigma_2}$ \\  
 \hline
 $\ket{\sigma_2}$ & $\ket{\sigma_2}$ & $|\mathbbm1\rangle$
    \end{tabular}}
    \caption{Multiplications of $[\Z_2]$ in $\mathbb Z_4$ topological order, row times column}
    \label{Alg2M}
\end{table}

\subsection{Bimodules Over Algebras}
Although the algebra multiplications are trivial, we need to deal with the problem carefully due to the nontrivial associator.
\begin{subsubsection}{$[\Z_1]$-bimodule}
    $[\Z_1]$ bimodule category $_{[\Z_1]}\mathcal{C}_{[\Z_1]}\simeq \mathcal{C}$, with monoidal structure just the tensor product of $\mathcal{C}$.
\end{subsubsection}

\begin{subsubsection}{$[\Z_1]$-$[\Z_2]$-bimodule}
\label{gamma}
    A $[\Z_1]$-$[\Z_2]$-bimodule is just a right $[\Z_2]$-module. $[\Z_2]$ has two non-isomorphic simple right modules 
    \begin{equation*}
       \gamma_+=\mathbbm 1\oplus \sigma_2, \gamma_-= \sigma_1\oplus\sigma_3
    \end{equation*}
$\gamma_+$ is just the algebra $[\Z_2]$ as its own module. $\gamma_-$ is isomorphic to the free module $\sigma_1\otimes [\Z_2] = \sigma_1\otimes  (\mathbbm 1\oplus \sigma_2)$.
 Again, remember the associator is generally nontrivial. 

 For example, the free right action of $\sigma_1\otimes [\Z_2]$ is
     \begin{equation}
     \begin{split}
    \tau_{\sigma_1\otimes[\Z_2]}: & (\sigma_1\otimes [\Z_2] )\otimes [\Z_2]\\ &
\xrightarrow{\omega}\sigma_1\otimes ([\Z_2] \otimes[\Z_2])\rightarrow \sigma_1\otimes[\Z_2],
\end{split}
    \end{equation} 
 
we'll have 
    \begin{equation}
(\ket{\sigma_1}\ket{\sigma_2})\ket{\sigma_2}\mapsto - \ket{\sigma_1}(\ket{\sigma_2}\ket{\sigma_2})\mapsto -\ket{\sigma_1}\ket{\mathbbm 1},
    \end{equation}
    since $\omega(1,2,2)=-1$. We redefine $\ket{\sigma_1}\ket{\mathbbm 1}$ as $\ket{\sigma_1}$, $\ket{\sigma_1}\ket{\sigma_2}$ as $\ket{\sigma_3}$, the action on $\gamma_-$ is then
\begin{equation}
\begin{split}
    \tau_{\gamma_-} :  \gamma_-\ot {[\Z_2]}&\rightarrow    \gamma_-\\
    \ket{\sigma_3}\ot \ket{\sigma_2}&\mapsto -\ket{\sigma_1}
\end{split}
\end{equation}
and the other actions are of fusion rule type. 

So we have $\mathcal{C}_{[\Z_2]}=\ve\oplus\ve$.  The result for $[\Z_2]$-$[\Z_1]$-bimodule is the same. Physically, $\gamma_+, \gamma_-$ are two point-like defects living on the end of $[\Z_2]$ wall.
\end{subsubsection}

\begin{subsubsection}{$[\Z_2]$-bimodule}
Using the same way of decomposing the free bimodule, and tracking the associator carefully, one can get all $[\Z_2]$ non-isomorphic indecomposable simple bimodules. we list the result below.

There are 4 indecomposable simple $[\Z_2]$-bimodules, 
\begin{equation*}
M_0^{\sigma_2}= \mathbbm1\oplus\sigma_2,
M_1^{\sigma_2}=\sigma_1\oplus \sigma_3, 
M_2^{\sigma_2} = \mathbbm1\oplus \sigma_2, 
M_3^{\sigma_2} = \sigma_1\oplus \sigma_3.
\end{equation*}
In details,
\begin{enumerate}
\item $M_0^{\sigma_2}\cong [\Z_2]$ is an $[\Z_2]$-bimodule.

\item  The left $[\Z_2]$ action on $M_1^{\sigma_2}$ is of fusion rule type, and the right action is twisted:
\begin{equation}
\begin{split}
\tau_{M_1^{\sigma_2}}: &M_1^{\sigma_2}
\otimes  [\Z_2]\rightarrow M_1^{\sigma_2}\\
& \ket{\sigma_1}\otimes\ket{\sigma_2}\mapsto -i\ket{\sigma_3}\\
&\ket{\sigma_3}\otimes \ket{\sigma_2}\mapsto -i\ket{\sigma_1}
\end{split}
\end{equation}

\item The right $[\Z_2]$ action  on $M_2^{\sigma_2}$ is of fusion rule type, and the left action is twisted:\begin{equation}
\begin{split}
     \rho_{M_2^{\sigma_2}} : &[\Z_2]\otimes M_2^{\sigma_2}\rightarrow M_2^{\sigma_2}\\
     &\ket{\sigma_2}\otimes\ket{\mathbbm1}\mapsto -\ket{\sigma_2}\\
     & \ket{\sigma_2}\otimes\ket{\sigma_2}\mapsto 
     -\ket{\mathbbm1}
\end{split}
\end{equation}

\item  The left $[\Z_2]$ action on $M_3^{\sigma_2}$ is \begin{equation}
\begin{split}
\rho_{M_3^{\sigma_2}}: & [\Z_2]\otimes M_3^{\sigma_2} \rightarrow M_3^{\sigma_2}\\
& \ket{\sigma_2}\otimes\ket{\sigma_1}\mapsto -i\ket{\sigma_3}\\
&\ket{\sigma_2}\otimes\ket{\sigma_3}\mapsto i\ket{\sigma_1}
\end{split}
\end{equation}
and right action is \begin{equation}
\begin{split}
\tau_{M_3^{\sigma_2}}: &M_3^{\sigma_2}
\otimes  [\Z_2]\rightarrow M_3^{\sigma_2}\\
& \ket{\sigma_1}\otimes\ket{\sigma_2}\mapsto \ket{\sigma_3}\\
&\ket{\sigma_3}\otimes \ket{\sigma_2}\mapsto -\ket{\sigma_1}
\end{split}
\end{equation} 
\end{enumerate}
\end{subsubsection}
These 4 simple bimodules also satisfy $\text{Vec}_{\Z_4}$ type of fusion rule\begin{equation} 
\begin{split}
&M_1^{\sigma_2}\ot[A] M_1^{\sigma_2}
 \cong M_2^{\sigma_2}\\
&M_2^{\sigma_2}\ot[A]M_2^{\sigma_2}\cong M_0^{\sigma_2};\\
& M_2^{\sigma_2}\ot[A] M_1^{\sigma_2} \cong M_3^{\sigma_2}.
\end{split}
\end{equation}
 As a fusion category, we must have $_{[\Z_2]}\mathcal{C}_{[\Z_2]}\simeq \mathcal{C}$. Physically, $M_0^{\sigma_2}, M_1^{\sigma_2}, M_2^{\sigma_2}, M_3^{\sigma_2}$ are four simple point-like defects living on the $[\Z_2]$ wall, and they have the $\mathbb Z_4$ fusion rule along the wall.

We list all the bimodule categories in table \ref{BimodZ4}.
\begin{table}
    \centering
    \scalebox{1}{
    \begin{tabular}{|c|c|c|} \hline 
 $\text{Bimod}$& $\mathbbm 1$ & $\mathbbm 1\oplus \sigma_2$  \\ \hline
$\mathbbm 1$ & $\ve_{\mathbb Z_4}^\omega$ &$\ve\oplus \ve$  \\ \hline
 $\mathbbm 1\oplus \sigma_2$&	$\ve\oplus\ve$ &	$\ve_{\mathbb Z_4}^\omega$	\\ \hline
    \end{tabular}}
    \caption{Bimodule categories, algebras in rows on the left and columns on the right.}
   \label{BimodZ4}
\end{table}

\subsection{Fusion of Domain Walls}
The only nontrivial fusion rule of domain walls in this model is \begin{equation}
    [\Z_2]\Box [\Z_2] \sim_M \mathbbm 1,
\end{equation} 
and passing through the wall will exchange $\sigma_1$ and $\sigma_3$.

The multiplication rule of $(\mathbbm 1\oplus \sigma_2)\Box(\mathbbm 1\oplus \sigma_2)$ are \eqref{multi} \begin{equation}
    \ket{ij}\otimes \ket{kl} \mapsto (-2\delta_{jk}\delta_{j\sigma_2}+1)\ket{i\ot k,j\ot l}, i,j,k,l = \mathbbm 1, \sigma_2,
\end{equation}
The invertible bimodules can be chosen the same as the case in toric code \eqref{InvBim1f}, but just replace all $f$ with $\sigma_2$. We denote this pair of invertible $[\Z_2]\Box[\Z_2]$-$\mathbbm 1$ bimodule and $\mathbbm 1$- $[\Z_2]\Box[\Z_2]$ bimodule as $N= \mathbbm 1\oplus \sigma_2$ and $N^\vee=\mathbbm 1\oplus \sigma_2$, respectively.

\subsection{Fusion of Point-like defects together with Domain Walls}
We compute the tensor product of $\gamma_{\pm}\in {}_{\mathbbm 1} \cC_{[\Z_2]}$ (see \ref{gamma}) as an example. We can use \begin{equation}
    (\gamma_+ \Box \gamma_-)\ot[{[\Z_2]\Box[\Z_2]}] N
\end{equation} as the fusion result of $\gamma_+$ and $\gamma_-$. After quotient the action of $[\Z_2]\Box [\Z_2]$, the result is the direct sum of two 1-dimensional vector spaces graded by $1, 3\in \mathbb Z_4$, so $(\gamma_+ \Box \gamma_-)\ot[(\mathbbm 1\oplus \sigma_2) \Box(\mathbbm 1\oplus \sigma_2)]  N \cong \sigma_1\oplus \sigma_3$. Similarly, one can get

\begin{widetext}
    \begin{equation}
    \begin{tikzpicture}
	\begin{pgfonlayer}{nodelayer}
		\node [style=none] (0) at (-9, 3) {};
		\node [style=none] (1) at (-9, -3) {};
		\node [style=none] (2) at (-12, 3) {};
		\node [style=none] (3) at (-12, -3) {};
		\node [style=none] (4) at (6, 0) {$\cong$};
		\node [style=none] (5) at (10.5, 3) {};
		\node [style=none] (6) at (10.5, -3) {};
		\node [style=none] (7) at (-14, -0.5) {$\gamma_+$};
		\node [style=none] (8) at (-7, -0.5) {$\gamma_+$};
		\node [style=red dot] (9) at (-12, 0) {};
		\node [style=none] (11) at (-12, -4) {$\mathbbm1\oplus \sigma_2$};
		\node [style=none] (12) at (12.5, -0.5) {$\mathbbm1\oplus\sigma_2$};
		\node [style=none] (13) at (-9, -4) {$\mathbbm1\oplus \sigma_2$};
		\node [style=circle dot] (14) at (10.5, 0) {};
		\node [style=none] (15) at (10.5, -4) {$\mathbbm1$};
		\node [style=none] (16) at (2, 3) {};
		\node [style=none] (17) at (2, -3) {};
		\node [style=none] (18) at (-1, 3) {};
		\node [style=none] (19) at (-1, -3) {};
		\node [style=none] (20) at (-5, 0) {$\cong$};
		\node [style=none] (21) at (-3, -0.5) {$\gamma_-$};
		\node [style=none] (22) at (4, -0.5) {$\gamma_-$};
		\node [style=none] (23) at (-1, -4) {$\mathbbm1\oplus \sigma_2$};
		\node [style=none] (24) at (2, -4) {$\mathbbm1\oplus \sigma_2$};
		\node [style=blue dot] (26) at (-1, 0) {};
		\node [style=red dot] (27) at (-9, 0) {};
		\node [style=blue dot] (28) at (2, 0) {};
	\end{pgfonlayer}
	\begin{pgfonlayer}{edgelayer}
		\draw [style=dashline] (2.center) to (9);
		\draw (9) to (3.center);
		\draw [style=dashline] (5.center) to (14);
		\draw [style=dashline] (14) to (6.center);
		\draw [style=dashline] (18.center) to (26);
		\draw (26) to (19.center);
		\draw (27) to (1.center);
		\draw (28) to (17.center);
		\draw [style=dashline] (0.center) to (27);
		\draw [style=dashline] (16.center) to (28);
	\end{pgfonlayer}
\end{tikzpicture}
\end{equation}

\begin{equation}
    \begin{tikzpicture}
	\begin{pgfonlayer}{nodelayer}
		\node [style=none] (0) at (3, 3) {};
		\node [style=none] (2) at (3, -3) {};
		\node [style=none] (3) at (0, 3) {};
		\node [style=none] (4) at (0, -3) {};
		\node [style=none] (5) at (18, 0) {$\cong$};
		\node [style=none] (9) at (22.5, 3) {};
		\node [style=none] (10) at (22.5, -3) {};
		\node [style=none] (12) at (-2, -0.5) {$\gamma_+$};
		\node [style=none] (13) at (5, -0.5) {$\gamma_-$};
		\node [style=red dot] (17) at (0, 0) {};
		\node [style=blue dot] (18) at (3, 0) {};
		\node [style=none] (19) at (0, -4) {$\mathbbm1\oplus \sigma_2$};
		\node [style=none] (21) at (24.5, -0.5) {$\sigma_1\oplus\sigma_3$};
		\node [style=none] (22) at (3, -4) {$\mathbbm1\oplus \sigma_2$};
		\node [style=none] (24) at (22.5, -4) {$\mathbbm1$};
		\node [style=none] (25) at (14, 3) {};
		\node [style=none] (26) at (14, -3) {};
		\node [style=none] (27) at (11, 3) {};
		\node [style=none] (28) at (11, -3) {};
		\node [style=none] (29) at (7, 0) {$\cong$};
		\node [style=none] (30) at (9, -0.5) {$\gamma_-$};
		\node [style=none] (31) at (16, -0.5) {$\gamma_+$};
		\node [style=none] (34) at (11, -4) {$\mathbbm1\oplus \sigma_2$};
		\node [style=none] (35) at (14, -4) {$\mathbbm1\oplus \sigma_2$};
		\node [style=red dot] (36) at (14, 0) {};
		\node [style=blue dot] (37) at (11, 0) {};
		\node [style=black dot] (38) at (22.5, 0) {};
	\end{pgfonlayer}
	\begin{pgfonlayer}{edgelayer}
		\draw [style=dashline] (3.center) to (17);
		\draw [style=dashline] (0.center) to (18);
		\draw (17) to (4.center);
		\draw (18) to (2.center);
		\draw [style=dashline] (27.center) to (37);
		\draw [style=dashline] (25.center) to (36);
		\draw (37) to (28.center);
		\draw (36) to (26.center);
		\draw [style=dashline] (9.center) to (38);
		\draw [style=dashline] (38) to (10.center);
	\end{pgfonlayer}
\end{tikzpicture}
\end{equation}
\end{widetext}

\section{Other Applications of Condensation Completion}
The condensation completion procedures for different kinds of categories are similar but may have quite different physical interpretations.
\subsection{Condensation Completion of a Braided Fusion Category and the Defects on the Boundary}
For a general braided fusion category $\cB$, 
\begin{enumerate}
    \item $\cB$ can be the input of Walker Wang model~\cite{Walker:2012mcd}, and it is believed that the output is a $3+1$D topological order whose codimension-2 and higher defects  form a braided fusion 2-category $Z(\Sigma \mathcal{B})$\cite{Walker:2012mcd, 
 von_Keyserlingk_2013,Wang_2017}. And \GY{if so}, it will have a canonical gapped $2+1$D boundary, the codimension-1 and higher defects on the boundary form the fusion 2-category $\Sigma \mathcal{B}$. 

    \item $\Sigma\cB$ can be viewed as a so called 1-symmetry in $2+1$D, which is in general non-invertible. If $\cB$ is a pointed braided fusion category $\mathrm{Vec}_G^\omega$, i.e. all simple objects in $\cB$ are invertible,  $\Sigma\cB$ is called a 1-form symmetry in $2+1$D. Physically, for system with the fusion 2-category $\Sigma\cB$ symmetry, although objects in $\Sigma\cB$ correspond to the codimension-1 symmetry defects or 0-(form) symmetry operators, they  are all condensation descendant of codimension-2 symmetry defects or 1-(form) symmetry operators, so $\Sigma \cB$ is essentially a 1-(form) symmetry. We use $\Sigma\cB$ instead of $\cB$ to present the symmetry because we prefer the mathematical completeness of categories of symmetry.
    
\end{enumerate}

\subsection{Condensation Completion and the category of SET orders}
\label{CatSET}
In the paper \cite{Lan2023CategoryOS}, we propose the representation principle to study the symmetry enriched topological orders (which include SPT orders and spontaneous symmetry breaking orders). We show that in $n+1$D, the SET orders with a fusion $n$-category symmetry $\cT$ and anomaly $\cX\in (n+2)\ve$ form the category $\Fun(\Sigma\cT, \cX)$. Then, the $n+1$D essentially anomaly-free $G$-SET orders for a finite group $G$ form $\Fun(\Sigma n\ve_G, (n+1)\ve)$, which is $(n+1)\Rep(G)$ by definition. Given  $F_1, F_2 \in \Fun(\Sigma n\ve_G, (n+1)\ve)$, the stacking of them is $F_1\boxtimes F_2 \circ \Sigma \Delta$, where  $\Delta: n\ve_G\rightarrow n\ve_G\boxtimes n\ve_G, g\mapsto g\boxtimes g$ is the monoidal functor that breaks the symmetry from $G\times G$ to $G$, and stacking is exactly the symmetric monoidal structure of $(n+1)\Rep (G)\simeq \Sigma n\Rep(G) $.

Let's take $n=1$ as an example.  
 Bosonic $1+1$D anomaly free gapped quantum phases with symmetry $G$ form the category $\Fun(\Sigma \ve_G, 2\ve)  = 2\Rep (G) \simeq \Sigma \Rep (G)$, which is the symmetric fusion 2-category of separable algebras, bimodules over algebras, bimodule maps in $\Rep(G)$. It's known that the  Morita class of separable algebras in $\Rep(G)$ is classified by a triple $(G,H<G, \psi\in \text{H}^2(H,U(1)))$~\cite{Ostrik:2001xnt}, which matches the physical classification of  bosonic $1+1$D anomaly free gapped quantum phases  ~\cite{Chen_2011,Schuch_2011}.  Every Morita class of algebras in $\Rep (G)$ gives an $1+1$D anomaly free gapped $G$-symmetric phases, a bimodule over algebras gives a defect between the two phases,  the tensor product of two algebras or bimodules gives the stacking of the two corresponding phases (with defects), see table.\ref{corres2}.
\onecolumngrid
    \begin{center}
        \begin{table}[h!]
    \centering
    \begin{tabular}{|c|c| c| c|}
    \hline
       &the category of $1+1$D $G$-SET orders  &$\Sigma \Rep(G)$  \\
       \hline
     object  & 1d gapped $G$ symmetric phase & separable algebra in $\Rep(G)$\\
     \hline
     1-morphism & 0d gapped domain wall between 1d phases  & bimodule over algebras in $\Rep(G)$\\
     \hline
     2-morphism & instanton &  bimodule map\\
     \hline
     tensor product & stacking of 1d phases  & tensor product of separable algebras\\
     \hline
    \end{tabular}
    \caption{$1+1$D $G$-SET orders}
    \label{corres2}
\end{table}
    \end{center}
\twocolumngrid
\begin{remark}
    The group structure of stacking of SPT phases can be covered in the following way. Let $(M,\tau: G\rightarrow \mathrm{End}(M)) \in \Rep^{\omega_2}(G)$ be a $G$ projective representation twisted by $\omega_2\in \mathrm{H}^2(G,U(1))$, then $M^*=\Hom(M,\mathbb C)$ has a induced projective representation $\tau^*$ twisted by $\omega_2^{-1}$ \begin{equation}
    (\tau^*(g)f) (-) := \omega_2^{-1}(g,g^{-1}) f(\tau(g^{-1})(-))
\end{equation}
$\forall g\in G, f\in M^*$, and
\begin{equation}
    \begin{split}
        \tau^*(g)\tau^*(h)&=\frac{\omega_2(gh,h^{-1}g^{-1})\omega_2(h^{-1},g^{-1})}{\omega_2(h,h^{-1})\omega_2(g,g^{-1})}\tau^*(gh)\\
        & = \omega_2^{-1}(g,h)\tau^*(gh).
    \end{split}
\end{equation}
 Then $(M\otimes M^*, \tau\otimes \tau^*)$ is a separable algebra in $\Rep(G)$ with evaluation and coevaluation as multiplication and comultiplication, respectively. The algebra $M\otimes M^*$ corresponds to the $(G,\omega_2)$ SPT phase, we can really use $M\otimes M^*$ to construct a commuting projector Hamiltonian with a unique ground state for periodic boundary condition to realize the $(G,\omega_2)$ SPT phase\cite{lan2023quantum, meng2024}. For another SPT phase $(G, \omega_2')$ given by $N\otimes N^*$, $(N, \sigma: G\rightarrow \End(N)) \in \Rep^{\omega_2'}(G)$, their stacking is $M\otimes M^*\otimes N\otimes N^*$. On the other hand, $(M\otimes N, \tau\otimes \sigma)$ is a projective representation twisted by $\omega_2\cdot \omega_2'$, so $M\otimes N\otimes N^*\otimes M^*$ is a separable algebra gives the $(G,\omega_2\cdot\omega_2')$ SPT. We can then prove that 
\begin{equation}
    M\otimes M^*\otimes N\otimes N^* \xrightarrow{\mathrm{id}_M\otimes c_{M^*,N\otimes N^*}} M\otimes N\otimes N^*\otimes M^*
\end{equation}
is an algebra isomorphism in $\Rep(G)$, where $c_{M^*,N\otimes N^*}$ is the braiding in $\ve$. So, the stacking of SPT phases tracks the group structure of $\mathrm{H}^2(G,U(1))$.
\end{remark}

\begin{example}
   As an example, we can compute the stacking of $1+1$D gapped phases with $S_3$ symmetry. There are $4$ Morita equivalent classes of separable algebras in $\Rep(S_3)$. We can use $\mathbbm 1, \Fun({S_3}/{\Z_3}), \Fun({S_3}/{\Z_2}),\Fun(S_3)$ as representatives, and they correspond to the SPT phase, symmetry breaking to $\Z_3$ subgroup phase , symmetry breaking to $\Z_2$ subgroup phase, and totally symmetry breaking phase, respectively. See table.\ref{S3Stacking} for the stacking results.

    Consider two 1d $G$ symmetry breaking phases with unbroken subgroup symmetry $H_1$ and $H_2$, the stable ground states form  $G$-sets $G/{H_1}$ and $G/{H_2}$, respectively.  The stacking of them gives another $G$ symmetry breaking phase, who stable ground states form the $G$-set  $G/{H_1}\times G/{H_2}$ with the diagonal $G$-action.  The fusion rule $\Fun(G/H_1)\otimes \Fun(G/H_2)$ shows how ground states in $G/{H_1}\times G/{H_2}$ are decomposed into different sectors, which are closed under symmetry actions. Mathematically, such decomposition is the orbit decomposition of $G$-sets.
\end{example}
\onecolumngrid
\begin{center}
    \begin{table}[h!]
        \centering
        \begin{tabular}{|c|c|c|c|c|}
        \hline
          $\boxtimes$ & $\mathbbm 1$ & $\Fun(S_3/{\Z_3})$   & $\Fun(S_3/{\Z_2})$ & $\Fun(S_3)$ \\
          \hline
          $\mathbbm 1$ &   $\mathbbm 1$ & $\Fun(S_3/{\Z_3})$   & $\Fun(S_3/{\Z_2})$ & $\Fun(S_3)$\\
          \hline
          $\Fun(S_3/{\Z_3})$& $\Fun(S_3/{\Z_3})$& $\Fun(S_3/{\Z_3})^{\oplus 2}$& $\Fun(S_3)$ & $\Fun(S_3)^{\oplus 2}$\\
          \hline$\Fun(S_3/{\Z_2})$&$\Fun(S_3/{\Z_2})$ & $\Fun(S_3)$ & $\Fun(S_3/{\Z_2})\oplus \Fun(S_3)$& $\Fun(S_3)^{\oplus 3}$\\
           \hline$\Fun(S_3)$&$\Fun(S_3)$& $\Fun(S_3)^{\oplus 2}$ & $\Fun(S_3)^{\oplus 3}$ & $\Fun(S_3)^{\oplus 6}$\\
          \hline
        \end{tabular}
        \caption{Stacking of $1+1$D gapped phases with $S_3$ symmetry}
        \label{S3Stacking}
    \end{table}
    \end{center}
\twocolumngrid

\begin{remark}
    For $n>1$ cases, $\Sigma n\Rep(G)$ can only cover non-chiral $n+1$D G-SET phases (including SPT phases, symmetry breaking phases), the SET phases with chiral underlying topological orders do not correspond to objects in $\Sigma n\Rep(G)$, since they have nontrivial anomaly $\cX$. Explicit formulation of $\Sigma n\Rep(G)$ for a general $n$ is still not clear,
    see~\cite{D_coppet_2023,D_coppet_20231} for some recent progress on algebras in $2\Rep (G)$.
\end{remark}

\subsection{1d Defects and Gapped Boundaries Correspondence}
\label{Def-Bdy}
If we have a 1d defect $A$ in a 2+1D topological order $\mathfrak{C}$, we can fold $\mathfrak{C}$ along $X$ and get a new 2+1D topological order $\mathfrak{C}\boxtimes \bar{\mathfrak{C}}$ with $A$ as its gapped boundary. The anyons in the stacking phase $\mathfrak{C}\boxtimes \bar{\mathfrak{C}}$ is described by the category $\mathcal{C}\boxtimes \bar{\mathcal{C}}$, where $\mathcal{C}$ is a modular tensor category and $\bar{\mathcal{C}}$ reverses the braiding.  
 
Let $A$ just be a separable algebra in $\cC$ that represents the 1d defect, remember $_A\mathcal{C}_A$ is the category of particle-like excitation living on $A$. Due to the folding picture, $_A\mathcal{C}_A$ should describe the particle-like excitation on the folding boundary.  There must be a Lagrangian algebra $L$ in $\mathcal{C}\boxtimes \bar{\mathcal{C}}$, such that  $(\mathcal{C}\boxtimes \bar{\mathcal{C}})_L \simeq\ _A\mathcal{C}_A$ as $\mathcal{C}\boxtimes \bar{\mathcal{C}}$-module categories. There is a one-to-one correspondence between the 1d defects in a topological order and the gapped boundaries of the stacking of the topological order and itself. In this section we will compute the Lagrangian algebra $L\in \cC\boxtimes \overline{\cC}$ from $A\in \cC$.

\begin{remark}
    Given a fusion category $\cC$, it is  well known that there is a bijection between Lagrangian algebras in $Z(\cC)$ and Morita classes of separable algebras in $\cC$ \cite{DavydovMügerNikshychOstrik+2013+135+177}. The folding trick above gives a physical picture of this bijection when $\cC$ is further an MTC. 
\end{remark}

According to \cite{DavydovMügerNikshychOstrik+2013+135+177}, let $\cC$ be an MTC, $\cA$ a fusion category and $F:\cC\rightarrow\cA$ a surjective central functor, such that $\dim \cC=(\dim \cA)^2$, then the right adjoint of $F^\vee(\mathbbm 1_\cA)$ is a Lagrangian algebra in $\cC$, and $\cC_{F^\vee(\mathbbm 1_\cA)}\cong \cA$, where $F^\vee$ is the right adjoint of $F$.

In our case when $\cC$ is an MTC, we need to find a central functor  $F:\cC\boxtimes \overline{\cC}\rightarrow \bmd {A}{A}\cC$. Then $X\boxtimes Y\in\mathcal{C}\boxtimes \bar{\mathcal{C}}$ acts on $M\in {}_A\mathcal{C}_A$, denoted by $(X\boxtimes Y)\odot M$, is defined by $F(X\boxtimes Y)\ot[A] M$. 
We choose the central functor $F$ defined as follows,
\begin{equation}
    \label{ModuleAction}
        F:\mathcal{C}\boxtimes \bar{\mathcal{C}}\simeq Z_1(\mathcal{C})\simeq \text{Fun}_{\mathcal{C}|\mathcal{C}}(\mathcal{C},\mathcal{C})\rightarrow \text{Fun}_\mathcal{C}(\mathcal{C}_A,\mathcal{C}_A)\simeq\ _A\mathcal{C}_A
\end{equation}

The first equivalence in~\eqref{ModuleAction} is the canonical equivalence between $\mathcal{C}\boxtimes \bar{\mathcal{C}}$ and $Z_1(\mathcal{C})$
when $\mathcal{C}$ is modular~\cite{MUGER2003159}, where $Z_1(\mathcal{C})$ is the Drinfeld center of $\mathcal{C}$. The equivalence is defined as follows.  Let $c_{X,Y}$ be the braiding natural isomorphism from $X\otimes Y$
to $Y\otimes X$. We define two fully faithful functors
\begin{equation}
    \begin{split}
        I_1:\  &\mathcal{C}\rightarrow Z_1(\mathcal{C})\\
        & X\mapsto (X,c_{-,X})
    \end{split}
\end{equation}
\begin{equation}
    \begin{split}
        I_2:\  &\bar{\mathcal{C}}\rightarrow Z_1(\mathcal{C})\\
        & X\mapsto (X,c^{-1}_{X,-})
    \end{split}
\end{equation}
Then the equivalence is 
    \begin{equation}
\begin{split}
    I:\ &\mathcal{C}\boxtimes \bar{\mathcal{C}}\rightarrow Z_1(\mathcal{C})\\
    & X\boxtimes Y \mapsto I_1(X)\ot[Z_1(\mathcal{C})] I_2(Y)\\& = (X\otimes Y, c^{-1}_{Y,-}\circ (c_{-,X}\circ \text{id}_Y))
\end{split}
\end{equation}

The second equivalence in \eqref{ModuleAction} is the equivalence between the Drinfeld center of $\mathcal{C}$ and the category of $\mathcal{C}$-bimodule functors from $\mathcal{C}$ to $\mathcal{C}$. Given $(X, \gamma_{-,X}) \in Z_1(\mathcal{C})$, $X\otimes -$ is a $\mathcal{C}$-bimodule functor. The left $\mathcal{C}$-module functor structure is induced by the half braiding $\gamma_{X}$, and the right $\mathcal{C}$-module functor structure is   induced by the associator in $\mathcal{C}$.

The third arrow just identify a $\mathcal{C}$-bimodule functor $K$  as a $\mathcal{C}$-module functor from $\mathcal{C}_A$ to $\mathcal{C}_A$. Let $N\in \mathcal{C}_A$, the right $A$ module strutre of $F(N)$ is given by the right $\mathcal{C}$-module functor structure of $K$
\begin{equation}
    K(N)\otimes A\simeq K(N\otimes A) \rightarrow K(N).
\end{equation} The left $\mathcal{C}$-module functor structure remains.

The last equivalence is given by $K\mapsto K(A)$. $K(A)$ is a right $A$ module by definition, and the left $A$ action is given by \begin{equation}
    A\otimes K(A) \simeq K(A\otimes A) \rightarrow K(A),
\end{equation}  
where the isomorphism is due to the left $\mathcal{C}$-module functor structure of $K$. Readers can go to \cite{etingof2016tensor} for more details. 

Finally, \begin{equation}
    \begin{split}
        F:\ &\mathcal{C}\boxtimes \bar{\mathcal{C}} \rightarrow \ _A\mathcal{C}_A\\
       & X\boxtimes Y \mapsto (X\otimes Y)\otimes A.
    \end{split}
\end{equation}
    As an $A$-bimodule, the right $A$ action is just the $A$ multiplication, and the left $A$ action comes from braiding with $X$ and anti-braiding with $Y$ and then do the $A$ multiplication 
  
        \begin{equation}
        \begin{split}
            A\otimes X\otimes Y \otimes A \stackrel{c_{A,X}}{\longrightarrow} X\otimes A\otimes Y \otimes A \\\stackrel{c^{-1}_{Y,A}}{\longrightarrow} X\otimes Y \otimes A\otimes A \rightarrow X\otimes Y \otimes A .
        \end{split}
    \end{equation}
    
Since $F^\vee$ is right adjoint to $F$, then we have 
\begin{equation}
_A \mathcal{C}_A(F(X\boxtimes Y), A) \simeq \mathcal{C}\boxtimes \bar{\mathcal{C}}(X\boxtimes Y, F^{\vee}(A)). 
\end{equation} 
Since $F(X\boxtimes Y) = (X\otimes Y)\otimes A\simeq ((X\otimes Y)\otimes A)\otimes_A A\simeq (X\boxtimes Y) \odot A$, and the internal hom $[A,A]$ represents the hom-functor  $_A \mathcal{C}_A(-\odot A, A)$\cite{etingof2016tensor}, then we have 
\begin{equation}
 F^{\vee}(A)\simeq [A,A].
\end{equation}
We will denote the Lagrangian algebra $F^{\vee}(A)$ as $[A,A]$ from now and then. To get the Lagrangian algebra explicitly, we just need to exhaust all simple objects $X\boxtimes Y$ in $\mathcal{C}\boxtimes \bar{\mathcal{C}}$ and compute the dimension of the hom-space $_A \mathcal{C}_A(X\otimes Y\otimes A, A)$. If the dimension is $n$, then there will be a $(X\boxtimes Y)^{\oplus n}$ direct summand in $[A,A]$. 

\begin{widetext}
For the examples of toric code, we have 
\begin{equation}
\begin{split}
[\mathbbm1,\mathbbm1]&=\mathbbm1\boxtimes \mathbbm1\oplus e\boxtimes e\oplus m\boxtimes m \oplus f\boxtimes f,\\
    [\mathbbm1\oplus e,\mathbbm1\oplus e] &= \mathbbm 1\boxtimes \mathbbm1 \oplus  \mathbbm1\boxtimes e\oplus e\boxtimes \mathbbm1 \oplus e\boxtimes e,\\
    [\mathbbm1\oplus m,\mathbbm1\oplus m]&=\mathbbm 1\boxtimes \mathbbm1 \oplus  \mathbbm1\boxtimes m\oplus m\boxtimes \mathbbm1 \oplus m\boxtimes m,\\
     [\mathbbm1\oplus f,\mathbbm1\oplus f]&=\mathbbm 1\boxtimes \mathbbm1 \oplus  e\boxtimes m\oplus m\boxtimes e \oplus f\boxtimes f,\\
      [\mathbbm1\oplus e\oplus m \oplus f,\mathbbm1\oplus e\oplus m \oplus f]&=\mathbbm 1\boxtimes \mathbbm1 \oplus  e\boxtimes\mathbbm1\oplus \mathbbm1\boxtimes m \oplus e\boxtimes m,\\
       [(\mathbbm1\oplus e\oplus m \oplus f)_\psi,(\mathbbm1\oplus e\oplus m \oplus f)_\psi]&=\mathbbm 1\boxtimes \mathbbm1 \oplus \mathbbm1\boxtimes e\oplus m\boxtimes\mathbbm1 \oplus m\boxtimes e,
\end{split}
\end{equation}    
which are six gapped boundaries of double toric code. 

For the $\TF$ topological order, we have 
    \begin{equation}
    \label{TFIH}
    \begin{split}
          [\mathbbm 1, \mathbbm 1]& = \mathbbm1\boxtimes \mathbbm1\oplus e\boxtimes e\oplus m\boxtimes m \oplus f\boxtimes f,\\
          [\mathbbm1\oplus e,\mathbbm1\oplus e] &= \mathbbm 1\boxtimes \mathbbm 1 \oplus m\boxtimes f \oplus f\boxtimes m \oplus e\boxtimes e  \\
          [\mathbbm1\oplus m,\mathbbm1\oplus m]&= \mathbbm 1\boxtimes \mathbbm 1\oplus e\boxtimes f \oplus f\boxtimes e\oplus m\boxtimes m\\
          [\mathbbm1\oplus f,\mathbbm1\oplus f]&= \mathbbm 1\boxtimes \mathbbm 1 \oplus e\boxtimes m\oplus m\boxtimes e \oplus f\boxtimes f\\
          [\mathbbm1\oplus e\oplus m \oplus f,\mathbbm1\oplus e\oplus m \oplus f]&= \mathbbm 1\boxtimes \mathbbm 1\oplus e\boxtimes f\oplus f\boxtimes m \oplus m\boxtimes e\\
          [(\mathbbm1\oplus e\oplus m \oplus f)_\psi,(\mathbbm1\oplus e\oplus m \oplus f)_\psi]&= \mathbbm 1\boxtimes \mathbbm 1\oplus e\boxtimes m \oplus m\boxtimes f\oplus f\boxtimes e
    \end{split}
\end{equation}
which are six gapped boundaries of the double $\TF$.
\end{widetext}

For the two-layer semion topological order, we have 
\begin{equation}
    \begin{split}
        [\mathbbm 1,\mathbbm 1] &= \mathbbm 1\boxtimes \mathbbm 1\oplus s_1\boxtimes s_1 \oplus s_2\boxtimes s_2 \oplus \psi \boxtimes \psi\\
    [\mathbbm 1\oplus \psi,\mathbbm 1\oplus \psi] &= \mathbbm 1\boxtimes \mathbbm 1 \oplus s_1\boxtimes s_2 \oplus s_2\boxtimes s_1 \oplus \psi\boxtimes \psi
    \end{split}
\end{equation}
which are two gapped boundaries of the double two-layer semion topological order.

For the case of $\mathbb Z_4$ chiral topological order, we have 
\begin{equation}
    \begin{split}
        [\mathbbm1, \mathbbm1]& = \mathbbm1\boxtimes \mathbbm1\oplus \sigma_1\boxtimes \sigma_3 \oplus \sigma_2\boxtimes\sigma_2\oplus\sigma_3\boxtimes \sigma_1\\
        [\mathbbm1\oplus\sigma_2, \mathbbm1\oplus\sigma_2] &= \mathbbm1\boxtimes \mathbbm1
        \oplus \sigma_1\boxtimes \sigma_1
        \oplus \sigma_2\boxtimes\sigma_2
    \oplus
    \sigma_3\boxtimes \sigma_3,
    \end{split}
\end{equation}
which are two gapped boundaries of the double $\text{Vec}_{\mathbbm Z_4}^\omega$ topological order.

\begin{remark} A separable algebra $A$ in a MTC $\cC$ may correspond to an invertible domain wall in the topological order $\cC$ (e.g., the $e$-$m$ exchange wall in toric code, corresponding to $\mathbbm1\oplus f$.) Whether $A$ corresponds to an invertible wall can be see from the internal hom $[A,A]$: According to~\cite{davydov2013structure}, $A$ corresponds to an invertible wall presicely when $[A,A]\cap (\mathbbm1\boxtimes\bar\cC)=\mathbbm1\boxtimes\mathbbm1=[A,A]\cap (\cC\boxtimes\mathbbm1)$, in other words, the only simple summands in $[A,A]$ of the form $\mathbbm1\boxtimes x$ or $x\boxtimes\mathbbm1$ is $\mathbbm1\boxtimes\mathbbm1$.
An invertible wall further determines a braided autoequivalence $\phi_A:$ $\cC\rightarrow \cC$.
One can also see how an invertible wall $A$ permutes anyons from $[A,A]$:
    \begin{equation}
        [A,A] = \oplus_{i\in \text{Irr}(\cC)}i\boxtimes\phi_A(i)^*.
    \end{equation}

For example, from the last line of \eqref{TFIH}, we can immediately see that the domain wall permutes anyons as $e\mapsto m, m\mapsto f, f\mapsto e$, the same as what we get in~\eqref{3Fpermutation}. 
\end{remark}

\section{Conclusion and Outlook}
This paper mainly discussed how to use condensation completion to find the codimension-1 gapped defects in $2+1$D topological orders and their algebraic properties. We also talked about some other applications of condensation completion in physics including defining higher symmetry, finding Lagrangian algebras in $\cC\boxtimes \overline{\cC}$ for a MTC $\cC$, and classifying phases with symmetries. Condensation completion provides us with powerful algebraic tools to study topological phases from a macroscopic point of view. 

At the same time, condensation completion related topics present many challenges and future work that need to be addressed: 
\begin{enumerate}
\item Given a fusion category 
$\cC$, the explicit model for 
$\Sigma\cC$ requires the classification theory of separable algebras in 
$\cC$, which is generally hard to do.
\item 
For a fusion $n$-category ($n\geq 2$) $\cC$, although the abstract construction of condensation completion has been developed in~\cite{gaiotto2019condensations}, we still lack a concrete, workable model of $\Sigma\cC$. Such a model will deepen our understanding of higher-dimensional topological phases, and we are actively working on this project.
\item 
With all these macroscopic descriptions, it remains to be seen whether we can realize this data from the microscopic degrees of freedom, for example, by studying local operator algebras or constructing fixed-point lattice models.
\end{enumerate}

\acknowledgments
We are grateful for the helpful discussions with Chenqi Meng, Holiverse Yang and Tian Yuan. TL is supported by start-up funding from The Chinese University of Hong Kong, and by funding from the Research Grants Council, University Grants Committee of Hong Kong (ECS No.~24304722).

\appendix{
\section{Algebras and Modules}
\label{AlgebraAndModule}
\begin{definition}[Algebra in a monoidal category]
    Let $\mathcal{C}=(\mathcal{C}, \otimes, \mathbbm{1}, \lambda, \epsilon)$ be a monoidal category. An (associative unital) algebra in $\mathcal{C}$ is a triple $(A, \mu,\eta)$, where $A\in \mathcal{C}$, $\mu: A\otimes A\to A$ and $\eta:\mathbbm{1}\to A$ such that the following diagrams 
    \begin{equation}
        \begin{tikzcd}
    	(A\otimes A)\otimes A &&&& A\otimes (A\otimes A) \\
    	\\
    	A\otimes A && A && A\otimes A
    	\arrow["\alpha_{A,A,A}", from=1-1, to=1-5]
    	\arrow["\mu\otimes \mathrm{id}_A"', from=1-1, to=3-1]
    	\arrow["\mu"', from=3-1, to=3-3]
    	\arrow["\mathrm{id}_A\otimes\mu", from=1-5, to=3-5]
    	\arrow["\mu", from=3-5, to=3-3]
        \end{tikzcd}
    \end{equation}
    \begin{equation}
        \begin{tikzcd}
    	\mathbbm{1}\otimes A && A\otimes A && A\otimes\mathbbm{1} \\
    	\\
    	&& A
    	\arrow["\eta\otimes \mathrm{id}_A", from=1-1, to=1-3]
    	\arrow["\mu", from=1-3, to=3-3]
    	\arrow["\lambda"', from=1-1, to=3-3]
    	\arrow["\epsilon", from=1-5, to=3-3]
    	\arrow["\mathrm{id}_A\otimes\eta"', from=1-5, to=1-3]
        \end{tikzcd}
    \end{equation}
    commute.
\end{definition}

\begin{definition}[Module over an algebra]
Given an algebra $(A,\mu,\eta)$. A right $A$-module is a pair $(M,\tau)$, where $M\in\mathcal{C}$ and $\tau: M\otimes A\to M$ such that the following diagrams
    \begin{equation}
        \begin{tikzcd}
    	(M\otimes A)\otimes A &&&& M\otimes (A\otimes A) \\
    	\\
    	M\otimes A && M && M\otimes A
    	\arrow["\alpha_{M,A,A}", from=1-1, to=1-5]
    	\arrow["\tau\otimes \mathrm{id}_A"', from=1-1, to=3-1]
    	\arrow["\tau"', from=3-1, to=3-3]
    	\arrow["\mathrm{id}_M\otimes\mu", from=1-5, to=3-5]
    	\arrow["\tau", from=3-5, to=3-3]
        \end{tikzcd}
    \end{equation}
    \begin{equation}
       \begin{tikzcd}
	{M\ot \mathbbm 1} \\
	{M\ot A} && M
	\arrow["{\id_M\ot \eta}"', from=1-1, to=2-1]
	\arrow["\cong", from=1-1, to=2-3]
	\arrow["\tau"', from=2-1, to=2-3]
\end{tikzcd}
    \end{equation}
    commute.
\end{definition}
\begin{remark}
    A left $A$-module $(N,\rho:A\otimes N\to N)$ is defined as the similar way with a left action.
\end{remark}

\begin{definition}[Free module]
    Given an algebra $(A,\mu,\eta)$. A free right $A$-module is a module $(X\otimes A,\tau)$, where $X\in \cC$ and $\tau=(\id_X\otimes\mu)\circ \alpha_{X,A,A}:(X\otimes A)\otimes A\xrightarrow{}X\otimes(A\otimes A)\xrightarrow{}X\otimes A$.
\end{definition}
\begin{definition}[Module map]
Given two right $A$-modules $(M_1,\tau_1)$ and $(M_2,\tau_2)$. A right module map between $M_1$ and $M_2$ is a morphism $f:M_1\to M_2$ such that the following diagram
\begin{equation}
    \begin{tikzcd}
	M_1\otimes A && M_1 \\
	\\
	M_2\otimes A && M_2
	\arrow["\tau_1", from=1-1, to=1-3]
	\arrow["f", from=1-3, to=3-3]
	\arrow["f\otimes\mathrm{id}_A"', from=1-1, to=3-1]
	\arrow["\tau_2"', from=3-1, to=3-3]
    \end{tikzcd}
\end{equation}
\end{definition}
\begin{remark}
    The right $A$-modules and $A$-module maps in $\cC$ form a category, we denote it as $\cC_A$.
\end{remark}
\begin{definition}[Algebra bimodule]
Given two algebras $(A,\mu_A,\eta_A)$ and $(B,\mu_B,\eta_B)$. A $B$-$A$-bimodule is a triple $(M, \rho, \sigma)$ where
\begin{itemize}
    \item the pair $(M, \tau: M\otimes A\to M)$ is a right $A$-module.
    \item the pair $(M, \rho: B\otimes M\to M)$ is a left $B$-module.
\end{itemize}
such that the following diagram
    \begin{equation}
    \label{CompatibleAction}
        \begin{tikzcd}
    	(B\otimes M)\otimes A &&&& B\otimes (M\otimes A) \\
    	\\
    	M\otimes A && M && B\otimes M
    	\arrow["\alpha_{B,M,A}", from=1-1, to=1-5]
    	\arrow["\rho\otimes \mathrm{id}_A"', from=1-1, to=3-1]
    	\arrow["\tau"', from=3-1, to=3-3]
    	\arrow["\mathrm{id}_B\otimes\tau", from=1-5, to=3-5]
    	\arrow["\rho", from=3-5, to=3-3]
        \end{tikzcd}
    \end{equation}
\end{definition}
commutes.
\begin{definition}[Bimodule map]
\label{BimoduleMap}
    Given two $B$-$A$-bimodules $(M_1,\tau_1,\rho_1)$, $(M_2,\tau_2,\rho_2)$, a $B$-$A$-bimodule map between $M_1$ and $M_2$ is a morphism $f:M_1\to M_2$ such that $f$ is both a left $B$-module map and a right $A$-module map.
\end{definition}
\begin{remark}
     $A$-$B$-bimodules and $A$-$B$-module maps in $\cC$ form a category, we denote it as ${}_A\cC_B$. Note that ${}_A\cC_A$ is further a monoidal category with tensor product as a relative tensor product over $A$. See definition \ref{RTP} below.
\end{remark}
\begin{definition}[Relative tensor product]
\label{RTP}
    Let $A$ be an algebra in a monoidal category $\cC$. Given a right $A$-bimodule $(M,\tau_M)$ and a left $A$-bimodule $(N,\rho_N)$. The relative tensor product between $M$ and $N$ is an object in $\cC$, denoted by $M\ot[A]N$, defined by the coequalizer
\begin{equation}
    \begin{tikzcd}
	M\otimes A \otimes N  && M\otimes N && M\ot[A]N \\
	&&&& X
	\arrow["\mathrm{id}_M\otimes \rho_N"', shift right=2, from=1-1, to=1-3]
	\arrow["\tau_M\otimes\mathrm{id}_N", shift left=2, from=1-1, to=1-3]
	\arrow["u", from=1-3, to=1-5]
	\arrow["\exists!\bar{f}",dashed, from=1-5, to=2-5]
	\arrow["\forall f"', from=1-3, to=2-5]
    \end{tikzcd}
\end{equation}
    where $u$ is the quotient map, and $u\circ \tau_M\otimes\mathrm{id}_N = u\circ\mathrm{id}_M\otimes \rho_N$, $X$ is any object in $\mathcal{C}$ and $f\circ \tau_M\otimes\mathrm{id}_N = f\circ\mathrm{id}_M\otimes \rho_N$. Here we omit the associator for simplicity, in a monoidal category with a nontrivial associator, the lower coequalized morphism should be $\mathrm{id}_M\otimes \rho_N\circ\alpha_{M,A,N}$, similar for remark.\ref{InducedModuleAction} below.
\end{definition}
\begin{remark}
\label{InducedModuleAction}
    If $B$, $C$ are algebras in $\cC$, and $M,N$ are further $B$-$A$-bimodule and $A$-$C$-bimodule respectively, then $M\ot[A]N$ has a natural $B$-$C$-bimodule structure, induced by left $B$ action on $M$ and right $C$ action on $N$,
    \begin{widetext}
        \begin{equation}
    \begin{tikzcd}
	{B\otimes M\otimes A\otimes N} && {B\otimes M\otimes N} && {B\otimes M\ot[A]N} \\
	\\
	{ M\otimes A\otimes N} && { M\otimes N} && { M\ot[A]N}
	\arrow["{\mathrm{id}_B\otimes\tau_M\otimes \mathrm{id}_N}", shift left=2, from=1-1, to=1-3]
	\arrow["{\mathrm{id}_B\otimes \mathrm{id}_M\otimes \rho_N}"', shift right=2, from=1-1, to=1-3]
	\arrow["{\mathrm{id}_B\otimes u}", from=1-3, to=1-5]
	\arrow["{\rho_M\otimes \mathrm{id}}"', from=1-1, to=3-1]
	\arrow["{\rho_M\otimes \mathrm{id}}", from=1-3, to=3-3]
	\arrow["u", from=3-3, to=3-5]
	\arrow["{\exists!\ \bar{\rho}_{M\ot[A] N}}"', dashed, from=1-5, to=3-5]
	\arrow["{\tau_M\otimes \mathrm{id}_N}", shift left=2, from=3-1, to=3-3]
	\arrow["{\mathrm{id}_M\otimes \rho_N}"', shift right=2, from=3-1, to=3-3]
\end{tikzcd}
    \end{equation}
    \end{widetext}
   here we use the property that $B\otimes -$ preserves the coequalizer.  The right $C$ action $\bar{\tau}_{M\ot[A]N}$ is defined similarly.
\end{remark}
\begin{definition}[Morita equivalence of algebras]
\label{MoritaEqAlg}
    Given two algebras $A$ and $B$ in a monoidal category $\cC$. They are Morita equivalent if there is a $A$-$B$-bimodule $M$, and a $B$-$A$-bimodule $N$, such that\begin{equation}
        \begin{split}
            M\ot[B]N&\cong A\\
            N\ot[A]M&\cong B.
        \end{split}
    \end{equation} 
as $A$-$A$-bimodule and $B$-$B$-bimodule, respectively. We call such bimodule $M,N$ as invertible bimodule.
\end{definition}
\begin{remark}
   Two algebras $A,B$ in $\cC$ are Morita equivalent if and only if their category of modules $\cC_A$, $\cC_B$ are equivalent as left $\cC$-module categories. See the definition of module category below.
   \end{remark}
\begin{definition}[Module category over a monoidal category]
   Given a monoidal category $\mathcal{C}$. A left module category over $\mathcal{C}$ is a category $\mathcal{M}$ with
   \begin{itemize}
       \item a functor $\otimes: \mathcal{C}\times\mathcal{M}\to\mathcal{M}$.
       \item two natural isomorphisms $\gamma_{X,Y,M}:(X\otimes Y)\otimes M\to X\otimes (Y\otimes M)$ and $l_M:\mathbbm1 \otimes M\to M$, for any $X,Y\in \mathcal{C}$, $M\in\mathcal{M}$ and $\mathbbm1$ is the tensor unit in $\mathcal{C}$, 
   \end{itemize}
   such that $\forall X,Y,Z\in \cC$ and $M\in \cM$, the following diagrams
   \begin{widetext}
          \begin{equation}
       \begin{tikzcd}
    	((X\otimes Y)\otimes Z)\otimes M &&&&(X\otimes (Y\otimes Z))\otimes M \\
    	\\
    	(X\otimes Y)\otimes (Z\otimes M) && X\otimes (Y\otimes (Z\otimes M)) && X\otimes ((Y\otimes Z)\otimes M)
    	\arrow["\alpha_{X,Y,Z}\otimes \id_M", from=1-1, to=1-5]
    	\arrow["\gamma_{X\otimes Y,Z,M}"', from=1-1, to=3-1]
    	\arrow["\gamma_{X,Y,Z\otimes M}"', from=3-1, to=3-3]
    	\arrow["\gamma_{X,Y\otimes Z,M}", from=1-5, to=3-5]
    	\arrow["\id_X\otimes\gamma_{Y,Z,M}", from=3-5, to=3-3]
        \end{tikzcd}
   \end{equation}
   \begin{equation}
       \begin{tikzcd}
	(X\otimes \mathbbm1)\otimes M &&&& X\otimes(\mathbbm1\otimes M) \\
	\\
	&& X\otimes M
	\arrow["\gamma_{X,\mathbbm1,M}", from=1-1, to=1-5]
	\arrow["\lambda_X\otimes \id_M"', from=1-1, to=3-3]
	\arrow["\id_X\otimes \epsilon_M", from=1-5, to=3-3]
        \end{tikzcd}
   \end{equation}
   \end{widetext}
   commute.
\end{definition}
\begin{remark}
Let $A$ be an algebra in $\cC$, then $\cC_A$ is a left $\cC$-module, with associator and unitor the same as the ones in $\cC$.
\end{remark}

\begin{definition}[Module functor]
\label{ModuleFunctor}
    Given two left $\cC$-module categories $\cM$ and $\cN$. A left $\cC$-module functor is a functor $F:\cM\to\cN$ equipped with a natural isomorphism $s_{X,M}:F(X\otimes M)\to X\otimes F(M)$ for any $X\in\cC$ and $M\in\cM$,
    
    such that $\forall X,Y\in \cC, M\in\cM$, the diagrams
    \begin{widetext}
    \begin{equation}
       \begin{tikzcd}
    	F((X\otimes Y)\otimes M) &&&& (X\otimes Y)\otimes F(M) \\
    	\\
    	F(X\otimes (Y\otimes M)) && X\otimes F(Y\otimes M) && X\otimes (Y\otimes F(M))
    	\arrow["s_{X\otimes Y,M}", from=1-1, to=1-5]
    	\arrow["F(\gamma_{X,Y,M})"', from=1-1, to=3-1]
    	\arrow["s_{X,Y\otimes M}"', from=3-1, to=3-3]
    	\arrow["\gamma_{X,Y,F(M)}", from=1-5, to=3-5]
    	\arrow["X\otimes s_{Y,M}"', from=3-3, to=3-5]
        \end{tikzcd}
   \end{equation}
   \begin{equation}
       \begin{tikzcd}
	F(\mathbbm1\otimes M) &&&& \mathbbm1\otimes F(M) \\
	\\
	&& F(M)
	\arrow["s_{\mathbbm1,M}", from=1-1, to=1-5]
	\arrow["F(l_M)"', from=1-1, to=3-3]
	\arrow["l_{F(M)}", from=1-5, to=3-3]
        \end{tikzcd}
   \end{equation}    
    \end{widetext}
    
   commute.
\end{definition}
\begin{definition}[Module natural transformation]
    Let $\cC$ be a monoidal category and $\cM$, $\cN$ be two left $\cC$-module, and $(F,s),(F',s'): \cM\rightarrow \cN$ be two left $\cC$-module functors, then a module functor natural transformation $\nu$ between $(F,s)$ and $(F',s')$ is a natural transformation satisfies an additional condition that the diagram\begin{equation}
        \begin{tikzcd}
	{F(X\otimes M)} && {F'(X\otimes M)} \\
	{X\otimes F(M)} && {X\otimes F'(M)}
	\arrow["{s_{X,M}}"', from=1-1, to=2-1]
	\arrow["{\nu_{X\otimes M}}", from=1-1, to=1-3]
	\arrow["{s'_{X,M}}", from=1-3, to=2-3]
	\arrow["{\text{id}_X\otimes \nu_M}"', from=2-1, to=2-3]
\end{tikzcd}
    \end{equation}
    commutes.
\end{definition}
\begin{remark}
    The $\cC$-module functors from $\cM$ to $\cN$ and module natural transformations form a category $\Fun_{\cC}(\cM,\cN)$.
\end{remark}
\begin{definition}[Morita equivalence of monoidal category]
Let $\cC$ and $\cD$ be two monoidal categories, they are Morita equivalent if there is a left $\cC$-module $\cM$, such that \begin{equation}
    \cD^{\text{rev}}\simeq \Fun_\cC(\cM,\cM):= \cC_{\cM}^\vee
\end{equation} 
as monoidal categories, here ``rev" means reversing the order of tensor product.
\end{definition}
\begin{remark}
    Two fusion categories $\cC$ and $\cD$ are Morita equivalent if and only if $Z(\cC)\simeq Z(\cD)$ as braided fusion categories.
\end{remark}
\begin{definition}[Separable algebra]
\label{SeparableAlg}
    An algebra $(A,\mu,\eta)$ in a monoidal category $\mathcal{C}$ is called a separable algebra if there is a $A$-$A$-bimodule map $\Delta: A\rightarrow A\otimes A$ such that $\mu\circ \Delta = \mathrm{id}_A$.
\end{definition}
\begin{remark}
The bimodule splitting $\Delta$ is not a part of the data of a separable algebra, only its existence is required. Any choice of 
    $\Delta :A \rightarrow A\ot A$ is automatically coassociative.
\end{remark}
\begin{lemma}
\label{A.1}
    If $A$ is a separable algebra in a fusion category $\mathcal{C}$, then all the simple right $A$-modules are the direct summand of a free module $i\otimes A$ for some simple object $i\in \mathcal{C}$\cite{DavydovMügerNikshychOstrik+2013+135+177}.
\end{lemma}
\begin{proof}
    Since $A$ is separable, $A\otimes A = A\oplus Y$ as $A$-$A$-bimodules. For any right $A$-modules $X$, we have \begin{equation}
    X\otimes A\simeq X\otimes_{A}(A\otimes A)\simeq X\otimes_A(A\oplus Y)\simeq X\oplus (X\otimes_A Y)
\end{equation}
Let's suppose $X$ is a simple module, and $X=i\oplus j\oplus\cdots$ as object in $\mathcal{C}$. Then  \begin{equation}
    (i\otimes A)\oplus(j\otimes A)\oplus\cdots\simeq X\oplus (X\otimes_A Y)
\end{equation}
so $X$ must be a direct summand of $i\otimes A$ for some simple $i\in\mathcal{C}$. 
\end{proof}
To find all the simple right $A$-modules, we just need to do direct sum decomposition for $i\otimes A, \forall i\in \mathcal{C}$. In the similar way we can find all simple $A$-$B$-bimodules for $A,B$ both separable algebras.

\section{Indecomposable modules over $\text{Vec}_G^\omega$}
\label{Gomega}
Let $G$ be a finite group and $\omega\in Z^3(G,U(1))$. We define $^gh = ghg^{-1}$. For a subgroup $L$ of $G$, we define $^gL=\{^gh|h\in L\}$. Let $H$ be another subgroup of $G$ and $H={}^gL$. For a $n$-cochain $\psi\in C^n(H, U(1))$, $\psi^g\in C^n(L, U(1))$ is defined by $\psi^g(g_1,\cdots,g_n)=\psi(^gg_1,\cdots,\ ^gg_n)$. Given $g\in G$, we define the 2-cochian $\Omega_g\in C^2(G, U(1))$, \begin{equation}
    \Omega_g(g_1,g_2) = \frac{\omega( ^gg_1,\ ^gg_2,g)\omega(g,g_1,g_2)}{\omega(^gg_1,g,g_2)}
\end{equation}
For a 2-cochain $\psi\in C^2(H,U(1))$, if it satisfies \begin{equation}
\label{TwistedAlgebraCondition}
    d\psi = \omega|_{H\times H\times H}=1
\end{equation}
then we denote $A(H,\psi)$ as the group algebra with multiplication twisted by $\psi$. Then every indecomposable left $\text{Vec}_G^\omega$-module is equivalent to  $(\text{Vec}_G^\omega)_{A(H,\psi)}$ for some $H$ and $\psi\in  C^2(H, U(1))$ statisfying $\eqref{TwistedAlgebraCondition}$ \cite{10.1155/S1073792803205079,etingof2016tensor}. The classification theorems of indecomposable modules over $\text{Vec}_G^\omega$ or separable algebras in $\text{Vec}_G^\omega$ are as follows (theorem 1.1 and 1.2 in \cite{Natale_2017}),

\begin{theorem}
    Let $H,L$ be subgroups of $G$ and let $\psi\in C^2(H,U(1))$ and $\xi\in C^2(L,U(1))$ statisfying \eqref{TwistedAlgebraCondition}, then $(\text{Vec}_G^\omega)_{A(H,\psi)}$ is equivalent to $(\text{Vec}_G^\omega)_{A(L,\xi)}$ as left $\text{Vec}_G^\omega$-modules if and only if there is a $g\in G$ such that $H=\ ^gL$ and $\xi^{-1}\psi^g \Omega_g|_{L\times L}$ is trivial in $\text{H}^2(L, U(1))$. 
\end{theorem}
There is also an equivalent way to state the theorem from the algebra perspective. 
\begin{definition}
    $\forall g\in G$, we define a adjoint action functor $\text{ad}_g : \text{Vec}_G^\omega\rightarrow \text{Vec}_G^\omega$, which maps $V$ to ${}^gV$, where $({}^gV)_h = V_{{}^gh}$. For $f: V\rightarrow W$, $\text{ad}(f)_h:= f_{{}^gh}$.
\end{definition}
\begin{remark}
    $\text{ad}_g$ is a monoidal antoequivalence, with monoidal structure \begin{equation}
        \begin{split}
            (\text{ad}_g^2)_{V,W}: {}^gV\otimes {}^gW&\simeq {}^g(V\ot W)\\
            v\otimes w &\mapsto \Omega_g(h,k)^{-1} v\otimes w
            \end{split}
    \end{equation}  
    where $h,k\in G$ and $v\in V_h=({}^gV)_{g^{-1}hg}$, $w\in W_k =({}^gW)_{g^{-1}kg}$. The coherent conditions of a monoidal functor follow from the fact that \begin{equation}
    d\Omega_g = \frac{\omega}{\omega_g}.
\end{equation}
\end{remark}
\begin{remark}
    $\Omega_g$ satisfies \begin{equation}
        \Omega_{g_1g_2} = \Omega_{g_1}^{g_2}\Omega_{g_2}d\gamma(g_1,g_2),
    \end{equation}
    where $\gamma(g_1,g_2): G\rightarrow U(1)$ is defined by \begin{equation}
        \gamma(g_1,g_2)(g) = \frac{\omega(g_1,g_2,g)\omega({}^{g_1g_2}g,g_1,g_2)}{\omega(g_1,{}^{g_2}g,g_2)}.
    \end{equation} 
    This gives a well defined adjoint $G$ action on $\text{Vec}_G^\omega$, which is a monoidal functor
    \begin{equation}
    \begin{split}
         \text{ad}: \underline{G} &\rightarrow \underline{\text{Aut}}_{\otimes}(\text{Vec}_G^\omega)\\
         g&\mapsto \text{ad}_g
    \end{split}
    \end{equation} 
with monoidal structure \begin{equation}
\begin{split}
    \text{ad}^2_V : {}^g({}^{g'}V)& \simeq {}^{gg'}V\\
    v&\mapsto \gamma(g,g')(h)v
\end{split}
\end{equation}
where $g,g'\in G$ and $v \in V_h$.
\end{remark}
\begin{theorem}
     Let $H,L$ be subgroups of $G$ and let $\psi\in C^2(H,U(1))$ and $\xi\in C^2(L,U(1))$ statisfying \eqref{TwistedAlgebraCondition}, then $(\text{Vec}_G^\omega)_{A(H,\psi)}$ is equivalent to $(\text{Vec}_G^\omega)_{A(L,\xi)}$ as left $\text{Vec}_G^\omega$-modules if and only if there exists $g\in G$ s.t. ${}^gA(H,\psi)\simeq A(L,\xi)$ as algebras in $\text{Vec}_G^\omega$.
\end{theorem}

\section{Relative Deligne tensor product of modules over pointed braided fusion category}
\label{RDT}
In this section, we review the result of relative Deligne tensor product of modules over pointed braided fusion category $\cC=\ve_G^{\omega,c}$. 


Let $E,F$ be two subgroup of $G$ and $\phi\in\mathrm{C}^2(E,U(1)), \psi\in \mathrm{C}^2(F,U(1))$ satisfying $d\phi = \omega|^{-1}_{E\times E\times E}$, $ d\psi = \omega|^{-1}_{F\times F\times F}$, then we can define a 2-cochain $\chi$ on $E\times F$ 
\begin{equation}
\begin{split}
    \chi((e_1,f_1),(e_2,f_2))=&\omega(e_1,f_1,e_2f_2)\omega^{-1}(f_1,e_2,f_2) c_{f_1,e_2} \omega(e_2,f_1,f_2)\\&\cdot\omega^{-1}(e_1,e_2,f_1f_2)\phi(e_1,e_2)\psi(f_1,f_2)
\end{split}
\end{equation}
$\forall e_1,e_2\in E, f_1, f_2\in F$.
Let $\pi:E\times F\rightarrow G$ be the group multiplication. One can check $\chi$ satisfies 
\begin{equation}
    d\chi = \pi^{*}\omega^{-1}.
\end{equation}
$\ker(\pi)=\{(e,f)|ef=1\}$,  the kernel of $\pi$ is a subgroup of $E\times F$. It's easy to see $\ker(\pi)$ is isomorphic to $E\cap F$, and an element $e\in E\cap F$ maps to $(e,e^{-1})\in \ker(\pi)$.  

We define a map $B(\chi): (E\times F)\times (E\cap F)\rightarrow U(1)$ by 
\begin{equation}
    B(\chi)(a,b)= \frac{\chi(a,b)}{\chi(b,a)}
\end{equation}
 $\forall a\in E\times F, b\in E\cap F$. We define the orthogonal complement of $E\cap F$ by 
\begin{equation}
    (E\cap F)^\perp=\{(e,f)\in E\times F| B(\chi)((e,f),(e',e'))=1, \forall e\in E\}.
\end{equation}
One can show that $(E\cap F)^\perp$ is actually a subgroup of $E\times F$. Let $H:=\Im(\pi)|_{(E\cap F)^{\perp}}$. Fix any set section $s: H\rightarrow (E\cap F)^\perp$, we define a 2-cochain $\rho$ on $H$:  
\begin{equation}
\begin{split}
     &\rho(g,h)\\=&\omega^{-1}(h^{-1}g^{-1},g,h)\omega(h^{-1},g^{-1},g)\chi^{-1}(s(h)^{-1}, s(g)^{-1})
\end{split}
\end{equation}
$\forall g,h\in H$.
\begin{theorem}\cite{Decoppet_2023RDT}
    \begin{equation}
     \cM(E,\phi)\boxtimes_\cC \cM(F,\psi)\simeq \cM(H,\rho)^{\oplus \alpha}
 \end{equation}
 where $H,\rho$ are the subgroup and 2-cochain constructed above, and 
 \begin{equation}
     \alpha=\frac{|(E\cap F)^{\perp}||E\cap F|}{|E||F|}
 \end{equation}
\end{theorem}
In algebra version, 
\begin{equation}
     A(E,\phi)\Box A(F,\psi)\sim_M A(H,\rho)^{\oplus \alpha}.
\end{equation}

As an example, let's use this formula to show the fusion rule 
$(\mathbbm 1 \oplus e\oplus m\oplus f) \Box (\mathbbm 1\oplus f) \sim_M \mathbbm 1 \oplus e$ in toric code. $(\mathbbm 1 \oplus e\oplus m\oplus f)=[\Z_2\times \Z_2]=[\{1,a,b,c\}]$, $\mathbbm 1\oplus f=[(\Z_2)_f]=[\{1,c\}]$. According to the braiding \eqref{TCBraiding} we choose,
the nontrivial values of $\chi$ are
\begin{equation}
\chi((g_1,f),(m,g_2))=\chi((g_1,f),(f,g_2))=-1   
\end{equation}
 for any $g_1\in \Z_2\times \Z_2, g_2\in (\Z_2)_f$. 

 $\ker(\pi)=(\Z_2\times \Z_2)\cap (\Z_2)_f = (\Z_2)_f$. One can then show that \begin{equation}
     ((\Z_2\times \Z_2)\cap (\Z_2)_f)^{\perp} = \{(1,1),(a,1), (b,c), (c,c)\},
 \end{equation}
 so $H=\Im(\pi)|_{((\Z_2\times \Z_2)\cap (\Z_2)_f)^{\perp}} = \{1,a\} =(\Z_2)_e$. We can take a section 
 \begin{equation}
 \begin{split}
        s: (\Z_2)_e&\rightarrow ((\Z_2\times \Z_2)\cap (\Z_2)_f)^{\perp}\\
            1&\mapsto (1,1)\\
           a &\mapsto (a,1)
 \end{split}
 \end{equation}
 the 2-cochain $\rho$ is totally trivial, so $[H,\rho]=[(\Z_2)_e]=\mathbbm 1\oplus e$. The multiplicity 
 \begin{equation}
     \alpha = \frac{4\times 2}{4\times 2}=1,
 \end{equation}
 so we finally have $(\mathbbm 1 \oplus e\oplus m\oplus f) \Box (\mathbbm 1\oplus f) \sim_M \mathbbm 1 \oplus e$.
 
\section{Identifying domain walls in $\TF$ with $S_3$ elements}
\label{3Fapp}
According to ~\cite{Kitaev_2012,etingof2003finite}, if there are two $2+1$D topological orders whose particles form modular tensor categories $\cC$ and $\cD$, and let $\cM$ be the fusion category of point-like defects on a domain wall of $\cC$ and $\cD$, moving bulk particles to the domain wall gives two central functors $L_\cM: \cC\rightarrow \cM$ and $R_\cM: \cD\rightarrow \cM$, if $\cM$ is an invertible wall, then $L_\cM, R_\cM$ are monoidal equivalences, and $T_\cM = R^{-1}_\cM L_\cM$ gives a braided equivalence between $\cC$ and $\cD$. We can use this to calculate what will a particle become after passing through a invertible domain wall, i.e. how does the invertible wall permutes anyons.

In our case, let $A\in \Sigma\cC$ be a domwain wall, the left bulk-to-wall functor is defined as \begin{equation}
    \begin{split}
        L_A: \cC &\rightarrow {}_A\cC_A\\
        x&\mapsto x\otimes A
    \end{split}
\end{equation} 
where the right $A$-module structure of $x\otimes A$ is just $A$ multiplication, and the left $A$-module structure is given by \begin{equation}
    A\otimes x\otimes A\xrightarrow{c_{A,x}\otimes \id_A}x\otimes A\otimes A\xrightarrow{\mu_A} x\otimes A.
\end{equation}

Similarly, the right bulk-to-wall functor is defined as \begin{equation}
    \begin{split}
        R_A: \cC &\rightarrow {}_A\cC_A\\
        x&\mapsto A\otimes x
    \end{split}
\end{equation} 
where the left $A$-module structure of $A\otimes x$ is just $A$ multiplication, and the right $A$-module structure is given by \begin{equation}
    A\otimes x\otimes A\xrightarrow{\id_A\otimes c_{x,A}}A\otimes A\otimes x\xrightarrow{\mu_A} A\otimes x,
\end{equation}
we have explained the definition of these functors in section~\ref{Def-Bdy}.

First, let's consider the domain wall $[(\Z_2)_e]$.  $L_{[\Z_2]_e}$ is a monoidal equivalence and we have \begin{equation}\begin{split}
    \mathbbm 1\otimes [(\Z_2
    )_e]&\cong M_0^e\\
      e\otimes [(\Z_2
    )_e]&\cong M_1^e\\
      m\otimes [(\Z_2
    )_e]&\cong N_1^e\\
      f\otimes[(\Z_2
    )_e]&\cong N_0^e
\end{split}
\end{equation}
as $\mathbbm 1\oplus e$-bimodule (see~\eqref{rrWall0dDefect} and below). We show $m\otimes [(\Z_2)_e]\cong  N_1^e$ and then all other isomorphic relations are self-evident. The left $[(\Z_2)_e]$ action on $ m\otimes [(\Z_2)_e]$ is 
\begin{equation}
\begin{split}
    \rho_{m\otimes [(\Z_2)_e]}: [(\Z_2)_e]\otimes (m\otimes [(\Z_2)_e])&\rightarrow m\otimes [(\Z_2)_e]\\
    \ket{e}\otimes \ket{m\mathbbm 1}&\mapsto -\ket{me}\\
    \ket{e}\otimes \ket{me}&\mapsto -\ket{m\mathbbm 1}
\end{split}
\end{equation}
since we have $c_{e,m}=-1$ (see~\eqref{3FBraiding}). The one can check \begin{equation}
\begin{split}
       \phi_{L,m}: m\otimes[(\Z_2)_e]&\rightarrow N_1^e\\
       \ket{m\mathbbm 1}&\mapsto \ket{m}\\
       \ket{me}&\mapsto \ket{f}
\end{split}
\end{equation}
is a bimodule isomorphism. Now we need to find whose images are $L_{[(\Z_2)_e]}(i), i=\mathbbm 1,e,m,f$ under $\Ree$. We claim \begin{equation}
    \begin{split}
      \Le(\mathbbm 1)&\cong \Ree(\mathbbm 1)\\
      \Le(e)&\cong \Ree(e)\\
      \Le(m)&\cong\Ree(f)\\
      \Le(f)&\cong\Ree(m)
    \end{split}
\end{equation}
as $[(\Z_2)_e]$-bimodule. Again, we show $\Le(m)\cong\Ree(f)$ and others then automatically hold. $\Ree(f)  = [(\Z_2)_e]\otimes f$, and its right $[(\Z_2)_e]$ action is given by \begin{equation}
\begin{split}
      \tau_{[(\Z_2)_e]\otimes f}: ([(\Z_2)_e]\otimes f)\otimes [(\Z_2)_e]&\rightarrow [(\Z_2)_e]\otimes f\\
      \ket{\mathbbm 1f}\otimes \ket{e}&\mapsto -\ket{ef}\\
      \ket{ef}\otimes \ket{e}&\mapsto
      -\ket{\mathbbm 1f},
\end{split}
\end{equation}
and then one can check \begin{equation}
\begin{split}
    \phi_{mf}: m\otimes[(\Z_2)_e]&\rightarrow [(\Z_2)_e]\otimes f\\
    \ket{m\mathbbm 1}&\mapsto i\ket{ef}\\
    \ket{m e}&\mapsto -i\ket{\mathbbm 1f}
\end{split}
\end{equation}
is a $[(\Z_2)_e]$-bimodule isomorphism. We can conclude that passing through $[(\Z_2)_e]$ wall will exchange $m$ and $f$. The calculation for $[(\Z_2)_m]$ and $[(\Z_2)_f]$ are almost the same, $[(\Z_2)_m]$ wall will exchange $e$ and $f$, $[(\Z_2)_f]$ wall will exchange $e$ and $m$.

Then we consider the wall $[\Z_2\times \Z_2]= \mathbbm 1\oplus e\oplus m\oplus f$. 
$L_{[\Z_2\times \Z_2]}, R_{[\Z_2\times \Z_2]}$ are monoidal equivalences and we have \begin{equation}
    \begin{split}
        \mathbbm 1\otimes [\Z_2\times \Z_2] &\cong M_0\\
        e\otimes  [\Z_2\times \Z_2]&\cong M_1\\
        m\otimes  [\Z_2\times \Z_2]&\cong M_3\\
        f\otimes  [\Z_2\times \Z_2]&\cong M_2 \\
    \end{split}
\end{equation}
\begin{equation}
    \begin{split}
          [\Z_2\times \Z_2] \otimes \mathbbm 1 &\cong M_0\\
         [\Z_2\times \Z_2]\otimes e&\cong M_3\\
          [\Z_2\times \Z_2] \otimes m&\cong M_2\\
          [\Z_2\times \Z_2] \otimes f&\cong M_1 \\
    \end{split}
\end{equation}
as $[\Z_2\times \Z_2]$-bimodule (see~\eqref{Z2Z2Bimod} and below for the action of $M_{0,1,2,3}$), or more explicitly 
\begin{equation}
\begin{split}
    \LV(\mathbbm 1)&\cong \RV(\mathbbm 1)\\
    \LV(e)&\cong\RV(f)\\
    \LV(m)&\cong\RV(e)\\
    \LV(f)&\cong\RV(m)\\
\end{split}
\end{equation}
Physically this means that passing through $[\Z_2\times \Z_2]$ wall from left to right will permute $e,m,f$ as \begin{equation}
    \begin{tikzcd}
	e && m \\
	& f
	\arrow[from=1-1, to=2-2]
	\arrow[from=2-2, to=1-3]
	\arrow[from=1-3, to=1-1]
\end{tikzcd}.
\end{equation}

For the last one, $[\Z_2\times \Z_2,\psi]
$ wall, passing which from left to right must permute $e,m,f$ as 
\begin{equation}
    \begin{tikzcd}
	e && m \\
	& f
	\arrow[from=1-1, to=1-3]
	\arrow[from=1-3, to=2-2]
	\arrow[from=2-2, to=1-1]
\end{tikzcd}.
\end{equation}

\section{$S$,$T$ matrices for $\mathbb Z_4$ topological order with $c=3,5,7$ mod $8$}
\label{ST}
\paragraph{$c=3$ mod $8$}
\begin{equation}
        S=\begin{pmatrix}
          1&1&1&1\\
          1&i&-1&-i\\
          1&-1&1&-1\\
          1&-i&-1&i
        \end{pmatrix}
        \end{equation}
        \begin{equation}
    T=\begin{pmatrix}
            1&&&\\
            &e^{\frac{3\pi i}{4}}&&\\
            &&-1&\\
            &&&e^{\frac{3\pi i}{4}}
        \end{pmatrix} 
        \end{equation}
\paragraph{$c=5$ mod $8$}
\begin{equation}
        S=\begin{pmatrix}
          1&1&1&1\\
          1&-i&-1&i\\
          1&-1&1&-1\\
          1&i&-1&-i
        \end{pmatrix}
        \end{equation}
        \begin{equation}
        T=\begin{pmatrix}
            1&&&\\
            &e^{\frac{5\pi i}{4}}&&\\
            &&-1&\\
            &&&e^{\frac{5\pi i}{4}}
        \end{pmatrix}
\end{equation}
\paragraph{$c=7$ mod $8$}
\begin{equation}
S=\begin{pmatrix}
          1&1&1&1\\
          1&i&-1&-i\\
          1&-1&1&-1\\
          1&-i&-1&i
        \end{pmatrix}
        \end{equation}
        \begin{equation}
        T=\begin{pmatrix}
            1&&&\\
            &e^{\frac{7\pi i}{4}}&&\\
            &&-1&\\
            &&&e^{\frac{7\pi i}{4}}
        \end{pmatrix}
\end{equation}
these three cases correspond to $6,10,14$ mod $16$ layers of $p+ip$ superconductors.
}

\bibliography{bib}

\end{document}